\documentclass[a4paper,UKenglish]{lipics}

\usepackage[pgf,svgnames,dvipsnames]{xcolor}
\usepackage{manfnt}

\usepackage{setspace}
\usepackage{amsmath}
\usepackage{amssymb}
\usepackage{xspace}
\usepackage{amssymb}
\usepackage{latexsym}
\usepackage{verbatim}
\usepackage{paralist}

\usepackage{bibunits}
\defaultbibliography{biblio}

\usepackage{tikz}
\usetikzlibrary{arrows}
\tikzstyle{block}=[draw opacity=0.7,line width=1.4cm]

\newcommand{\details}[1]{}

\newcommand{\QPTL}{\text{\sffamily QPTL}}

\newcommand{\CTLStar}{\text{\sffamily CTL$^{*}$}}
\newcommand{\MCTLStar}{\text{\sffamily CTL$^{*}_{lp}$}}
\newcommand{\LTL}{\text{\sffamily LTL}}
\newcommand{\PLTL}{\text{\sffamily PLTL}}
\newcommand{\NWA}{\text{\sffamily NWA}}
\newcommand{\SNWA}{\text{\sffamily SNWA}}
\newcommand{\HCTLStar}{\text{\sffamily HyperCTL$^{*}$}}
\newcommand{\MHCTLStar}{\text{\sffamily HyperCTL$^{*}_{lp}$}}
\newcommand{\KCTLStarS}{\text{\sffamily KCTL$^{*}$}}
\newcommand{\KLTL}{\text{\sffamily KLTL}}
\newcommand{\HLTL}{\text{\sffamily HyperLTL}}
\newcommand{\KLTLS}{\text{\sffamily KLTL}}

\newcommand{\HAA}{\text{\sffamily HAA}}

\newcommand{\Know}{\textsf{K}}
\newcommand{\GExists}{\ensuremath{\exists^{G}}}
\newcommand{\GForall}{\ensuremath{\forall^{G}}}
\newcommand{\Until}{\textsf{U}}
\newcommand{\PUntil}{\textsf{U}^{-}}
\newcommand{\Release}{\textsf{R}}
\newcommand{\PRelease}{\textsf{R}^{-}}
\newcommand{\Next}{\textsf{X}}
\newcommand{\PNext}{\textsf{X}^{-}}
\newcommand{\Always}{\textsf{G}}
\newcommand{\PAlways}{\textsf{G}^{-}}
\newcommand{\Eventually}{\textsf{F}}
\newcommand{\PEventually}{\textsf{F}^{-}}

\newcommand{\Agents}{\textsf{Agts}}
\newcommand{\agent}{\ensuremath{a}}
\newcommand{\AP}{\textsf{AP}}

\newcommand{\Var}{\textsf{VAR}}
\newcommand{\Obs}{\textit{Obs}}
\newcommand{\Obsequiva}{\Obs_\agent}
\newcommand{\sad}{\textit{sad}}
\newcommand{\dir}{\textit{dir}}
\newcommand{\Acc}{\textit{Acc}}

\newcommand{\finite}{\textit{fin}}

\newcommand{\Inf}{\textit{Inf}}
\newcommand{\Sub}{\textit{Sub}}
\newcommand{\IN}{\textit{init}}
\newcommand{\NoIN}{\textit{no-init}}
\newcommand{\pTag}{\textit{tag}}

\newcommand{\Lang}{\mathcal{L}}
\newcommand{\PosBool}{\ensuremath{\mathcal{B}^{+}}}
\newcommand{\FamStratum}{\ensuremath{\mathsf{F}}}
\newcommand{\Stratum}{\ensuremath{\mathcal{S}}}
\newcommand{\Region}{\ensuremath{\mathcal{R}}}
\newcommand{\Regions}{\ensuremath{\textit{REG}}}
\newcommand{\Bu}{\ensuremath{\textsf{B}}}
\newcommand{\Co}{\ensuremath{\textsf{C}}}
\newcommand{\Tower}{\mathsf{Tower}}
\newcommand{\Bl}{\textit{bl\,}}
\newcommand{\Nodes}{\textit{Nodes}}
\newcommand{\SUCC}{\textit{succ}}
\newcommand{\join}{\textit{join}}
\newcommand{\EmphInt}{\mathfrak{m}}
\newcommand{\Alert}{\textit{alert}}
\newcommand{\Trans}{\texttt{t}}
\newcommand{\dual}[1]{\ensuremath{\widetilde{#1}}} 

\def\Op{{\cal O}}
\def\A{{\cal A}}

\def\N{{\mathbb{N}}}

\newcommand \tpl[1]{\langle #1 \rangle}

\newcommand{\true}{\texttt{true}}
\newcommand{\false}{\texttt{false}}

\def\PSPACE{{\sffamily PSPACE}}
\def\EXPSPACE{{\sffamily EXPSPACE}}
\def\NLOGSPACE{{\sffamily NLOGSPACE}}

\newif\ifdraft\drafttrue

\ifdraft
\newcommand\modedraft[1]{#1}

\newcommand\todo[1]{{\color{purple}
[\textbf{To do:} #1]}}
\newcommand\lb[1]{\marginpar[{\color{cyan}\small\dbend}]{\color{cyan}\small\dbend}{\footnotesize \color{cyan}[#1 - \textbf{Laura}]}}
\newcommand\bm[1]{\marginpar[{\color{blue}\small\dbend}]{\color{blue}\small\dbend}{\footnotesize \color{blue}[#1 - \textbf{Bastien}]}}
\renewcommand\sp[1]{\marginpar[{\color{orange}}]{\color{orange}\small\dbend}{\footnotesize \color{orange}[#1 - \textbf{Sophie}]}}

\else
\newcommand\modedraft[1]{}
\newcommand\todo[1]{}
\newcommand\lb[1]{}
\newcommand\bm[1]{}
\renewcommand\sp[1]{}

\fi

\title{Unifying Hyper  and Epistemic Temporal  Logics}

\author[1]{Laura Bozzelli}
\affil[1]{UPM,
  Madrid, Spain -- \texttt{laura.bozzelli@fi.upm.es}}

\author[2]{Bastien Maubert}
\affil[2]{IRISA, Universit\'e de Rennes 1, France --
  \texttt{bastien.maubert@irisa.fr}}

\author[3]{Sophie Pinchinat}
\affil[3]{IRISA, Universit\'e de Rennes 1,  France --
  \texttt{sophie.pinchinat@irisa.fr}}

 \Copyright{\hspace{1pt}}

\begin{document}

\maketitle

\begin{abstract}
  In the literature, two powerful temporal logic formalisms have been
  proposed for expressing information flow security requirements, that
  in general, go beyond regular properties. One is classic, based on
  the knowledge modalities of epistemic logic. The other one, the so
  called hyper logic, is more recent and subsumes many proposals from
  the literature; it is based on explicit and simultaneous
  quantification over multiple paths. In an attempt to better
  understand how these logics compare with each other, we consider the
  logic KCTL* (the extension of CTL* with knowledge modalities and synchronous perfect recall
  semantics) and HyperCTL*. We first establish that KCTL* and
  HyperCTL* are expressively incomparable. Second, we introduce and
  study a natural linear past extension of HyperCTL* to
  unify KCTL* and HyperCTL*; indeed, we show that KCTL* can be easily
  translated in linear time into  the proposed logic. Moreover, we show
  that the 
  model-checking problem for this novel logic is
  decidable, and we provide its exact computational complexity in
  terms of a 
new measure of  path
  quantifiers' alternation. For this, we settle open complexity issues for
  unrestricted quantified propositional temporal logic.
\end{abstract}

\begin{bibunit}

\section{Introduction}
\label{sec-intro}

Temporal logics provide a fundamental framework for the description of the dynamic behavior
of reactive systems. Additionally, they support the successful model-checking technology
that allow complex finite-state systems to be verified automatically.

Classic \emph{regular} temporal logics
such as standard $\LTL$~\cite{Pnueli77} or the more expressive 
$\CTLStar$~\cite{EmersonH86} lack mechanisms
to relate distinct paths or executions of a system.
Therefore, they cannot
express information-flow security properties which specify how information
may propagate from inputs to outputs, 
 such as non-interference \cite{goguen1982security} or
  opacity \cite{DBLP:journals/ijisec/BryansKMR08}.

 In the literature, two powerful temporal logic formalisms have been
  proposed for expressing 
such  security requirements that,
  in general, go beyond regular properties.

  One is classical and is based on the extension of temporal logic
  with the knowledge modalities of epistemic
  logic~\cite{fagin1995reasoning}, which allow to relate 
  paths that are observationally equivalent for a
  given agent. We consider  $\KCTLStarS$, the extension
  of 
  $\CTLStar$ with  knowledge modalities under the
  synchronous perfect recall semantics (where an agent remembers the
  whole sequence of its 
  observations, and the observations are
  time-sensitive)~\cite{halpern2004complete,MeydenS99,ShilovG02,Dima08}.
  This logic and its linear-time fragment, $\KLTL$, can be
  used to specify secrecy 
  policies 
  \cite{alur2007model,HalpernO08,BalliuDG11}.

  The second framework is more recent~\cite{ClarksonS10} and allows to
  express properties of sets of execution traces, known as
  \emph{hyperproperties}, useful to formalize security policies, such
  as noninterference and observational determinism. The general hyper
  logical framework introduced in~\cite{ClarksonS10} is based on a
  second-order logic for which 
  model-checking 
  is undecidable.  More recently, fragments of this logic have been
  introduced~\cite{ClarksonFKMRS14}, namely the logics $\HCTLStar$ and
  $\HLTL$, which extend $\CTLStar$ and $\LTL$
  by allowing explicit and simultaneous quantification over multiple
  paths.  $\HCTLStar$ represents a simple and natural
  non-regular extension of $\CTLStar$ which admits a decidable
  model-checking problem and in which important information-flow
  security policies can be expressed. $\HCTLStar$ also generalizes a
  related temporal logic introduced in~\cite{DimitrovaFKRS12}.  Other
  logics for hyperproperties have been introduced
  in~\cite{MilushevC12}, but as pointed in~\cite{ClarksonFKMRS14}, no
  general approach to verifying such logics exists.

\noindent \textbf{Contribution.} Our first goal in this paper is to compare the expressive power of hyper temporal logics
and epistemic temporal logics. 
We establish by formal non-trivial arguments that $\HCTLStar$ and $\KCTLStarS$ are expressively incomparable.
More precisely, we prove that
$\HLTL$ (resp., $\KLTL$) cannot be expressed in $\KCTLStarS$ (resp., $\HCTLStar$),
even with respect to the restricted class of finite-state
Kripke structures. The main intuitions are as follows. On the one hand, 
differently from
$\HCTLStar$, $\KCTLStarS$ cannot express a linear-time requirement, simultaneously, on multiple paths. 
On the other hand, unlike $\KCTLStarS$, $\HCTLStar$ cannot express requirements
which relate at some timestamp an unbounded number of paths.

As a second contribution, we introduce and investigate a natural 
 linear past extension
of $\HCTLStar$, denoted by $\MHCTLStar$, to unify $\HCTLStar$ and
$\KCTLStarS$. This extension is strictly more expressive than $\HCTLStar$; indeed, we show
that $\KCTLStarS$ can be easily translated in linear time into $\MHCTLStar$.
Like $\HCTLStar$ and  $\KCTLStarS$, the finite-state model-checking problem for the novel logic is non-elementarily
decidable, and we provide the exact complexity of this problem in terms of a variant
of the standard alternation depth of path quantifiers. For this, we settle complexity issues
for satisfiability of unrestricted Quantified Propositional Temporal Logic ($\QPTL$)~\cite{SistlaVW87}. The optimal upper bounds
for full $\QPTL$ are obtained by a sophisticated generalization of the standard automata-theoretic approach
for $\QPTL$ in prenex normal form~\cite{SistlaVW87}, which  exploits  a subclass of parity two-way alternating word automata.
Our results also solve complexity issues for $\HCTLStar$ left open in~\cite{ClarksonFKMRS14}.
Due to lack of space some proofs are omitted and can be found in the
Appendix.

\begin{remark} In~\cite{ClarksonFKMRS14}, an extension of the semantics of $\HCTLStar$ is also considered. In this setting,
the path quantification can simulate  quantification over
propositional variables, 
and within this generalized semantics,
$\KLTL$ can be effectively expressed in  $\HCTLStar$~\cite{ClarksonFKMRS14}.
\end{remark}

\section{Preliminaries}
\label{sec-prelim}

For all $i,j\in\N$, let $[i,j] := \{h\in\N \mid i\leq h \leq j\}$.
Fix a \emph{finite} set
$\AP$ of atomic propositions. A \emph{trace} is a finite or infinite
word over $2^{\AP}$.  For a word  $w$ over some 
alphabet,  $|w|$ is the length of $w$ ($|w|=\infty$ if $w$ is infinite), and for each $0\leq i<|w|$,
$w(i)$ is the $i^{th}$ symbol of $w$.\vspace{0.2cm}

\noindent \textbf{Structures and tree structures.}  A \emph{Kripke
  structure $($over $\AP$$)$} is a tuple $K=\tpl{S,s_0,E,V}$, where
$S$ is a set of states, $s_0\in S$ is the initial state, $E\subseteq
S\times S$ is a transition relation 
such that for each $s\in S$, $(s,t)\in E$ for some $t\in S$,
and $V:S \rightarrow 2^{\AP}$ is an \emph{$\AP$-valuation} assigning
to each state $s$ the set of propositions in $\AP$ which hold at
$s$. The mapping $V$ can be extended to words over $S$ in the obvious
way. A path $\pi= t_0,t_1,\ldots$ of $K$ is an infinite word over 
$S$ such that for all $i\geq 0$, $(t_{i},t_{i+1})\in
E$. For each $i\geq 0$, $\pi[0,i]$ denotes the prefix of $\pi$ leading
to the $i^{th}$ state, and $\pi[i,\infty]$ the suffix of $\pi$ from
the $i^{th}$ state. A finite path of $K$ is a prefix of some path of
$K$. An \emph{initial path} of $K$ is a path starting from the initial
state. We say that $K=\tpl{S,s_0,E,V}$ is a \emph{tree structure} if
$S$ is a prefix-closed subset of $\N^{*}$, $s_0=\varepsilon$ (the root
of $K$), and $(\tau,\tau')\in E$ $\Rightarrow$ $\tau'=\tau\cdot i$ for some
$i\in\N$. States of a tree structure are also called \emph{nodes}. For
a Kripke structure $K$, $\textit{Unw}(K)$ 
is the tree unwinding of $K$ from the initial state.
A \emph{tree structure} is \emph{regular} if it is the
unwinding of some finite Kripke structure.


\subsection{Temporal Logics with knowledge modalities}

We recall the \emph{non-regular} extensions,  denoted by $\KCTLStarS$ and $\KLTLS$, of standard
$\CTLStar$ and $\LTL$  obtained by adding  the knowledge modalities of
epistemic logic under the \emph{synchronous} perfect recall semantics \cite{halpern2004complete,MeydenS99,ShilovG02,Dima08}.
Differently from the asynchronous setting,
 the synchronous setting
can be considered time sensitive in the sense that it can model an
observer who knows that a transition has occurred even if the
observation has not changed. 

For a finite set $\Agents$ of agents, formulas $\varphi$ of
$\KCTLStarS$  over $\Agents$ and $\AP$  are defined as:
\[
\varphi ::= \top \ | \ p \ | \ \neg \varphi \ | \ \varphi \vee \varphi \ | \ \Next \varphi\ | \ \varphi \Until \varphi \ | \ \exists  \varphi \ | \ \Know_\agent  \varphi
\]
where $p\in \AP$, $\agent\in \Agents$, $\Next$ and $\Until$ are the
``next'' and ``until'' temporal modalities, $\exists$ is the
$\CTLStar$ existential path quantifier, and $\Know_\agent$ is the
knowledge modality for agent $\agent$.
 We also use standard shorthands: $\forall\varphi:=\neg\exists\neg\varphi$ (``universal path quantifier''),
$\Eventually\varphi:= \top \Until \varphi$ (``eventually'') and its
dual $\Always \varphi:=\neg \Eventually\neg\varphi$ (``always'').
 A formula $\varphi$ is a
\emph{sentence} if each temporal/knowledge modality is in the scope of
a path quantifier.  The logic $\KLTLS$ is the $\LTL$-like fragment of $\KCTLStarS$
consisting of sentences of the form $\forall\varphi$, where $\varphi$
does not contain path quantifiers.

The logic $\KCTLStarS$ is interpreted over \emph{extended} Kripke structures $(K,\Obs)$, i.e., Kripke structures $K$  equipped
with an \emph{observation map} $\Obs: \Agents \rightarrow 2^{\AP}$
associating to each agent $\agent\in\Agents$, the set $\Obs(\agent)$
of propositions which are observable by agent $\agent$. For an agent
$\agent$ and a finite trace $w\in (2^{\AP})^{*}$, the
$\agent$-observable part $\Obs_\agent(w)$ of $w$ is the finite trace
of length $|w|$ such that for all $0\leq i<|w|$,
$\Obs_\agent(w)(i)=w(i)\cap \Obs(\agent)$.  Two finite traces
$w$ and $w'$ are \emph{(synchronously)
  $\Obs_\agent$-equivalent} if $\Obs_\agent(w) =\Obs_\agent(w')$ (note
that  $|w|=|w'|$). Intuitively, an agent $\agent$ does
not distinguish  prefixes of paths
 whose traces are $\Obs_\agent$-equivalent.

Given a $\KCTLStarS$ formula $\varphi$, an extended Kripke
structure $\Lambda=(K,\Obs)$, an \emph{initial} path $\pi$ of
$K$, and a position $i$ along $\pi$, the
satisfaction relation $\pi,i \models_{\Lambda} \varphi$ for
$\KCTLStarS$ is inductively defined as follows (we omit the clauses
for the Boolean connectives which are standard):
%
\[ \begin{array}{ll}
  \pi, i \models_{\Lambda} p \quad &   \Leftrightarrow \quad p \in V(\pi(i))\\
  \pi, i \models_{\Lambda} \Next \varphi \quad & \Leftrightarrow \quad  \pi, i+1 \models_{\Lambda} \varphi \\
  \pi, i \models_{\Lambda} \varphi_1\Until \varphi_2 \quad &
  \Leftrightarrow \quad \text{for some $j\geq i$}: \pi, j
  \models_{\Lambda} \varphi_2
  \text{ and } \pi, k \models_{\Lambda}  \varphi_1 \text{ for all }k\in [i,j-1] \\
  \pi, i \models_{\Lambda} \exists \varphi \quad & \Leftrightarrow \quad  \text{for some \emph{initial} path $\pi'$ of $K$ such that }\pi'[0,i]=\pi[0,i],\,
  \pi', i \models_{\Lambda} \varphi\\
  \pi, i \models_{\Lambda} \Know_\agent \varphi \quad & \Leftrightarrow \quad  \text{for all \emph{initial} paths $\pi'$ of $K$  such that }\\
  & \phantom{ \Leftrightarrow \quad \,\,} \text{$V(\pi[0,i])$ and
    $V(\pi'[0,i])$ are $\Obs_\agent$-equivalent, $\pi', i
    \models_{\Lambda} \varphi$}
\end{array} \]
%
$(K,\Obs)$ \emph{satisfies} $\varphi$, written
$(K,\Obs)\models \varphi$, if there is an initial path $\pi$ of $K$  such that $\pi, 0 \models_{(K,\Obs)} \varphi$.  Note
that if $\varphi$ is a \emph{sentence}, then the satisfaction relation
$\pi, 0 \models_{(K,\Obs)} \varphi$ is independent of $\pi$.  
One can easily show that  $\KCTLStarS$ is bisimulation invariant  and satisfies the tree-model property. In particular,
$(K,\Obs)\models \varphi$ iff $(\textit{Unw}(K),\Obs)\models \varphi$.


\begin{example}\label{example:witnessFormulaKnowledge} Let us consider the  $\KLTLS$ sentence $ \varphi_p:= \forall\Next\Eventually\Know_\agent\,\neg  p$.\\
  For all observation maps $\Obs$ such that $ \Obs(\agent)=\emptyset$,
$(K,\Obs)\models\varphi_p$ means that there is some non-root level in
the unwinding of $K$ at which \emph{no} node satisfies $p$.
  This requirement represents a
  well-known non-regular context-free branching temporal property (see
  e.g. \cite{AlurCZ06}).
\end{example}

\subsection{Hyper Logics}\label{SubSec:HyperLogics}

In this subsection, first, we recall the hyper logics $\HCTLStar$ and
$\HLTL$~\cite{ClarksonFKMRS14} which are non-regular extensions of
standard $\CTLStar$ and $\LTL$, respectively, with a restricted form
of explicit first-order quantification over paths of a Kripke
structure.  Intuitively, path variables are used
to express a linear-temporal requirement,  simultaneously, on multiple paths.
Then, we
introduce a linear-time past extension of 
$\HCTLStar$, denoted by
$\MHCTLStar$.  In this novel logic, 
path quantification is `memoryful',
i.e., it ranges over paths that start at the root of the computation
tree (the unwinding of the Kripke structure) and either visit the
current node $\tau$ (\emph{regular} path quantification), or visit a
node $\tau'$ at the same level as $\tau$ (\emph{non-regular} path
quantification).
\vspace{0.2cm}

\noindent \textbf{The logic $\HCTLStar$~\cite{ClarksonFKMRS14}.}
For  a finite set $\Var$  of \emph{path variables},
the syntax of $\HCTLStar$ formulas $\varphi$ over $\AP$ and $\Var$  is defined as follows:
\[\varphi ::= \top \ | \ p[x] \ | \ \neg \varphi \ | \ \varphi \wedge
\varphi \ | \ \Next \varphi\ | \ \varphi \Until \varphi \ | \ \exists
x. \varphi
\]
where $p\in \AP$, $x\in \Var$,  and $\exists x$ is the
\emph{hyper} existential path quantifier for variable $x$.  Informally, formula
$\exists x. \varphi$ requires that there is an initial path $\pi$  such that $\varphi$ holds with $x$ mapped to $\pi$, and the atomic formula $p[x]$ assert that $p$ holds at the current position of the
path bound by $x$.
The hyper universal path quantifier $\forall x$ is defined as: $\forall x.\varphi := \neg \exists
x. \neg\varphi$.   A
$\HCTLStar$ formula $\varphi$ is a \emph{sentence} if
each temporal modality occurs in the scope of a path quantifier and for each atomic formula
$p[x]$, $x$ is bound by a path quantifier.
The logic $\HLTL$ is the fragment of  $\HCTLStar$ consisting of formulas in prenex
form, i.e., of the form $Q_1 x_1. \ldots .Q_n x_n. \varphi$, where
$Q_1,\ldots,Q_n\in\{\exists,\forall\}$ and $\varphi$ does not contain
path quantifiers.

We give a semantics for $\HCTLStar$ which is equivalent to that given
in~\cite{ClarksonFKMRS14}, but is more suitable for a
linear-past memoryful generalization.
$\HCTLStar$ formulas $\varphi$ are interpreted over Kripke structures
$K=\tpl{S,s_0,E,V}$ equipped with a \emph{path assignment} $\Pi: \Var \rightarrow S^{\omega}$ associating to each variable $x\in
\Var$ an \emph{initial path} of $K$, a
variable $y\in \Var$ ($\Pi(y)$ represents the current path), and a
position $i\geq 0$ (denoting
the current position along the paths in
$\Pi$).  The satisfaction relation $\Pi,y,i \models_{K} \varphi$ is
 defined as follows (we omit the clauses for the Boolean
connectives which are standard):
%
\[ \begin{array}{ll}
\Pi, y, i \models_K p[x] \quad &   \Leftrightarrow \quad p \in
V(\Pi(x)(i))\\
\Pi, y, i \models_K \Next \varphi \quad & \Leftrightarrow  \quad \Pi, y, i+1 \models_K \varphi \\
\Pi, y,i \models_K  \varphi_1\Until \varphi_2 \quad & \Leftrightarrow  \quad  \text{for some $j\geq i$}: \Pi, y,j \models_K  \varphi_2
                                                                  \text{ and } \Pi, y, k \models_K  \varphi_1 \text{ for all }k\in [i,j-1]\\
\Pi, y, i \models_K \exists x. \varphi \quad & \Leftrightarrow  \quad \text{for some initial path $\pi$ of $K$ such that }  \pi[0,i]=\Pi(y)[0,i],\,
\\
& \phantom{\Leftrightarrow \quad  }\,\,\,
 \Pi[x \leftarrow \pi],x,i \models \varphi
\end{array} \]
where $\Pi[x \leftarrow \pi](x)=\pi$
and $\Pi[x \leftarrow \pi](y)=\Pi(y)$ for all $y\neq x$.
$K$ \emph{satisfies} $\varphi$, written $K\models \varphi$, if there
is a path assignment $\Pi$ of $K$ and $y\in \Var$ such that $\Pi,
y,0 \models_K \varphi$. If $\varphi$ is a \emph{sentence},
then the satisfaction relation $\Pi, y,0 \models_K \varphi$ is
independent of $y$ and $\Pi$.  
Note that $\CTLStar$ corresponds to the set of sentences in
the one-variable fragment of $\HCTLStar$.

\begin{example}\label{example:witnessFormulaHyper} The $\HLTL$ sentence
  $\varphi_p := \exists x.\exists y.\,\, p[x]\,\,\Until\,\,\Bigl((p[x]\wedge \neg p[y])\wedge \Next\Always(p[x] \leftrightarrow p[y])\Bigr)$\\
 asserts that there are $\ell>0$ and two distinct initial paths $\pi$ and $\pi'$ such that $p$ always holds along the prefix $\pi[0,\ell]$, $p$ does not hold at position $\ell$ of $\pi'$, and for all $j>\ell$, the valuations of $p$ at position $j$ along $\pi$ and $\pi'$ coincide.
This requirement is clearly non-regular.  
\end{example}

\noindent \textbf{The novel logic $\MHCTLStar$.}
The syntax of $\MHCTLStar$ formulas $\varphi$  is as follows:
\[\varphi ::= \top  \ | \ p[x] \ | \ \neg \varphi \ | \ \varphi \wedge \varphi \ | \ \Next \varphi \ | \ \PNext \varphi\ | \ \varphi \Until \varphi \ | \ \varphi \PUntil \varphi \ | \ \exists x. \varphi \ | \ \GExists x. \varphi
\]
where $\PNext$ and $\PUntil$ are the past counterparts of the temporal
modalities $\Next$ and $\Until$, respectively, and $\GExists x$ is the
\emph{general} (hyper) existential quantifier for variable $x$.  We also use some
shorthands: $\GForall x.\,\varphi:=\neg \GExists x.\,\neg\varphi$ (``general
universal path quantifier''), $\PEventually\varphi:= \top \PUntil
\varphi$ (``past eventually'') and its dual $\PAlways \varphi:=\neg
\PEventually\neg\varphi$ (``past always''). The notion of sentence is
defined as for $\HCTLStar$. The semantics of the modalities $\PNext$,
$\PUntil$, and $\GExists x$ is as follows.
\[ \begin{array}{ll}
\Pi, y, i \models_K \PNext \varphi \quad & \Leftrightarrow \quad i>0 \text{ and }  \Pi, y, i-1 \models_K \varphi \\
\Pi, y,i \models_K  \varphi_1\PUntil \varphi_2 \quad & \Leftrightarrow \quad  \text{for some $j\leq i$}: \Pi, y,j \models_K  \varphi_2
                                                                  \text{ and } \Pi, y, k \models_K  \varphi_1 \text{ for all }k\in [j+1,i]\\
\Pi, y, i \models_K \GExists x. \varphi \quad & \Leftrightarrow \quad  \text{for some \emph{initial} path  $\pi$ of $K$, }
 \Pi[x \leftarrow \pi],x,i \models \varphi
\end{array} \]
The model-checking problem for $\MHCTLStar$ is checking given a \emph{finite} Kripke structure $K$ and a $\MHCTLStar$ sentence $\varphi$, whether
$K\models \varphi$. 
It is plain to see that $\MHCTLStar$ is bisimulation invariant  and satisfies the tree-model property. Hence, $K\models \varphi$ iff
  $\textit{Unw}(K)\models \varphi$.
Note that the set of sentences of the
$\GExists$-free one-variable
fragment of $\MHCTLStar$  corresponds to the well-known equi-expressive
linear-time  memoryful
extension $\MCTLStar$ of $\CTLStar$
\cite{KupfermanPV12}.

We consider now two relevant examples from the literature which
demonstrate the expressive power of $\MHCTLStar$. Both
  examples rely on the ability to express observational equivalence in
  the logic: for an agent $\agent \in \Agents$ and given two paths variables $x$ and
  $y$ in $\Var$, define $\Obsequiva(x,y):=\PAlways(\bigwedge_{p \in
    \Obs(\agent)}p[x] \leftrightarrow p[y])$.

The first example shows that the logic can express \emph{distributed
  knowledge}, a notion extensively investigated in
\cite{fagin1995reasoning}: a group of agents $A \subseteq \Agents$ has
distributed knowledge of $\varphi$, which we will write $D_A \varphi$,
if the combined knowledge of the members of $A$ implies $\varphi$. It
is well known that the modality $D_A$ cannot be expressed by means of modalities $K_\agent$ \cite{fagin1995reasoning}. 
Also, since $\HCTLStar$ cannot express the modality
$\Know_\agent$ (see Section~\ref{SubSec:FromKnoweldgeToHyper}), it
cannot either express $D_A$. However, $D_A$ is expressible in $\MHCTLStar$: given a group of agents $A \subseteq
\Agents$ and a formula $\varphi \in \MHCTLStar$, we define $\Pi, x,i
\models_K D_A \varphi$ by $\Pi, x,i \models_K \GForall
y. \,[(\bigwedge_{a \in A} \Obsequiva(x,y)) \rightarrow \varphi$].


%
%
%


The second example, inspired from
   \cite{alur2007model},
is an opacity requirement that we conjecture cannot be
  expressed neither in $\HCTLStar$ nor in $\KCTLStarS$. Assume that
  agent $\agent$ 
 can observe the   low-security (boolean)
  variables $p$ 
  (i.e., $p\in\Obs(a)$),  but not the
  high-security variables $p$ 
  (i.e., $p\notin
  \Obs(a)$). Consider the case of a secret represented by the value
  $\true$ of a high variable $p_s$. Then, the requirement $\forall
  x.\Always(p_s \rightarrow \GForall y. \Obs_\agent(x,y))$ says that
  whenever $p_s$ holds at a node in the computation tree, all the
  nodes at the same level have the same valuations of low
  variables. Hence, the observer $\agent$ cannot infer that the secret
  has been revealed.

\section{Expressiveness issues}
\label{sec:ExpressivenessIssues}

In this section, we establish that $\HCTLStar$ and $\KCTLStarS$ are expressively incomparable. 
Moreover, we show that $\KCTLStarS$ can be easily translated in linear time into $\MHCTLStar$. As a consequence,
$\MHCTLStar$ turns to be more expressive than both $\HCTLStar$ and $\KCTLStarS$.

Let $\Lang$ be a logic interpreted over Kripke structures,
$\Lang'$ be a logic interpreted over \emph{extended} Kripke structures, and $C$ be a class of Kripke structures.
For a sentence $\varphi$ of $\Lang$, a sentence
$\varphi'$ of $\Lang'$, and an observation map $\Obs$, $\varphi$ and $\varphi'$ \emph{are equivalent
  w.r.t. $C$ and $\Obs$}, written $\varphi \equiv_{C,\Obs} \varphi'$ if for all Kripke structures $K\in C$,
$K\models \varphi$ iff $(K,\Obs)\models \varphi'$.
$\Lang'$ is \emph{at least as expressive as} $\Lang$ w.r.t.
$C$, written $\Lang \leq_C \Lang'$, if for every sentence $\varphi$ of
$\Lang$, there is an observation map $\Obs$ and a sentence $\varphi'$ of $\Lang'$ such that
$\varphi \equiv_{C,\Obs} \varphi'$.
$\Lang$ is \emph{at least as expressive as} $\Lang'$ w.r.t. the class
$C$, written $\Lang' \leq_C \Lang$, if for every sentence $\varphi'$ of
$\Lang'$ and for every observation map $\Obs$, there is a sentence $\varphi$ of $\Lang$ such that
$\varphi \equiv_{C,\Obs} \varphi'$. Note the obvious asymmetry in the above two definitions due to the fact that for evaluating a sentence
in $\Lang'$, we need to fix an observation map.
If $\Lang \not\leq_C \Lang'$ and $\Lang' \not\leq_C \Lang$, then $\Lang$ and $\Lang'$ are \emph{expressively incomparable w.r.t. $C$}.
 We write
$\leq_\finite$ instead of $\leq_C$ if $C$ is the class of finite
 Kripke structures.

In order to prove that a
given formula $\varphi$ cannot be expressed in a logic $\Lang$, the naive technique is to build
two models   that $\varphi$ can distinguish  (i.e., $\varphi$ evaluates to
true on one model and to false on the other one), and prove that no formula of $\Lang$
can distinguish those two models.
A more involved technique, that we will use in the sequel in the expressiveness comparison between $\HCTLStar$ and $\KCTLStarS$, consists in building
two families of models  $(K_n)_{n\geq 1}$ and $(M_n)_{n\geq 1}$  such that $\varphi$ distinguishes between $K_n$
and $M_n$ for all $n$, and for every formula $\psi$ in $\Lang$, there is $n\geq 1$ such that $\psi$ does not distinguish
between $K_n$ and $M_n$.




\subsection{$\HCTLStar$ is not subsumed by $\KCTLStarS$}\label{SubSec:FromHyperToKnoweldge}

In this subsection, we show that $\HCTLStar$ (and the $\LTL$-like fragment $\HLTL$ as well) is not subsumed by
$\KCTLStarS$ even if we restrict ourselves to the class of finite   Kripke structures.

\begin{theorem}\label{theorem:HypernotSubsumedByKCTL} $\HLTL \not \leq_\finite \KCTLStarS$.
\end{theorem}


In order to prove Theorem~\ref{theorem:HypernotSubsumedByKCTL}, as witness $\HLTL$ sentence, we use the $\HLTL$ sentence $\varphi_p$ of
Example~\ref{example:witnessFormulaHyper} given by $  \varphi_p := \exists x.\exists y.\,\, p[x]\,\,\Until\,\,\Bigl((p[x]\wedge \neg p[y])\wedge \Next\Always(p[x] \leftrightarrow p[y])\Bigr)$.

We exhibit two families of \emph{regular} tree structures $(K_n)_{n>
  1}$ and $(M_n)_{n> 1}$ such that: (i) for all $n>1$, $\varphi_p$
distinguishes between $K_n$ and $M_n$, and (ii) for every $\KCTLStarS$
sentence $\psi$, there is $n>1$  such that
$\psi$ does \emph{not} distinguish between $(K_n,\Obs)$ and $(M_n,\Obs)$ for all observation maps $\Obs$. Hence,
Theorem~\ref{theorem:HypernotSubsumedByKCTL} follows. In the following, we fix $n>1$.

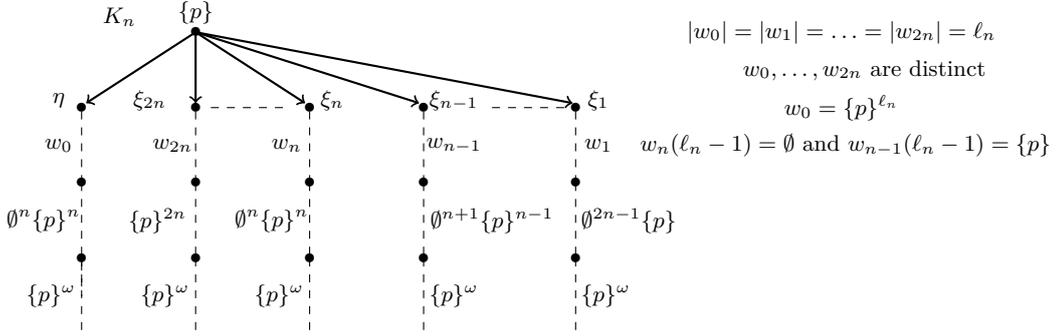
\begin{figure}[htp]
\begin{center}
\begin{tikzpicture}[node distance=15mm]
     \node(nodeTreeMn) at (1.5cm,2.0cm) {\footnotesize $\bullet$};
    \node() at (0.5cm,2.2cm) {\footnotesize \textbf{$K_n$}};
    \node() at (1.5cm,2.26cm) {\footnotesize \textbf{$\{p\}$}};
   \path[->, thick,black] (1.5cm,2.0cm) edge  (0.06cm,1.06cm);
    \path[->, thick,black] (1.5cm,2.0cm) edge  (1.50cm,1.06cm);
   \path[->, thick,black] (1.5cm,2.0cm) edge  (2.94cm,1.06cm);
    \path[->, thick,black] (1.5cm,2.0cm) edge  (4.44cm,1.06cm);
    \path[->, thick,black] (1.5cm,2.0cm) edge  (6.46cm,1.06cm);

       \node() at (10.0cm,2cm) {\footnotesize {$|w_0|=|w_1|=\ldots=|w_{2n}|=\ell_n$}};
    \node() at (10.3cm,1.5cm) {\footnotesize {$w_0,\ldots,w_{2n}$ are   distinct}};
    \node() at (10.0cm,1cm) {\footnotesize {$w_0=\{p\}^{\ell_n}$}};
     \node() at (10.0cm,0.5cm) {\footnotesize { $w_n(\ell_n-1)=\emptyset$ and $w_{n-1}(\ell_n-1)=\{p\}$}};

    \node() at (0cm,1cm) {\footnotesize $\bullet$};
     \node() at (-0.3cm,0.5cm) {\footnotesize \textbf{$w_0$}};
    \draw[dashed, black] (0cm,0cm) edge (0cm,1cm);
    \node(nodeX0) at (0cm,0cm) {\footnotesize $\bullet$};
    \node() at (-0.3cm,1.1cm) {\footnotesize \textbf{$\eta$}};
    \node(nodePX0) at (0cm,-1.0cm) {\footnotesize $\bullet$};
    \draw[dashed, black] (0cm,0cm) edge (0cm,-1.5cm);
    \node() at (-0.5cm,-0.5cm) {\footnotesize $ \emptyset^{n}\{p\}^{n}$};
    \node() at (-0.4cm,-1.5cm) {\footnotesize $\{p\}^{\omega}$};
    \draw[dashed, black] (0cm,-1.0cm) edge (0cm,-2cm);

    \node() at (1.5cm,1cm) {\footnotesize $\bullet$};
     \node() at (1.2cm,0.5cm) {\footnotesize \textbf{$w_{2n}$}};
    \draw[dashed, black] (1.5cm,0cm) edge (1.5cm,1cm);
    \node(nodeY0) at (1.5cm,0cm) {\footnotesize $\bullet$};
    \node() at (0.9cm,1.1cm) {\footnotesize \textbf{$\xi_{2n}$}};
    \node(nodePY0) at (1.5cm,-1.0cm) {\footnotesize $\bullet$};
     \draw[dashed, black] (1.5cm,0cm) edge (1.5cm,-1.0cm);
     \node() at (1.0cm,-0.50cm) {\footnotesize $\{p\}^{2n}$};
    \node() at (1.1cm,-1.5cm) {\footnotesize $\{p\}^{\omega}$};
    \draw[dashed, black] (1.5cm,-1.0cm) edge (1.5cm,-2cm);

        \node() at (3.0cm,1cm) {\footnotesize $\bullet$};
     \node() at (2.7cm,0.5cm) {\footnotesize \textbf{$w_{n}$}};
    \draw[dashed, black] (3.0cm,0cm) edge (3.0cm,1cm);
    \node(nodeZ0) at (3cm,0cm) {\footnotesize $\bullet$};
    \draw[dashed, black] (1.7cm,1cm) edge (2.7cm,1cm);
    \node() at (3.3cm,1.1cm) {\footnotesize \textbf{$\xi_{n}$}};
    \node(nodePZ0) at (3cm,-1.0cm) {\footnotesize $\bullet$};
    \draw[dashed, black] (3cm,0cm) edge (3cm,-1.0cm);
    \node() at (2.5cm,-0.50cm) {\footnotesize $  \emptyset^{n}\{p\}^{n}$};
    \node() at (2.6cm,-1.5cm) {\footnotesize $\{p\}^{\omega}$};
    \draw[dashed, black] (3cm,-1.0cm) edge (3cm,-2cm);

     \node() at (4.5cm,1cm) {\footnotesize $\bullet$};
     \node() at (4.9cm,0.5cm) {\footnotesize \textbf{$w_{n-1}$}};
    \draw[dashed, black] (4.5cm,0cm) edge (4.5cm,1cm);
    \node(nodeZ1) at (4.5cm,0cm) {\footnotesize $\bullet$};
    \node() at (4.9cm,1.1cm) {\footnotesize \textbf{$\xi_{n-1}$}};
    \node(nodePZ1) at (4.5cm,-1.0cm) {\footnotesize $\bullet$};
    \draw[dashed, black] (4.5cm,0cm) edge (4.5cm,-1.0cm);
    \node() at (5.4cm,-0.50cm) {\footnotesize $  \emptyset^{n+1}\{p\}^{n-1}$};
    \node() at (4.9cm,-1.5cm) {\footnotesize $\{p\}^{\omega}$};
    \draw[dashed, black] (4.5cm,-1.0cm) edge (4.5cm,-2cm);

     \node() at (6.5cm,1cm) {\footnotesize $\bullet$};
     \node() at (6.8cm,0.5cm) {\footnotesize \textbf{$w_{1}$}};
    \draw[dashed, black] (6.5cm,0cm) edge (6.5cm,1cm);
    \node(nodeZN) at (6.5cm,0cm) {\footnotesize $\bullet$};
    \node() at (6.8cm,1.1cm) {\footnotesize \textbf{$\xi_1$}};
    \node(nodePZN) at (6.5cm,-1.0cm) {\footnotesize $\bullet$};
     \draw[dashed, black] (6.5cm,0cm) edge (6.5cm,-1.0cm);
    \draw[dashed, black] (6.5cm,-1.0cm) edge (6.5cm,-2cm);
    \node() at (7.2cm,-0.50cm) {\footnotesize $ \emptyset^{2n-1}\{p\}$};
    \node() at (6.9cm,-1.5cm) {\footnotesize $\{p\}^{\omega}$};
    \draw[dashed, black] (5.4cm,1cm) edge (6.5cm,1cm);

 \end{tikzpicture}
\end{center}
\caption{The regular tree structure $K_n$ for the witness $\HLTL$ formula $\varphi_p$}
\label{fig-RegularTreeKNForHyper}
\end{figure}

\begin{definition}[The regular tree structures $K_n$ and $M_n$]  $K_n$, which is  illustrated in
    Fig.~\ref{fig-RegularTreeKNForHyper}, is any regular tree
    structure over $2^{\{p\}}$ satisfying the following for some
    $\ell_n>1$:
\begin{compactenum}
  \item The root has label $\{p\}$ and  $2n+1$  successors $\eta,\xi_1\ldots,\xi_{2n}$, and    there is a \emph{unique} initial path  visiting $\eta$ (resp., $\xi_k$ with $k\in [1,2n]$). We denote such a path  by $\pi(\eta)$ (resp., $\pi(\xi_k)$).
  \item There are $2n+1$ \emph{distinct} finite words $w_0,\ldots,w_{2n}$ over $2^{\{p\}}$ of length $\ell_n$ such that:
      \begin{compactitem}
        \item $w_0=\{p\}^{\ell_n}$, $w_n(\ell_n-1)=\emptyset$ and $w_{n-1}(\ell_n-1)=\{p\}$;
        \item  the trace of $\pi(\eta)$ is  $\{p\}\,w_0 \, \emptyset^{n}\{p\}^{n}\{p\}^{\omega}$;
        \item for all $k\in [1,2n]$, the trace of $\pi(\xi_k)$ is  $\{p\}\,w_k \, \emptyset^{2n-k}\{p\}^{k}\{p\}^{\omega}$.
      \end{compactitem}
\end{compactenum}\vspace{0.2cm}

$M_n$ is obtained from $K_n$ by replacing the label $\{p\}$ of the node  $\pi(\xi_n)(\ell_n+1+n)$ with $\emptyset$. Note that in $M_n$, the traces of $\pi(\xi_n)[\ell_n+1,\infty]$  and $\pi(\xi_{n-1})[\ell_n+1,\infty]$ coincide.
\end{definition}

In the regular tree structure $K_n$,   the trace
of the finite path $\pi(\eta)[0,\ell_n]$ is $\{p\}^{\ell_n+1}$, the label
of $\pi(\xi_n)$ at position $\ell_n $ is $\emptyset$, and the traces
of $\pi(\eta)[\ell_n+1,\infty]$ and $\pi(\xi_n)[\ell_n+1,\infty]$
coincide, which make  $\pi(\eta)$ and  $\pi(\xi_n)$ good candidates
to fulfill $\varphi_p$. Hence:

\begin{proposition}\label{prop:PropertiesOfKnHyper} $ K_n \models \varphi_p$.
\end{proposition}

\begin{proposition}\label{prop:PropertiesOfMnHyper}  $M_n \not\models \varphi_p$.
\end{proposition}
\begin{proof}

The construction ensures that for all distinct initial paths $\pi$ and $\pi'$ and $\ell\in [0,\ell_n]$, the  traces of  $\pi[\ell,\infty]$ and $\pi'[\ell,\infty]$ in $M_n$ are distinct (recall that $\pi(\xi_n)(\ell_n)$ and $\pi(\xi_{n-1})(\ell_n)$ have distinct labels). Moreover, $\pi(\eta)$ is the  unique initial path   of $M_n$ such that  for all $i\in [0,\ell_n]$, $p$ holds at position $i$. Thus, since $\pi(\eta)(\ell_{n}+1)$ has label $\emptyset$ and  there is no distinct initial path $\pi''$ of  $M_n$ such that the traces of  $\pi(\eta)[\ell_{n}+1,\infty]$ and $\pi''[\ell_{n}+1,\infty]$ coincide, by construction of $\varphi_p$, the result easily follows.
\end{proof}

A  $\KCTLStarS$ formula $\psi$ is  \emph{balanced} if for every until subformula $\psi_1\Until\psi_2$ of $\psi$, it holds that $|\psi_1|=|\psi_2|$. By using the atomic formula $\top$, it is trivial to convert a $\KCTLStarS$ sentence $\psi$ into an \emph{equivalent}
balanced $\KCTLStarS$ sentence of size at most $|\psi|^{2}$. This observation together with Propositions~\ref{prop:PropertiesOfKnHyper} and~\ref{prop:PropertiesOfMnHyper}, and the following non-trivial result provide a proof of
Theorem~\ref{theorem:HypernotSubsumedByKCTL}.

\newcounter{theo-MainTheoremExpressivenessHyper}
\setcounter{theo-MainTheoremExpressivenessHyper}{\value{theorem}}

 \begin{theorem}\label{theorem:MainTheoremExpressivenessHyper} Let $\psi$ be a \emph{balanced} $\KCTLStarS$ sentence such that $|\psi|<n$. Then, for all observation maps
 $\Obs$,   $(K_n,\Obs)\models \psi \Leftrightarrow (M_n,\Obs)\models \psi$.
\end{theorem}
\begin{proof}
A full proof is in Appendix~\ref{APP:MainTheoremExpressivenessHyper}.
Let $\Obs$ be an observation map. 
Evidently, it suffices to show that for all  initial paths $\pi$ and positions $i\in [0,\ell_n]$, $\pi,i\models_{K_n,\Obs} \psi$ iff $\pi,i\models_{M_n,\Obs} \psi$. The key for obtaining this result is that since $|\psi|<n$, $\psi$ cannot distinguish the nodes $\pi(\xi_n)(\ell_n+1)$ and $\pi(\xi_{n-1})(\ell_n+1)$ both in $(K_n,\Obs)$ and in $(M_n,\Obs)$. For $M_n$, this indistinguishability easily follows from the construction and is independent of the size of $\psi$. For $K_n$, the indistinguishability is non-trivial and is formally proved by defining equivalence relations on the set of nodes at distance $d\in [\ell_n+1,\ell_n+2n]$ from the root, which are parameterized by a natural number $h\in [1,n]$, where $h$ intuitively represents
the size of the current balanced subformula of $\psi$ in the recursive evaluation of $\psi$ on $K_n$.
\end{proof}

\subsection{$\KCTLStarS$ is not subsumed by $\HCTLStar$}\label{SubSec:FromKnoweldgeToHyper}

In this Subsection, we show that $\KCTLStarS$ (and the $\LTL$-like fragment $\KLTLS$ as well) is not subsumed by
$\HCTLStar$ even with respect to the the class of finite  Kripke structures.

For $p\in \AP$, an observation map $\Obs$ is \emph{$p$-blind} if for all agents $\agent$, $p \notin\Obs(a)$.

\begin{theorem}\label{theorem:KCTLnotSubsumedByHyper} $\KLTLS \not \leq_\finite \HCTLStar$.
\end{theorem}

As witness $\KLTLS$ sentence for Theorem~\ref{theorem:KCTLnotSubsumedByHyper}, we use the
$\KLTLS$ sentence $\varphi_p$ of Example~\ref{example:witnessFormulaKnowledge} given by
$\varphi_p:= \forall\Next\Eventually\Know_\agent\,\neg  p$.
We exhibit two families of \emph{regular} tree  structures $(K_n)_{n> 1}$ and $(M_n)_{n> 1}$ such that the following holds for all $n> 1$: (i) for each $p$-blind observation map $\Obs$, $\varphi_p$ distinguishes between $(K_n,\Obs)$ and
and $(M_n,\Obs)$, and (ii)  no $\HCTLStar$  formula $\psi$ of size less than $n$ distinguishes between
$K_n$ and $M_n$. Hence, Theorem~\ref{theorem:KCTLnotSubsumedByHyper} follows.

Fix $n> 1$.  In order to define $K_n$ and $M_n$,  we need additional definitions.

 An \emph{$n$-block} is a word in $\{p\}\emptyset^{*}$ of length at least $n+2$. Given  finite words $w_1,\ldots,w_k$ over $2^{\{p\}}$ having the same length $\ell$,
 the join $\join(w_1,\ldots,w_k)$ of $w_1,\ldots,w_k$ is the finite word over $2^{\{p\}}$ of length $\ell$ such that for all $i\in [0,\ell-1]$, $\join(w_1,\ldots,w_k)(i)= w_1(i)\cup \ldots \cup w_k(i)$.  For a finite word $w$ over $2^{\{p\}}$, the \emph{dual $\widetilde{w}$} of $w$ is the finite word over $2^{\{p\}}$ of length $|w|$ such that for all $i\in [0,|w|-1]$, $p\in \widetilde{w}(i)$ iff $p\notin w (i)$.

 Given  $n$ finite words $w_1,\ldots,w_n$ over $2^{\{p\}}$ of the same length, the tuple $\tpl{w_1,\ldots,w_n}$
 \emph{satisfies the $n$-fractal requirement} if for all $k\in [1,n]$, $\join(w_1,\ldots,w_k)$ has the form
  \[
  \join(w_1,\ldots,w_k)=\Bl^{k}_1\ldots \Bl^{k}_{m_k}\cdot \{p\}
  \]
  where $\Bl^{k}_1\ldots \Bl^{k}_{m_k}$ are $n$-blocks. Moreover, $m_1=n+4$ and the following holds:
 \begin{compactitem}
   \item  if $k<n$, then $w_{k+1}$ is obtained from $\join(w_1,\ldots,w_k)$ by replacing the   last symbol with $\emptyset$, and by replacing each $n$-block $\Bl^{k}_i$    of $\join(w_1,\ldots,w_k)$ by a sequence of $n+4$ $n$-blocks preceded by a non-empty word in $\emptyset^{*}$ of length at least $n+2$.
 \end{compactitem}

 \begin{remark}\label{remark:fractalRequirement} Assume that $\tpl{w_1,\ldots,w_n}$ satisfies the $n$-fractal requirement and let $\ell$ be the common length of $w_1,\ldots,w_n$. Then, for all $i\in [0,\ell-1]$, there is at most one $k\in [1,n]$ such that $p\in w_k(i)$.
 Moreover, $p\in w_1(0)$ and $p\in w_1(\ell-1)$.
 \end{remark}

\begin{figure}[htp]
\begin{center}
\begin{tikzpicture}[node distance=15mm]
    \node(nodeTreeMn) at (1.5cm,0.5cm) {\footnotesize $\bullet$};
    \node() at (1.2cm,0.7cm) {\footnotesize \textbf{$K_n$}};
    \node(nodeX0) at (0cm,0cm) {\footnotesize $\bullet$};
    \node() at (-0.3cm,0cm) {\footnotesize \textbf{$\eta$}};
    \node(nodePX0) at (0cm,-1.0cm) {\footnotesize $\bullet$};
    \draw[dashed, black] (0cm,0cm) edge (0cm,-1.5cm);
    \node() at (-0.4cm,-0.5cm) {\footnotesize $w_0$};
    \node() at (-0.4cm,-1.4cm) {\footnotesize $\{p\}^{\omega}$};
    \draw[dashed, black] (0cm,-1.0cm) edge (0cm,-1.8cm);
    \node(nodeY0) at (1.5cm,0cm) {\footnotesize $\bullet$};
    \node() at (1.2cm,0cm) {\footnotesize \textbf{$\xi_{1}$}};
    \node(nodePY0) at (1.5cm,-1.0cm) {\footnotesize $\bullet$};
     \draw[dashed, black] (1.5cm,0cm) edge (1.5cm,-1.0cm);
     \node() at (1.2cm,-0.5cm) {\footnotesize $w_1$};
    \node() at (1.3cm,-1.4cm) {\footnotesize $\emptyset^{\omega}$};
    \draw[dashed, black] (1.5cm,-1.0cm) edge (1.5cm,-1.8cm);
    \node(nodeZ0) at (3cm,0cm) {\footnotesize $\bullet$};
    \draw[dashed, black] (1.7cm,0cm) edge (2.7cm,0cm);
    \node() at (3.3cm,0cm) {\footnotesize \textbf{$\xi_{n}$}};
    \node(nodePZ0) at (3cm,-1.0cm) {\footnotesize $\bullet$};
    \draw[dashed, black] (3cm,0cm) edge (3cm,-1.0cm);
    \node() at (2.7cm,-0.5cm) {\footnotesize $w_n$};
    \node() at (2.8cm,-1.4cm) {\footnotesize $\emptyset^{\omega}$};
    \draw[dashed, black] (3cm,-1.0cm) edge (3cm,-1.8cm);

    \path[->, thick,black] (1.5cm,0.5cm) edge  (0.06cm,0.06cm);
    \path[->, thick,black] (1.5cm,0.5cm) edge  (1.50cm,0.06cm);
    \path[->, thick,black] (1.5cm,0.5cm) edge  (2.94cm,0.06cm);

    \node() at (7.0cm,0.4cm) {\footnotesize {$|w_0|=|w_1|=\ldots=|w_n|=\ell_n$}};
    \node() at (7.5cm,-0.3cm) {\footnotesize {$\tpl{w_1,\ldots,w_n}$ satisfies the $n$-fractal requirement}};
    \node() at (7.0cm,-1cm) {\footnotesize {$w_0$ is the dual of $\join(w_1,\ldots,w_n)$}};

 \end{tikzpicture}
\end{center}
\caption{The regular tree structure $K_n$ for the witness $\KLTLS$ formula $\varphi_p:= \forall\Next\Eventually\Know_\agent\,\neg  p$}
\label{fig-RegularTreeTNForKLTL}
\end{figure}\vspace{-0.2cm}

\begin{definition}[The regular tree structures $K_n$ and $M_n$]   $K_n$, which is illustrated in Fig.~\ref{fig-RegularTreeTNForKLTL}, is any regular tree structure over   $2^{\{p\}}$ satisfying the following for some $\ell_n>1$:
\begin{compactenum}
  \item The root has $n+1$ distinct successors $\eta,\xi_1\ldots,\xi_{n}$ and    there is a \emph{unique} initial path  visiting $\eta$ (resp., $\xi_k$ with $k\in [1,n]$). We denote such a path  by $\pi(\eta)$ (resp., $\pi(\xi_k)$).
  \item There are $n+1$ finite words $w_0,\ldots,w_n$ of length $\ell_n$ such  that:
  \begin{compactitem}
  \item the trace of $\pi(\eta)$ is $\emptyset \, w_0\, \{p\}^{\omega}$ and for all $k\in [1,n]$, the trace of $\pi(\xi_k)$
  is $\emptyset \, w_k\, \emptyset^{\omega}$;
    \item $\tpl{w_1,\ldots,w_n}$  satisfies the $n$-fractal requirement and $w_0$ is the dual of $\join(w_1,\ldots,w_n)$.
  \end{compactitem}
\end{compactenum}\vspace{0.2cm}

A \emph{main position} is a position in $[1,\ell_n]$. Let $i_\Alert$ be the \emph{third} (in increasing order) main position $i$ along $\pi(\xi_1)$ such that the label of $\pi(\xi_1)(i)$ in $K_n$ is $\{p\}$ (note that $i_\Alert$ exists). Then,
$M_n$ is obtained from $K_n$ by replacing the label $\{p\}$ of $\pi(\xi_1)$ at position $i_\Alert$ with $\emptyset$.
\end{definition}

By construction, in the regular tree structure $K_n$, for each non-root level, there is a node where $p$ holds and a node where  $p$ does not hold. Hence:

\begin{proposition}\label{prop:PropertiesOfKn} For each $p$-blind observation map $\Obs$, $(K_n,\Obs)\not\models \varphi_p$.
\end{proposition}

By  Remark~\ref{remark:fractalRequirement}, for each main position $i$, there is at most one $k\in [1,n]$ such that  the label of
 $\pi(\xi_k)(i)$ in $K_n$
is $\{p\}$. If such a $k$ exists, we say that $i$ is a \emph{main $p$-position} and $\xi_k$ is the \emph{type} of $i$.
Now, for the level of $M_n$ at distance $i_\Alert$ from the root, $p$ \emph{uniformly} does not hold (i.e., there is no node of $M_n$ at distance $i_\Alert$ from the root where $p$ holds). Hence:

\begin{proposition}\label{prop:PropertiesOfMn} For each $p$-blind observation map $\Obs$, $(M_n,\Obs)\models \varphi_p$.
\end{proposition}


Theorem~\ref{theorem:KCTLnotSubsumedByHyper} directly follows from Propositions~\ref{prop:PropertiesOfKn} and~\ref{prop:PropertiesOfMn} and the following result.

\newcounter{theo-MainResultExpressivenessKCTL}
\setcounter{theo-MainResultExpressivenessKCTL}{\value{theorem}}

\begin{theorem}\label{theorem;MainResultExpressivenessKCTL} For all $\HCTLStar$ sentences $\psi$  such that $|\psi|<n$,
$K_n\models \psi \Leftrightarrow M_n\models \psi$.
\end{theorem}
\begin{proof} A full proof is in Appendix~\ref{APP;MainResultExpressivenessKCTL}. The main idea  is that for a
$\HCTLStar$ sentence $\psi$ of size less than $n$, in the recursive evaluation of $\psi$ on the tree structure  $M_n$, there will be
$h_*\in [2,n]$ such that the initial path $\pi(\xi_{h_*})$ is not bound by the current path assignment. Then, the $n$-fractal requirement ensures that
in $M_n$, the main $p$-position $i_\Alert$ (which in $M_n$ has label $\emptyset$ along $\pi(\xi_1)$) is indistinguishable from the main $p$-positions $j$
of type $\xi_{h_*}$ which are sufficiently `near' to $i_\Alert$ (such positions $j$ have label $\emptyset$ along the initial paths $\pi(\xi_k)$ with $k\neq h*$). We formalize this intuition by defining equivalence relations on the set of main positions which are parameterized by $h_*$ and a natural number  $\EmphInt\in [0,n]$ and reflect the fractal structure of the main $p$-position displacement. Since the number of main $p$-positions of type $\xi_1$ following $i_\Alert$ is at least $n$, we then deduce that in all the positions $i$ such that $i\leq i_F$, where $i_F$ is the main $p$-position of type $\xi_1$ preceding $i_\Alert$, no $\HCTLStar$ formula $\psi$ can distinguish $M_n$ and $K_n$ with respect to path assignments such that $|\Pi|+|\psi|<n$, where $|\Pi|$ is the number of initial paths bound by $\Pi$. Hence, the result follows.
\end{proof}

\subsection{$\MHCTLStar$ unifies $\KCTLStarS$ and  $\HCTLStar$}\label{SubSec:FromKnoweldgeToMemoryfulHyper}

We show that $\KCTLStarS$ can be easily translated in linear time into the two-variable fragment of $\MHCTLStar$. Intuitively,
for a given  observation map, the knowledge modalities can be simulated by the general hyper path quantifiers combined with the temporal past
modalities. Hence, we obtain the following result (for a detailed proof see Appendix~\ref{APP:FromKCTLStarTOMemoryfulHyper}).

\newcounter{theo-FromKCTLStarTOExtHyper}
\setcounter{theo-FromKCTLStarTOExtHyper}{\value{theorem}}

\begin{theorem}\label{theo:FromKCTLStarTOMemoryfulHyper} Given a $\KCTLStarS$ sentence $\psi$ and an observation map $\Obs$, one can construct in \emph{linear time} a  $\MHCTLStar$ sentence $\varphi$ with just two path variables such that for each Kripke structure $K$, $K\models \varphi$ $\Leftrightarrow$ $(K,\Obs)\models \psi$.
\end{theorem}

By Theorems~\ref{theorem:HypernotSubsumedByKCTL}, \ref{theorem:KCTLnotSubsumedByHyper}, and~\ref{theo:FromKCTLStarTOMemoryfulHyper}, we obtain the following result.

\begin{corollary}  $\MHCTLStar$ is strictly more expressive than both $\HCTLStar$ and $\KCTLStarS$.
\end{corollary} 

\section{Model-checking against $\MHCTLStar$}
\label{sec-modelchecking}

In this section,
we address the model-checking problem for   $\MHCTLStar$. Similarly to the proof given
in \cite{ClarksonFKMRS14} for the less expressive logic $\HCTLStar$,
we show that the above problem is non-elementarily decidable by
linear time reductions from/to satisfiability of \emph{full}
Quantified Propositional Temporal Logic ($\QPTL$, for
short)~\cite{SistlaVW87}, which extends $\LTL$ with past ($\PLTL$) by quantification over propositions. As main contribution of
this section, we address complexity issues for the considered problem
by providing 
optimal complexity bounds in terms of a
parameter of the given $\MHCTLStar$ formula, we call \emph{strong
  alternation depth}.  For this, we first provide similar optimal
complexity bounds for satisfiability of full $\QPTL$. Our results also
solve complexity issues for $\HCTLStar$ left open in
\cite{ClarksonFKMRS14}.  With regard to $\QPTL$, well-known optimal
complexity bounds in terms of the alternation depth of existential and
universal quantifiers, concern the fragment of $\QPTL$ in prenex
normal form (quantifiers cannot occur in the scope of temporal
modalities)~\cite{SistlaVW87}. Unrestricted $\QPTL$ formulas can be
translated in polynomial time into equivalent (with respect to
satisfiability) $\QPTL$ formulas in prenex normal form, but in this
conversion, the nesting depth of temporal modalities in the original
formula (in particular, the alternation depth between always and
eventually modalities and the nesting depth of until modalities) lead
to an equal increasing in the quantifier alternation depth of the resulting formula. We
show that this can be avoided by \emph{directly} applying a non-trivial automatic
theoretic approach to unrestricted $\QPTL$ formulas. 
\vspace{0.2cm}

\noindent \textbf{Syntax and semantics of $\QPTL$.} $\QPTL$ formulas $\varphi$ over $\AP$ are defined as follows:
\[
\varphi ::= \top \ | \ p \ | \ \neg \varphi \ | \ \varphi \wedge \varphi \ | \ \Next \varphi\ | \ \PNext \varphi \ | \ \varphi \Until \varphi \ | \ \varphi \PUntil \varphi \ | \ \exists p\,.\varphi
\]
where $p\in \AP$.    The  positive normal form  of a $\QPTL$ formula $\varphi$ is obtained by pushing inward negations to propositional literals using De Morgan's laws and the duals $\Release$ (release), $\PRelease$ (past release), and $\forall p$ (propositional universal quantifier) of the modalities $\Until$, $\PUntil$, and $\exists p$.
A formula is \emph{existential}  if its positive normal form does not contain 
universal  quantifiers.

The semantics of $\QPTL$ is given w.r.t.~(infinite) \emph{pointed words} $(w,i)$ over $2^{\AP}$ consisting of an infinite word $w$ over $2^{\AP}$ and a position $i\geq 0$. All $\QPTL$ operators have the same semantics as in $\PLTL$ except for propositional quantification.
\[
(w,i) \models \exists p. \varphi \ \Leftrightarrow \  \text{ there is  } w'\in (2^{\AP})^{\omega} \text{ such that } w =_{\AP\setminus\{p\}} w' \text{ and } (w',i) \models \varphi
\]
where $w =_{\AP\setminus\{p\}} w'$ means that the projections of $w$ and $w'$ over $\AP\setminus \{p\}$ coincide.
For a $\QPTL$ formula $\varphi$, we denote by $\Lang_p(\varphi)$  the set of pointed words $(w,i)$ satisfying $\varphi$, and by $\Lang(\varphi)$ the set of infinite words $w$ such that
$(w,0)\in\Lang_p(\varphi)$;  $\varphi$ is satisfiable if $\Lang(\varphi)\neq\emptyset$.\vspace{0.2cm}

\noindent\textbf{Optimal bounds for satisfiability of $\QPTL$.} 
First, we provide a generalization of the standard notion of alternation depth between existential and universal quantifiers in a $\QPTL$ formula, we call \emph{strong alternation depth}. This notion takes into account also the presence  of temporal modalities occurrences between quantifier occurrences, but the nesting depth of temporal modalities is not considered (intuitively, it is collapsed to one).

\begin{definition} 
Let $\Op=\{\exists,\forall,\Until,\PUntil,\Release,\PRelease,\Always,\PAlways,\Eventually,\PEventually\}$. First, we define
the strong alternation length $\ell(\chi)$ of finite sequences $\chi\in \Op^{*}$: $\ell(\varepsilon)=0$,
$\ell(Q)=1$ for all $Q\in \Op$, and
\[
\ell(QQ'\chi)=
    \left\{
    \begin{array}{ll}
     \ell(Q'\chi)
      &    \text{ if \emph{either} }  Q,Q'\in\Op\setminus\{\exists,\forall\},  \text{ \emph{or} } Q\in \{\exists,\forall\} \text{ and }Q'\in\Op\setminus\{\exists,\forall\}
      \\
      \ell(Q'\chi)
      &   \text{ if either }  Q,Q'\in\{\exists,\Eventually,\PEventually\}  \text{ or } Q,Q'\in \{\forall,\Always,\PAlways\}
      \\
      1+   \ell(Q'\chi)
      &    \text{ otherwise }
    \end{array}
  \right.
  \]
Then, the strong alternation depth $\sad(\varphi)$ of a $\QPTL$ formula $\varphi$   is the maximum over the strong alternation lengths $\ell(\chi)$, where $\chi$ is the sequence of modalities in $\Op$ along a path in the tree encoding of the positive normal form of $\varphi$. 
\end{definition}

Note that for a $\QPTL$ formula $\varphi$ in prenex normal form, the strong alternation depth corresponds to the alternation depth of existential and universal quantifiers plus one.
For all $n,h\in\N$,  $\Tower(h,n)$ denotes a tower of exponential of height $h$ and argument $n$: $\Tower(0,n)=n$ and
          $\Tower(h+1,n)=2^{\Tower(h,n)}$. We establish the following result, 
          where $h$-\EXPSPACE\ is  the class of languages decided by  deterministic Turing machines bounded in space
          by functions of $n$ in $O(\Tower(h,n^c))$ for some constant $c\geq 1$.
.

\begin{theorem}\label{theorem:QPTLsatisfiability}
For all $h\geq 1$, satisfiability of $\QPTL$ formulas $\varphi$ with strong alternation depth at most $h$  is  $h$-\EXPSPACE-complete, and  $(h-1)$-\EXPSPACE-complete in case $\varphi$ is \emph{existential} (even if the allowed temporal modalities are in $\{\Next,\PNext,\Eventually,\PEventually,\Always,\PAlways\}$).
\end{theorem}

Here, we illustrate the upper bounds of Theorem~\ref{theorem:QPTLsatisfiability} (for the lower bounds see Appendix~\ref{APP:QPTLsatisfiability}).
In the   automata-theoretic approach for $\QPTL$ formulas $\varphi$ in prenex normal form, first, one converts the quantifier-free part $\psi$ of $\varphi$ into an equivalent B\"{u}chi nondeterministic automaton (B\"{u}chi $\NWA$) accepting $\Lang(\psi)$. Then, by using
the closure under language projection and complementation for B\"{u}chi $\NWA$, one obtains a   B\"{u}chi $\NWA$ accepting $\Lang(\varphi)$.  This approach does not work for unrestricted $\QPTL$ formulas $\varphi$, where quantifiers can occur in the scope of temporal modalities. In this case, for a   subformula $\varphi'$ of $\varphi$, we need to keep track of the full set $\Lang_p(\varphi')$ of pointed words $(w,i)$ satisfying $\varphi$, and not simply $\Lang(\varphi')$.

Thus, we need to use \emph{two-way} automata $\A$   accepting languages $\Lang_p(\A)$ of \emph{pointed words}. In particular, the proposed approach
is based on
a compositional translation of $\QPTL$ formulas into the so called class of
\emph{simple two-way  B\"{u}chi $($nondeterministic$)$ word automata}
(B\"{u}chi $\SNWA$). 
Essentially, given an input pointed word $(w,i)$,   a B\"{u}chi $\SNWA$, splits in two copies: the first one moves forwardly along the suffix $w[i,\infty]$ and the second one moves backwardly along the prefix $w[0,i]$ (see Appendix~\ref{APP:SNWAandHAA} for details).

Moreover, at each step of the translation into B\"{u}chi $\SNWA$, 
we use as an intermediate formalism, a two-way extension  of the class of one-way hesitant \emph{alternating} automata ($\HAA$, for short) over infinite words introduced in~\cite{KupfermanVW00}.
Like one-way $\HAA$,  the set of states $Q$ of a two-way $\HAA$   is partitioned into a set of components $Q_1,\ldots,Q_n$ such that moves from states in $Q_i$ lead to states in components $Q_j$ such that $j\leq i$. Moreover, each component is classified as either  negative, or B\"{u}chi,  or coB\"{u}chi:
in a negative (resp., B\"{u}chi/coB\"{u}chi) component $Q_i$, the unique allowed moves from $Q_i$ to $Q_i$ itself are backward (resp., forward). These syntactical requirements ensure that in a run     over a pointed word $(w,i)$, every infinite path  $\pi$ of the run  gets trapped in some B\"{u}chi or coB\"{u}chi component, and the path $\pi$ eventually use only forward moves. Moreover, the acceptance condition of a two-way $\HAA$ encodes a particular kind of parity condition of index 2:  a  B\"{u}chi/coB\"{u}chi component $Q_i$ has associated a subset $F_i\subseteq Q_i$ of accepting states. Then, a run is accepting if for every infinite path $\pi$, denoting with $Q_i$ the B\"{u}chi/coB\"{u}chi component in which $\pi$ get trapped, $\pi$ satisfies the
B\"{u}chi/coB\"{u}chi acceptance condition associated with $Q_i$. See Appendix~\ref{APP:SNWAandHAA} for a formal definition of two-way $\HAA$.

For two-way $\HAA$, we establish two crucial results. First, for a two-way $\HAA$ $\A$, the dual automaton $\widetilde{\A}$ obtained from $\A$ by dualizing the transition function, and by converting a B\"{u}chi (resp., coB\"{u}chi) component into a coB\"{u}chi (resp., B\"{u}chi) component is still a two-way $\HAA$. Thus, by standard arguments (see
e.g. \cite{Zielonka98}), we obtain the following.

\begin{lemma}[Complementation Lemma]
  The dual automaton $\dual{\A}$ of a two-way $\HAA$ $\A$ is a two-way $\HAA$ accepting the complement of
 $\Lang_p(\A)$.
\end{lemma}

Second,
by a non-trivial variation of the   method used in \cite{DaxK08}
to convert parity two-way alternating word automata  into  equivalent B\"{u}chi $\NWA$, we obtain the following result.

\newcounter{theo-FromHAAtoSNPA}
\setcounter{theo-FromHAAtoSNPA}{\value{theorem}}

\begin{theorem}
  \label{theorem:FromHAAtoSNPA}
  For a two-way $\HAA$ $\A$ with $n$ states, one can construct ``on the fly'' and in singly
  exponential time a B\"{u}chi $\SNWA$  accepting $\Lang_p(\A)$ with
  $2^{O(n\cdot \log(n))}$ states.
\end{theorem}

The proof of  Theorem~\ref{theorem:FromHAAtoSNPA} is in Appendix~\ref{APP:FromHAAtoSNPA}. 
Finally, by using the complementation lemma for two-way $\HAA$ and Theorem~\ref{theorem:FromHAAtoSNPA},   we establish the following Theorem~\ref{theorem:FromQPTLtoSNWA} (whose proof is in Appendix~\ref{APP:FromQPTLtoSNWA}), from which
the upper bounds of Theorem~\ref{theorem:QPTLsatisfiability} directly follow (note that  B\"{u}chi $\SNWA$ $\A$ can be trivially converted into  B\"{u}chi $\NWA$ accepting the set of infinite words $w$ such that $(w,0)\in \Lang_p(\A)$, and for B\"{u}chi $\NWA$ checking nonemptiness is in \NLOGSPACE).
For a $\QPTL$ formula $\varphi$ in positive normal form, if there is a universal quantified subformula $\forall p.\, \psi$ of $\varphi$ such that
$\sad(\forall p.\, \psi)=\sad(\varphi)$,
 we say that $\varphi$ is a \emph{first-level universal} formula; otherwise, we say that $\varphi$ is a  \emph{first-level existential} formula.

\newcounter{theo-FromQPTLtoSNWA}
\setcounter{theo-FromQPTLtoSNWA}{\value{theorem}}

\begin{theorem}\label{theorem:FromQPTLtoSNWA}
Let $\varphi$ be a first-level existential (resp., first-level universal) $\QPTL$ formula  and $h=\sad(\varphi)$. Then, one can construct ``on the fly'' a B\"{u}chi $\SNWA$ $\A_\varphi$ accepting $\Lang_p(\varphi)$   in time $\Tower(h,O(|\varphi|))$ (resp., $\Tower(h+1,O(|\varphi|))$).
\end{theorem}

\noindent \textbf{Optimal bounds for model-checking of $\MHCTLStar$.} 
By giving linear-time reductions from/to satisfiability of $\QPTL$ and by exploiting Theorem~\ref{theorem:QPTLsatisfiability}, we provide   optimal bounds on the complexity of the finite-state model-checking problem of $\MHCTLStar$ in terms of the \emph{strong alternation depth}
of a $\MHCTLStar$ formula, which is defined as the homonym notion for $\QPTL$ formulas. In particular, the linear time reduction to satisfiability of $\QPTL$ generalizes the one given in~\cite{ClarksonFKMRS14} for the   model checking of $\HCTLStar$ (for details, see Appendix~\ref{APP:MHCTLModelChecking}).


\newcounter{theo-MHCTLModelChecking}
\setcounter{theo-MHCTLModelChecking}{\value{theorem}}

\begin{theorem}\label{theorem:MHCTLModelChecking}
For all $h\geq 1$ and $\MHCTLStar$ sentences $\varphi$ with strong alternation depth at most $h$, model-checking against $\varphi$ is  $h$-\EXPSPACE-complete, and  $(h-1)$-\EXPSPACE-complete in case $\varphi$ is \emph{existential} (even if the allowed temporal modalities are in $\{\Next,\PNext,\Eventually,\PEventually,\Always,\PAlways\}$).
\end{theorem}

\section{Discussion}
\label{sec-conclusion}

We plan to extend this work in many directions. We expand a few.
First, we  intend to   identify
tractable fragments of
$\MHCTLStar$ and to investigate their synthesis problem; note that
satisfiability of $\HCTLStar$ is already undecidable
\cite{ClarksonFKMRS14}. Second, we should extend the proposed
framework in order to deal with asynchronicity, as this would allow us
to considering more realistic information-flow security
requirements. 
In the same line, we would like to investigate the possibility of extending the verification 
of flow-information requirements to relevant classes of infinite-state systems such as the class of pushdown systems,
a model extensively investigated in
  software verification. 

\putbib
\end{bibunit}


\newpage

\appendix
\begin{bibunit}

\begin{LARGE}
  \noindent\textbf{Appendix}
\end{LARGE}

\newcounter{aux}

\section{Proofs from Section~\ref{sec:ExpressivenessIssues}}\label{APP:ExpressivenessIssues}

\subsection{Proof of Theorem~\ref{theorem:MainTheoremExpressivenessHyper}}\label{APP:MainTheoremExpressivenessHyper}

In this Subsection, we prove the following result, where for the fixed $n>1$, $K_n$ and $M_n$ are the regular tree structures
over $2^{\{p\}}$ defined in Subsection~\ref{SubSec:FromHyperToKnoweldge}.

\setcounter{aux}{\value{theorem}}
\setcounter{theorem}{\value{theo-MainTheoremExpressivenessHyper}}

\begin{theorem}
Let $\psi$ be a \emph{balanced} $\KCTLStarS$ sentence such that $|\psi|<n$. Then, for all observation maps
 $\Obs$,
\[
   (K_n,\Obs)\models \psi \Leftrightarrow (M_n,\Obs)\models \psi
\]
\end{theorem}
\setcounter{theorem}{\value{aux}}

In order to prove Theorem~\ref{theorem:MainTheoremExpressivenessHyper}, first, we give some definitions and preliminary results which capture some crucial properties of the regular tree structure $K_n$.

Recall that the sets of nodes of the regular tree structures $K_n$ and
$M_n$ coincide. Thus, in the following, for node, we mean a node of
$K_n$ (or, equivalently, $M_n$).  We denote by $\preceq$ the partial
order over the set of nodes defined as: $\tau \preceq \tau'$ iff there
is path from $\tau$ visiting $\tau'$. We write $\tau \prec \tau'$ to
mean that $\tau \preceq \tau'$ and $\tau\neq \tau'$. For nodes $\tau$
and $\tau'$, $\Nodes(\tau,\tau')$ denotes the set of nodes $\tau''$
such that $\tau \preceq \tau''\preceq \tau'$. A \emph{descendant} of a
node $\tau$ is a node $\tau'$ such that $\tau'\succeq \tau$.  By
construction of $K_n$ and $M_n$, for each non-root node $\tau$, there
is a unique initial path visiting $\tau$. Such a path will be denoted
by $\pi(\tau)$.  In particular, $\tau$ has a unique successor which is
denoted by $\SUCC(\tau)$.  For all observation maps $\Obs$,
$\KCTLStarS$ formulas $\psi$, and non-root nodes $\tau$, we write
$\tau\models_{K_n,\Obs}\psi$ (resp., $\tau\models_{M_n,\Obs}\psi$) to
mean that $\pi(\tau),|\tau|\models_{K_n,\Obs}\psi$ (resp.,
$\pi(\tau),|\tau|\models_{M_n,\Obs}\psi$). Recall that $|\tau|$ is the
distance of $\tau$ from the root.

Given an observation map $\Obs$ and an agent $\agent$, we say that two
nodes $\tau$ and $\tau'$ are $\Obs_\agent$-equivalent in $K_n$ (resp.,
$M_n$) if the traces of the unique finite paths from the root to
$\tau$ and $\tau'$, respectively, are $\Obs_\agent$-equivalent.

\begin{definition}[Main nodes]

A \emph{main position} is a position in $[\ell_n+1,\ell_n+2n]$.\footnote{Recall that $\ell_n$ is the common length of the words $w_0,w_1,\ldots,w_{2n}$ labeling $\pi(\eta)[1,\ell_n]$, $\pi(\xi_1)[1,\ell_n],\ldots, \pi(\xi_{2n})[1,\ell_n]$, respectively.} A \emph{main node} is a non-root node $\tau$ such that $\pi(\tau)$ visits $\tau$ at a main position (i.e., the distance $|\tau|$ of $\tau$ from the root is in $[\ell_n+1,\ell_n+2n]$).\footnote{See Fig.~\ref{fig-RegularTreeKNForHyper} for clarity.}   A \emph{main $p$-node} (resp., \emph{main $\emptyset$-node})  is a main node whose label in $K_n$ is $\{p\}$ (resp., $\emptyset$).
 For a $\emptyset$-main node $\tau$, we denote by
  $p(\tau)$ the \emph{smaller descendant} $\tau'$ of $\tau$ in $K_n$ (with respect to the partial order $\preceq$) such that $\tau'$ is a $p$-main node. Note that by construction $p(\tau)$ is always defined. Moreover, for a main node $\tau$, let $D(\tau)$ be the number of descendants of $\tau$ which are main nodes.
  The \emph{order of a $\emptyset$-main node $\tau$} is the number of   descendants of $\tau$ in $K_n$ which are
$\emptyset$-main nodes.
\end{definition}

Since the traces of $\pi(\eta)[1,\ell_n]$, $\pi(\xi_1)[1,\ell_n],\ldots, \pi(\xi_{2n})[1,\ell_n]$ are distinct, and the labels of $K_n$ and $M_n$ are in $2^{\{p\}}$, by construction, the following holds.

\begin{remark}\label{remark:ObservationMapsExpressivenessHyper}
For all observation maps $\Obs$ and agents $\agent$ such that $p\in \Obs(\agent)$, two main nodes $\tau$ and $\tau'$ are
$\Obs_\agent$-equivalent in $K_n$ (resp., $M_n$) iff $\tau=\tau'$.
\end{remark}

 Now, for each $h\in [1,n]$, we introduce the crucial notion of $h$-compatibility  between main nodes.  Intuitively, this notion allows to capture the properties
 which make two main nodes  indistinguishable from balanced $\KCTLStarS$ formulas of size at most $h$ when evaluated on the regular tree structure $K_n$.

\begin{definition}[$h$-Compatibility]Let $h\in [1,n]$. Two main nodes $\tau$ and $\tau'$ are \emph{$h$-compatible} if one of the following holds:
\begin{compactitem}
  \item $\tau$ and $\tau'$ are $p$-main nodes, and \emph{either} $D(\tau)=D(\tau')$, \emph{or} $D(\tau)\geq 2h$, $D(\tau')\geq 2h$, and
  $|D(\tau)-D(\tau')|=1$;
  \item  $\tau$ and $\tau'$ are $\emptyset$-main nodes, and one of the following holds,  where $o(\tau)$ and $(\tau')$ are the orders of $\tau$ and $\tau'$:
  \begin{compactitem}
  \item $o(\tau)=o(\tau')$ and $D(\tau)=D(\tau')$;
  \item $o(\tau)=o(\tau')$, $D(p(\tau))\geq h$,  $D(p(\tau'))\geq h$, $D(\tau)\geq 2h$, $D(\tau')\geq 2h$, $|D(\tau)-D(\tau')|=1$;
  \item  $o(\tau)\geq h$, $ o(\tau')\geq h$, $|o(\tau)-o(\tau')|=1$, $D(\tau)\geq 2h$, $D(\tau')\geq 2h$, $|D(\tau)-D(\tau')|=1$.
  \end{compactitem}
\end{compactitem}
We denote by $R(h)$ the  binary relation over the set of main nodes such that $(\tau,\tau')\in R(h)$ iff $\tau$ and $\tau'$ are
 $h$-compatible.
\end{definition}

\begin{remark}For all $h\in [1,n]$, $R(h)$ is an equivalence relation.
\end{remark}

The following two Propositions~\ref{prop:StructuralPropertiesRkOne} and~\ref{prop:StructuralPropertiesRkTwo} capture some crucial properties of the equivalence relation $R(h)$. They are used in the next Lemma~\ref{lemma:mainLemmaExpressivenessHyper}, two show that two $h$-compatible main nodes are indistinguishable from balanced $\KCTLStarS$ formulas of size at most $h$ when evaluated on the regular tree structure $K_n$.

\begin{proposition}\label{prop:StructuralPropertiesRkOne} Let $h\in [2,n ]$,  $(\tau,\tau')\in R(h)$, and $\Obs$ be an observation map.  Then, for all agents $\agent$ and nodes $\tau_1$ such that  $\tau$ and $\tau_1$ are $\Obs_\agent$-equivalent in $K_n$, there exists a node $\tau'_1$ such that
  $\tau'$ and $\tau'_1$ are $\Obs_\agent$-equivalent in $K_n$ and $(\tau_1,\tau'_1)\in R(h-1)$.
\end{proposition}
\begin{proof}
Fix an observation map $\Obs$.
 Let $h\in [2,n ]$, $(\tau,\tau')\in R(h)$,  and $\tau_1$ be a node such that $\tau$ and $\tau_1$ are $\Obs_\agent$-equivalent in $K_n$.
 We prove that there exists a  main node $\tau'_1$ such that
  $\tau'$ and $\tau'_1$ are $\Obs_\agent$-equivalent in $K_n$ and $(\tau_1,\tau'_1)\in R(h-1)\cup R(h)$. Thus, since
  $R(h)\subseteq R(h-1)$, the result follows. Note that $\tau_1$ is a main node. If $p\in \Obs(\agent)$, by Remark~\ref{remark:ObservationMapsExpressivenessHyper}, $\tau$ is the unique node which is $\Obs_\agent$ equivalent
  to $\tau$ itself. Hence, $\tau_1=\tau$, and by setting $\tau'_1=\tau'$, the result follows.

Now, assume that $p\notin \Obs(\agent)$. Hence, two nodes  are $\Obs_\agent$-equivalent in $K_n$  iff they have the same distance from the root.
 We assume that $\tau$ is a $p$-main node, hence, $\tau'$ is a $p$-main node as well. The case where $\tau$ is a $\emptyset$-main node is similar, and we omit the details here. In the rest of the proof, for a $\emptyset$-main node $\tau''$, we denote by $o(\tau'')$ the order of $\tau''$.

The case where $D(\tau)=D(\tau')$ is trivial (note that in this case, by construction, the main nodes $\tau$ and $\tau'$ have the same distance from the root). Now, assume that
$D(\tau)\neq D(\tau')$. If $\tau_1$ is a $p$-main node by setting $\tau'_1=\tau'$, by construction, the result easily follows. Otherwise,  $\tau_1$ is a $\emptyset$-main node and $|\tau_1|=|\tau|$. Let $\tau_*$ be the node of $\pi(\tau_1)$ having the same distance from the root as $\tau'$.
If $(\tau_1,\tau_*)\in R( h-1)$, then  by setting $\tau'_1=\tau_*$, the result follows. Otherwise, since
$(\tau,\tau')\in R(h)$, by construction, $D(\tau)\geq 2h$, $D(\tau')\geq 2h$, $|D(\tau)-D(\tau')|=1$, and one of the following holds:
\begin{compactitem}
  \item $D(\tau)>D(\tau')$, $\tau_*=\SUCC(\tau_1)$, and \emph{either} $\tau_*=p(\tau_1)$, \emph{or} $\tau_*$ is a $\emptyset$-main node, and $o(\tau_*)<h-1$: since $D(\tau_*)=D(\tau')$ ($\tau_*$ and $\tau'$ have the same distance from the root), we deduce that
  $D(p(\tau_1))>2h -(h-1)\geq h+1$. By construction,
  there exists a $\emptyset$-main node $\tau'_1$ at the same distance from the root as $\tau_*$ such that $o(\tau'_1)=o(\tau_1)$ and
  $D(p(\tau'_1))=D(p(\tau_1))-1$. Since $D(\tau_1)=D(\tau)$ and $D(\tau'_1)=D(\tau')$, we obtain that $(\tau_1,\tau'_1)\in R(h)$ and the result follows.
  \item $D(\tau)<D(\tau')$, $\tau_1=\SUCC(\tau_*)$,  and $o(\tau_1)<h-1$: since $D(\tau_1)=D(\tau)$ ($\tau$ and $\tau_1$ have the same distance from the root), we deduce that
  $D(p(\tau_1))>2h -(h-1)\geq h+1$. Since $o(\tau_*)\geq 2$, by construction,
  there exists a  $\emptyset$-main node $\tau'_1$ at the same level as $\tau_*$ such that $o(\tau'_1)=o(\tau_1)$,
 and $D(p(\tau'_1))=D(p(\tau_1))+1$. Thus, since $D(\tau_1)=D(\tau)$ and $D(\tau'_1)=D(\tau')$, we obtain that $(\tau_1,\tau'_1)\in R(h)$ and the result follows.
\end{compactitem}
\end{proof}

For a real number $r$, $\lfloor r\rfloor$ denotes the integral part of $r$.

\begin{proposition}\label{prop:StructuralPropertiesRkTwo} Let $h\in [2,n]$, $(\tau,\tau')\in R(h)$, and $\tau_2$ be a  main node such that $\tau_2\succeq \tau$.  Then, the following holds:
\begin{compactenum}
  \item either $\SUCC(\tau)$ and $\SUCC(\tau')$ are not main nodes, or $(\SUCC(\tau),\SUCC(\tau'))\in R(h-1)$;
  \item  there exists a main node $\tau'_2\succeq \tau'$ such that
   $(\tau_2,\tau'_2\,)\in R(\lfloor \frac{h}{2}\rfloor)$
  and  the restriction of $R(\lfloor \frac{h}{2}\rfloor)$
  to $\Nodes(\tau,\tau_2)\times \Nodes(\tau',\tau'_2)$ is total.\footnote{Recall that a binary relation $R\subseteq S\times S'$ is total if for all $s\in S$
  (resp., $s'\in S'$), there exists $s'\in S'$ (resp., $s\in S$) such that $(s,s')\in R$.}
\end{compactenum}
\end{proposition}
\begin{proof}
We use the following preliminary result.\vspace{0.2cm}

\noindent \emph{Claim~1:} Let $h\in [1,n]$, $(\tau,\tau')\in R(h)$ such that $\tau$ is a $p$-main node,
and $\tau_2$ be a  main node such that $\tau_2\succeq \tau$. Then, there exist a main node $\tau'_2\succeq \tau'$ such that
   $(\tau_2,\tau'_2\,)\in R(h)$
  and  the restriction of $R(h)$
  to $\Nodes(\tau,\tau_2)\times \Nodes(\tau',\tau'_2)$ is total.

\noindent \emph{Proof of Claim~1:}
   Assume that $D(\tau)\neq D(\tau')$ (the other case being trivial).
  Since $(\tau,\tau')\in R(h)$, $\tau'$ is a $p$-main node as well. Moreover, $|D(\tau)-D(\tau')|=1$, and
  $D(\tau)\geq 2h$ and $D(\tau)\geq 2h$. We focus on the case $D(\tau')=D(\tau)+1$ (the other case when  $D(\tau)=D(\tau')+1$ being similar).
By construction every main node which is a descendent of either $\tau$ or $\tau'$ is a $p$-main node.
If $\tau_2=\tau$, then by setting $\tau'_2=\tau'$, the result trivially follows.
Otherwise, let $\tau'_1=\SUCC(\tau')$. Note that $\tau'_1$ is a $p$-main node and $D(\tau)=D(\tau'_1)$. Hence, the restriction of
$R(h)$ to $\{\tau\}\times \{\tau',\tau'_1\}$ is total. Thus, by definition of $R(h)$, the result easily follows.
\qed\vspace{0.2cm}

\noindent Now, we prove Proposition~\ref{prop:StructuralPropertiesRkTwo}.

Let $h\in [2,n]$, $(\tau,\tau')\in R(h)$, and $\tau_2$ be a  main node such that $\tau_2\succeq \tau$.  We prove Properties~1 and~2 by induction on $D(\tau)$.

For the base case, $D(\tau)=1$. By  definition of $R(h)$, we deduce that $D(\tau')=1$ as well, hence, Properties~1 and~2
   easily follow.

For the induction step, assume that $D(\tau)>1$. Hence,  $D(\tau')>1$ as well. Since $(\tau,\tau')\in R(h)$, only the following two cases can occur:
\begin{compactitem}
  \item $\tau$ and $\tau'$ are $p$-main nodes: Property~2 directly follows from Claim 1 and the fact that  $R(h)\subseteq R(\lfloor\frac{h}{2}\rfloor)$. Moreover,  since $D(\tau)>1$, $D(\tau')>1$, and $(\tau,\tau')\in R(h)$, by definition of $R(h)$, Property~1 easily follows.
\item   $\tau$ and $\tau'$ are $\emptyset$-main nodes: Property~1 easily follows. Now, let us consider Property~2.
 Let $o(\tau)$ and $o(\tau')$ be the orders of $\tau$ and $\tau'$. We distinguish two cases:
    \begin{compactitem}
       \item $o(\tau)\neq o(\tau')$:  since $(\tau,\tau')\in R(h)$,  $|o(\tau)-o(\tau')|=1$,
       $o(\tau)\geq h$,  $o(\tau')\geq h$,
        $D(\tau)\geq 2h$, $D(\tau')\geq 2h$, and $p(\tau)$ and $p(\tau')$ have the same distance from the root. Assume that
       $o(\tau)=o(\tau')+1$ (the other case being similar).  If $\tau_2=\tau$, then by setting $\tau'_2=\tau'$, the result trivially follows.
Otherwise, let $\tau_1=\SUCC(\tau)$. Note that $\tau_1$ is a $\emptyset$-main node and $D(\tau_1)=D(\tau')$. Hence, the restriction of
$R(h)$ to $\{\tau,\tau_1\}\times \{\tau'\}$ is total. Thus, by definition of $R(h)$ and the fact that $R(h)\subseteq R(\lfloor\frac{h}{2}\rfloor)$, the result easily follows.
      \item $o(\tau)= o(\tau')$:
if $D(\tau)=D(\tau')$, the result easily follows. Otherwise, since $(\tau,\tau')\in R(h)$,        $|D(p(\tau))-D(p(\tau'))|=1$,
      $D(p(\tau))\geq h$ and $D(p(\tau'))\geq h$. Hence, $(p(\tau),p(\tau'))\in R(\lfloor \frac{h}{2} \rfloor)$.
      By construction, it easily follows that for each main node $\tau_1 \in \Nodes(\tau,p(\tau))$, there exists
       $\tau'_1 \in \Nodes(\tau',p(\tau'))$ such that $(\tau_1,\tau'_1)\in R(\lfloor \frac{h}{2} \rfloor)$ and the restriction of
       $R(\lfloor \frac{h}{2} \rfloor)$ to $\Nodes(\tau,\tau_1)\times \Nodes(\tau',\tau'_1)$ is total.
       Thus, since $p(\tau)$ and $p(\tau')$ are $p$-main nodes, by Claim~1, Property~2 follows.
    \end{compactitem}
\end{compactitem}
This concludes the proof of Proposition~\ref{prop:StructuralPropertiesRkTwo}.
\end{proof}

\begin{lemma}\label{lemma:mainLemmaExpressivenessHyper} Let $\psi$ be a \emph{balanced}  $\KCTLStarS$ formula such that $|\psi|\leq n$, $\Obs$ be an observation map, and
$(\tau,\tau')\in R(|\psi|)$. Then,
    \[
    \tau\models_{K_n,\Obs} \psi \Leftrightarrow \tau'\models_{K_n,\Obs} \psi
    \]
  \end{lemma}
  \begin{proof} Fix an observation map $\Obs$.
  We use the following fact that directly follows from the semantics of $\KCTLStarS$ and the fact that in $K_n$, for  every  node $\tau$  such that $\tau$ is not a main node, and $\tau$ is a descendant of some main node, the trace of the unique path from $\tau$ is $\{p\}^{\omega}$.\vspace{0.2cm}

\noindent \emph{Claim 1.} Let $\tau$ and $\tau'$ be descendants of main nodes such that $\tau$ and $\tau'$ are not main nodes. Then, for each $\KCTLStarS$ formula,
\[
    \tau\models_{K_n,\Obs} \psi \Leftrightarrow \tau'\models_{K_n,\Obs} \psi
\]

Now, we prove Lemma~\ref{lemma:mainLemmaExpressivenessHyper}. Let $\psi$ be a balanced $\KCTLStarS$ formula such that $|\psi|\leq n$  and
$(\tau,\tau')\in R(|\psi|)$. We need to show that
    \[
    \tau\models_{K_n,\Obs} \psi \Leftrightarrow \tau'\models_{K_n,\Obs} \psi
    \]
The proof is by induction on $|\psi|$. The cases for the boolean connectives $\neg$ and $\wedge$, and the existential path quantifier $\exists$, directly follow from the inductive hypothesis and the fact that $R(h)\subseteq R(k)$ for all $h,k\in [1,n]$ such that $h\geq k$. For the other cases, we proceed as follows.
  \begin{itemize}
    \item Case $\psi=p'$ for some $p'\in\AP$: since $(\tau,\tau')\in R(|\psi|)$, $\tau$ and $\tau'$ have the same label in $K_n$. Hence, the result follows.
    \item Case $\psi=\Next\psi'$. If $\SUCC(\tau)$ and $\SUCC(\tau')$ are not main nodes, the result directly follows from Claim~1. Otherwise, by applying Proposition~\ref{prop:StructuralPropertiesRkTwo}(1), we obtain that    $(\SUCC(\tau),\SUCC(\tau'))\in R(|\psi'|)$. Hence, in this case, the result directly follows from the induction hypothesis.
    \item $\psi=\psi_1\Until\psi_2$:  we focus on the implication
     $\tau\models_{K_n,\Obs} \psi \Rightarrow \tau'\models_{K_n,\Obs} \psi$ (the converse implication being symmetric).
 Let $\tau\models_{K_n} \psi$.   Hence,
        there exists $\tau_2 \succeq \tau$ such that $\tau_2\models_{K_n,\Obs} \psi_2$ and $\tau_1\models_{K_n,\Obs} \psi_1$ for all nodes $\tau_1$ such that
        $\tau\preceq \tau_1\prec \tau_2$. We need to prove that $\tau'\models_{K_n,\Obs} \psi$.
        We distinguish two cases:
        \begin{compactitem}
          \item $\tau_2$ is a main node: since $(\tau,\tau')\in R(|\psi|)$, by applying Proposition~\ref{prop:StructuralPropertiesRkTwo}(2), there exists a main node $\tau'_2 \succeq \tau'$ such that the restriction of $R(\lfloor\frac{|\psi|}{2}\rfloor)$ to
              $\Nodes(\tau,\tau_2)\times \Nodes(\tau',\tau'_2)$ is total and $(\tau_2,\tau'_2)\in R(\lfloor\frac{|\psi|}{2}\rfloor)$. Since $\psi$ is balanced, $|\psi_1|=|\psi_2|$. Hence, for all $h=1,2$,
              $|\psi_h| \leq \lfloor\frac{|\psi|}{2}\rfloor$, and in particular, $R(|\psi_h|)\supseteq R( \lfloor\frac{|\psi|}{2}\rfloor)$. Hence, by applying the induction hypothesis, the result easily follows.
          \item $\tau_2$ is not a main node: let $\tau_*$ be the greatest (with respect to $\preceq$) ancestor of $\tau_2$ which is a main node. Note that
          $\tau_* \succeq \tau$. By reasoning as in the previous case, there exists $\tau'_* \succeq \tau'$ such that
          for all $\tau_1\in \Nodes(\tau',\tau'_*)$, $\tau_1\models_{K_n,\Obs} \psi_1$ and $(\tau_*,\tau'_*)\in R(\lfloor\frac{|\psi|}{2}\rfloor)$.
          We show that $\SUCC(\tau'_*)\models_{K_n,\Obs} \psi$, hence, the result follows. Since $\SUCC(\tau_*)$ is not a main node, by Proposition~\ref{prop:StructuralPropertiesRkTwo}(1), $\SUCC(\tau'_*)$ is not a main node as well. Thus, since
          $\tau_2\models_{K_n,\Obs} \psi$, by applying Claim~1, the result follows.
        \end{compactitem}
     \item  $\psi=\Know_\agent\,\psi_1$. We focus on the implication
     $\tau\models_{K_n,\Obs} \psi \Rightarrow \tau'\models_{K_n,\Obs} \psi$ (the converse implication being symmetric). Assume that
     $\tau\models_{K_n,\Obs} \psi$. Let $\tau'_1$ be a node such that  $\tau'_1$ and $\tau'$ are $\Obs_\agent$-equivalent in $K_n$.
     We need to show that $\tau'_1\models_{K_n,\Obs} \psi_1$.
Since $(\tau,\tau')\in R(|\psi|)$ and $R(|\psi|)$ is an equivalence relation, by applying Proposition~\ref{prop:StructuralPropertiesRkOne}, there exists
     a main node $\tau_1$ such that $\tau$ and $\tau_1$ are $\Obs_\agent$-equivalent in $K_n$, and $(\tau_1,\tau'_1)\in R(|\psi_1|)$.
      Since $\tau\models_{K_n,\Obs} \psi$, it holds that $\tau_1\models_{K_n,\Obs} \psi_1$. Thus,
       by applying the induction hypothesis, the result follows.
       \end{itemize}
 \end{proof}

Now, we can prove the crucial lemma from which Theorem~\ref{theorem:MainTheoremExpressivenessHyper} directly follows.

 \begin{lemma}\label{lemma:FinalLemmaExpressivenessHyper}  Let $\psi$ be a \emph{balanced} $\KCTLStarS$ formula such that $|\psi|<n$ and $\Obs$ be an observation map. Then, for all initial paths $\pi$ of $K_n$ (or, equivalently, $M_n$)  the following holds:
 \begin{compactenum}
   \item $\pi,0\models_{K_n,\Obs}\psi$ $\Leftrightarrow$ $\pi,0\models_{M_n,\Obs}\psi$.
   \item $\pi,i\models_{K_n,\Obs}\psi$ $\Leftrightarrow$ $\pi,i\models_{M_n,\Obs}\psi\quad$ for all $i\in [1,\ell_n]$.
      \item if $\pi$ does not visit node $\xi_n$ (i.e., $\pi\neq \pi(\xi_n)$), then for all $i\geq \ell_n+1$,
      \[\pi,i \models_{K_n,\Obs}\psi \Leftrightarrow \pi,i\models_{M_n,\Obs}\psi
      \]
 \end{compactenum}
\end{lemma}
\begin{proof} First, we make the following observation which directly follows from the semantics of $\KCTLStarS$, Remark~\ref{remark:ObservationMapsExpressivenessHyper}, and the fact that for each observation map $\Obs$ and agent $\agent$ such that $p\notin \Obs(\agent)$, two nodes $\tau$ and $\tau'$ are $\Obs_\agent$-equivalent in $K_n$ (resp., $M_n$) iff $|\tau|=|\tau'|$ (i.e., $\tau$ and $\tau'$ have the same distance from the root).\vspace{0.2cm}

\noindent \emph{Claim 1.} Let $K\in \{K_n,M_n\}$, $\tau$ and $\tau'$ be two non-root nodes such that   $|\tau|=|\tau'|\geq \ell_n+1$ and
in $K$, the traces of the unique paths starting from $\tau$ and $\tau'$, respectively,   coincide. Then, for all $\KCTLStarS$ formulas $\psi$ and observation maps $\Obs$:
\[
      \tau\models_{K,\Obs}\psi \Leftrightarrow \tau'\models_{K,\Obs}\psi
\]

 Now, we prove Properties~1--3 of Lemma~\ref{lemma:FinalLemmaExpressivenessHyper}.
  Fix an observation map $\Obs$. Let $\psi$ be a \emph{balanced} $\KCTLStarS$ formula such that $|\psi|<n$ and $\pi$ be an initial path of $K_n$ (or, equivalently, $M_n$). The proof of Properties~1--3 is by induction on $|\psi|$. The case for atomic propositions directly follows from construction. The cases for negation, conjunction, and existential path quantifier directly follows from the induction hypothesis (recall that for each non-root node $\tau$ there is exactly one initial path visiting $\tau$). For the remaining case, we proceed as follows.
 \begin{itemize}
   \item Cases $\psi=\Next\psi'$ or $\psi=\psi_1\Until \psi_2$:  assume that $\psi=\psi_1\Until \psi_2$ (the case where $\psi=\Next\psi'$ being similar).
   For Property~1, we apply the induction hypothesis for Property~1, and Property~2 for the considered case.
           For Property~3, we apply the induction hypothesis on Property~3. Now, let us consider Property~2.
       The case where $\pi$ does not visit $\xi_n$ directly follows from the induction hypothesis on Properties~2 and~3.
           Now, assume that  $\pi$ visits node $\xi_n$, i.e. $\pi=\pi(\xi_n)$.  Let $\tau_n$ be the first main node visited by $\pi(\xi_n)$. Note that $\tau_n$ is a $\emptyset$-main node and $\tau_n=\pi(\xi_n)(\ell_n+1)$. Since $\tau_n \succ \xi_n$, by the semantics of the until modality and applying the induction hypothesis on Property~2, it suffices to show that
           \[
            \tau_n\models_{K_n,\Obs}\psi \Leftrightarrow \tau_n\models_{M_n,\Obs}\psi
           \]
           Let $\tau_{n-1}$ be the first main node visited by $\pi(\xi_{n-1})$, and $\tau'_{n-1}=\SUCC(\tau_{n-1})$.
           Note that $\tau_{n-1}$ is a $\emptyset$-main node and $\tau_{n-1}=\pi(\xi_{n-1})(\ell_n+1)$. By construction, we have that $(\tau_n,\tau'_{n-1})\in R(n-1)$ and $(\tau_{n-1},\tau'_{n-1})\in R(n-1)$. Since $R(n-1)\subseteq R(|\psi|)$ (recall that $|\psi|<n$), by applying twice Lemma~\ref{lemma:mainLemmaExpressivenessHyper}, we obtain
         \[
          \tau_n\models_{K_n,\Obs}\psi \Leftrightarrow \tau_{n-1}\models_{K_n,\Obs}\psi
         \]
         Moreover, since in $M_n$, the traces of  the paths starting from $\tau_n$ and $\tau_{n-1}$ coincide, $\tau_n$ and $\tau_{n-1}$ have the same distance from the root, and such a distance is $\ell_n+1$, by Claim~1
         \[
          \tau_n\models_{M_n,\Obs}\psi \Leftrightarrow \tau_{n-1}\models_{M_n,\Obs}\psi
         \]
         By applying Property~3 for the considered case, we have that
         \[
          \tau_{n-1}\models_{K_n,\Obs}\psi \Leftrightarrow \tau_{n-1}\models_{M_n,\Obs}\psi
         \]
         Hence, the result follows.

   \item Case $\psi: \Know_\agent\psi'$: Properties~1 and~2 directly follow from the induction hypothesis and the fact that for all $i\in [0,\ell_n]$, the two traces of $\pi[0,i]$ in $K_n$ and $M_n$ coincide. Now, we prove Property~3.
       If $p\in \Obs(a)$, then since $i\geq \ell_n+1$, by Remark~\ref{remark:ObservationMapsExpressivenessHyper}, $\pi(i+1)$ is the unique node of $K_n$ (resp., $M_n$) which is $\Obs_\agent$-equivalent to $\pi(i+1)$ itself. Hence, in this case, the result directly follows from the induction hypothesis.

   Now, assume that $p\notin \Obs(a)$. Hence, two nodes are $\Obs_\agent$-equivalent if they have the same distance from the root.
   First, we consider the implication $\pi,i \models_{K_n,\Obs}\psi \Rightarrow \pi,i\models_{M_n,\Obs}\psi$.
   Assume that $\pi,i \models_{K_n,\Obs}\psi$. Let $\pi'$ be an initial path. We need to show that
   $\pi',i \models_{M_n,\Obs}\psi$. Since $\pi,i \models_{K_n,\Obs}\psi$, it holds that $\pi',i \models_{K_n,\Obs}\psi'$.
   Thus, if $\pi'$ does not visit $\xi_n$, then the result directly follows from the induction hypothesis on Property~3.
   Otherwise, since $i\geq \ell_n+1$, by construction, in $M_n$, the traces of $\pi(\xi_n)[i,\infty]$ and $\pi( \xi_{n-1})[i,\infty]$ coincide. Thus, since
   $i\geq \ell_n+1$, by Claim~1, $\pi(\xi_n),i \models_{M_n,\Obs}\psi' \Leftrightarrow \pi(\xi_{n-1}),i \models_{M_n,\Obs}\psi'$. Since
   $\pi(\xi_{n-1}),i \models_{K_n,\Obs}\psi'$, by applying the induction hypothesis on Property~3, the result follows.
   The converse implication $\pi,i \models_{M_n,\Obs}\psi \Rightarrow \pi,i\models_{K_n,\Obs}\psi$ is similar, but we use the fact that
   in $K_n$, for each $i\geq \ell_n+1$, the traces of $\pi(\xi_n)[i,\infty]$ and $\pi(\eta)[i,\infty]$ coincide.
 \end{itemize}
\end{proof}

\subsection{Proof of Theorem~\ref{theorem;MainResultExpressivenessKCTL}}\label{APP;MainResultExpressivenessKCTL}

In this Subsection, we prove the following result, where for the fixed $n>1$, $K_n$ and $M_n$ are the regular tree structures
over $2^{\{p\}}$ defined in Subsection~\ref{SubSec:FromKnoweldgeToHyper}.

\setcounter{aux}{\value{theorem}}
\setcounter{theorem}{\value{theo-MainResultExpressivenessKCTL}}

\begin{theorem}   For all $\HCTLStar$ sentences $\psi$  such that $|\psi|<n$,
$K_n\models \psi \Leftrightarrow M_n\models \psi$.
\end{theorem}
\setcounter{theorem}{\value{aux}}

In order to prove Theorem~\ref{theorem;MainResultExpressivenessKCTL},
first, we give some definitions and preliminary results which capture the recursive structure of $K_n$ and $M_n$.
In the following, for path assignment $\Pi$, we mean a path assignment of $K_n$.
Since $K_n$ and $M_n$ coincide but for the labeling (in particular, the labeling of the path $\pi[\xi_1]$ at position $i_\Alert$), a path assignment of $K_n$ is a path assignment of $M_n$ as well, and vice versa.

For the nodes $\xi_h$ and $\xi_k$ with $h,k\in [1,n]$, we write $\xi_h\preceq \xi_k$ to mean that $h\leq k$.
Recall that $\ell_n$ is the greatest main position and by construction, $\ell_n$ is a main $p$-position of type $\xi_1$.
For a main position $i$, $p(i)$ denotes the smallest main $p$-position $j$ such that $j\geq i$.
 A main position which is not a $p$-main position is called a \emph{main $\emptyset$-position}.

Fix $h_*\in [1,n]$ (representing node $\xi_{h*}$).

\begin{definition}[$h_*$-types and $h_*$-macro-blocks]
Let $i$ be a main $p$-position. The \emph{$h_*$-type of $i$}  is the type of $i$ if either $h_*=1 $, or $i\neq i_\Alert$; otherwise, the $h_*$-type of $i$ is $\xi_{h*}$.

An  \emph{$h_*$-macro-block $\Bl$} is a set of main positions of the form  $[i,j]$ such that $i<j$, $i$ and $j$ are main $p$-positions  having the same $h_*$-type $\xi_k$, and there is \emph{no} main $p$-position in  $[i+1,j-1]$ with $h_*$-type $\xi_r\preceq \xi_k$.

A \emph{pure  macro-block} is a  $1$-macro-block. For an $h*$-macro-block $\Bl=[i,j]$,
the \emph{$h_*$-type of $\Bl$} is the common $h_*$-type of $i$ and $j$.
\end{definition}

\begin{remark} For each main $p$-position $i$, there is at most one $h_*$-macro-block $\Bl$ whose first position is $i$.
\end{remark}

 When $h_*\neq 1$, intuitively, the main $p$-position
$i_\Alert$ is ``considered" a main $p$-position associated to the path $\pi(\xi_{h_*})$. More precisely, if we consider $h_*$-macro-blocks $\Bl=[i,j]$, where one bound is $i_\Alert$ and the other one has type $\xi_{h_*}$ (by construction, there are exactly two of such macro-blocks), then as we will prove,
positions $i$ and $j$ are indistinguishable by   $\HCTLStar$ formulas of size at most $n-1$ which are evaluated on $M_n$ with respect to  path assignments where \emph{$\pi(\xi_{h_*})$ is not bound.}

\begin{definition}[$h_*$-low-ancestors and $h_*$-orders]
Let $i$ be a main $p$-position of $h_*$-type  $\xi_k$. The \emph{$h_*$-low-ancestor of $i$} is the smallest main $p$-position $j>i$ whose $h_*$-type $\xi_r$ satisfies $\xi_r \prec \xi_k$, if such a position exists; otherwise the $h_*$-low-ancestor of $i$ is undefined.

Let $\Bl$ and $\Bl'$ be two $h*$-macro-blocks:  \emph{$\Bl'$ is the $h_*$-successor of $\Bl$}   if $\Bl$ and $\Bl'$ are of the forms $[i,j]$ and $[j,k]$, respectively. The \emph{$h_*$-order of  $\Bl$} is the length $\ell\geq 1$ of the maximal sequence $\Bl_1,\ldots,\Bl_\ell$ of $h_*$-macro-blocks
 such that $\Bl_1=\Bl$ and $\Bl_k$ is the $h_*$-successor of $\Bl_{k-1}$ for all $k\in [2,\ell]$. The \emph{$h_*$-order of a main $p$-position $i$} is the $h_*$-order of the $h_*$-macro-block having $i$ as first position if such a $h_*$-macro-block exists; otherwise, the $h_*$-order of $i$ is $0$.
\end{definition}

\begin{remark} For a main $p$-position $i$, either the $h_*$-type of $i$ is $\xi_1$, or the $h_*$-low-ancestor of $i$ is defined.
\end{remark}

 Now, for each $\EmphInt\in [0,n]$, we introduce the crucial notion of $(h_*,\EmphInt)$-compatibility  between main positions.  Intuitively, this notion allows to capture the properties
 which make two main positions  indistinguishable from $\HCTLStar$ sentences of size at most $\EmphInt$ when evaluated on $K_n$ (resp., $M_n$)
 and in case $h_*=1$ (resp., $h_*\neq 1$).

\begin{definition}[$(h_*,\EmphInt)$-Compatibility] Let $\EmphInt\in [0,n]$.
Two main $p$-positions $i$ and $j$ are \emph{$(h_*,\EmphInt)$-compatible} if  the following conditions are inductively satisfied, where $o(i)$ and $o(j)$ are the $h_*$-orders of $i$ and $j$:
\begin{compactitem}
  \item $i$ and $j$ have the same $h_*$-type;
\item either $o(i)=o(j)$, or $o(i)>\EmphInt$ and $o(j)>\EmphInt$;
  \item either  the common $h_*$-type of $i$ and $j$ is  $\xi_1$, or the $h_*$-low-ancestors of $i$ and $j$  are $(h_*,\EmphInt)$-compatible.
\end{compactitem}
\vspace{0.1cm}

Two main $\emptyset$-positions $i$ and $j$ are \emph{$(h_*,\EmphInt)$-compatible} if the following holds:
\begin{compactitem}
  \item $p(i)$ and $p(j)$ are $(h_*,\EmphInt)$-compatible;
  \item   \emph{either}
  $p(i)-i=p(j)-j$, \emph{or} $p(i)-i>\EmphInt$ and $p(j)-j>\EmphInt$.
\end{compactitem}
\vspace{0.2cm}

We denote by $R(h_*,\EmphInt)$ the binary relation over main positions defined as: $(i,j)\in R(h_*,\EmphInt)$ iff either $i$ and $j$ are
$(h_*,\EmphInt)$-compatible main $p$-positions, or $i$ and $j$ are $(h_*,\EmphInt)$-compatible main $\emptyset$-positions.
 \vspace{0.1cm}

 Two $h_*$-macro-blocks $[i,j]$ and $[i',j']$ are \emph{$(h_*,\EmphInt)$-compatible} if $(i,i')\in R(h_*,\EmphInt)$ and $(j,j')\in R(h_*,\EmphInt)$.
\end{definition}

\begin{remark}$R(h_*,\EmphInt)$ is an equivalence relation.
\end{remark}

The following two Propositions~\ref{prop:SpecularityPreliminary} and~\ref{prop:Specularity}  capture some crucial properties of the equivalence relation $R(h_*,\EmphInt)$. They are used in the next Lemma~\ref{lemma:specularityForFormulasSatisfaction} to show that two $(h_*,\EmphInt)$-compatible main positions  are indistinguishable from $\HCTLStar$ formulas of size at most $\EmphInt$ when evaluated on the regular tree structure $K_n$ whenever $h_*=1$, and are indistinguishable  from $\HCTLStar$ formulas of size at most $\EmphInt$ when evaluated on the regular tree structure $M_n$ whenever $h_*\neq 1$ and the initial path $\pi(\xi_{h_*})$ is \emph{not bound by the given path assignment}.

\begin{proposition}\label{prop:SpecularityPreliminary} Let $h_*\in [1,n]$, $\EmphInt\in [0,n]$, $i_L$ and $j_L$ be
two main $p$-positions of $h_*$-type $\xi_h$ such that $(i_L,j_L)\in R(h_*,\EmphInt)$, and $i_U$ and $j_U$ be two main $p$-positions of $h_*$-type $\xi_k$ such that $i_U>i_L$, $j_U>j_L$ and $(i_U,j_U)\in R(h_*,\EmphInt)$. If there is  \emph{no} main $p$-position in $[i_L+1,i_U-1]\cup [j_L+1,j_U-1]$ of $h_*$-type $\xi_r$ such that $\xi_r \preceq \xi_h$ and $\xi_r\preceq \xi_k$, then, the following holds:
\begin{enumerate}
  \item the restriction of $R(h_*,\EmphInt)$ to $[i_L,i_U]\times [j_L,j_U]$ is total and $(i_L+1,j_L+1)\in R(h_*,\EmphInt)$.
  \footnote{Recall that a binary relation $R\subseteq S\times S'$ is total if for all $s\in S$ (resp., $s\in S'$), there is $s\in S'$ (resp., $s\in S$) such that  $(s,s')\in R$.}
  \item for each $\wp\in [i_L,i_U-1]$, there exists $\wp'\in [j_L,j_U-1]$, such that $(\wp,\wp')\in R(h_*,\EmphInt)$ and the restriction of $R(h_*,\EmphInt)$ to $[i_L,\wp-1]\times [j_L,\wp'-1]$ is total.
\end{enumerate}
\end{proposition}
\begin{proof}
We prove Properties~1  and~2 by induction on $2n-(h+k)$. \vspace{0.1cm}

 \noindent \emph{Base case: $2n-(h+k)=0$.} Hence, $k=h=n$. By hypothesis and construction, the sets $[i_L+1,i_U-1]$ and $[j_L+1,j_U-1]$ contain only main $\emptyset$-positions, and they have cardinality at least $n$+1. Thus, since $\EmphInt\in [0,n]$, by definition of $(h_*,\EmphInt)$-compatibility, the result easily follows.\vspace{0.1cm}

 \noindent \emph{Base case: $2n-(h+k)>0$.} Let $\ell=\max{(h,k)}$. If $\ell=n$, by hypothesis $[i_L+1,i_U-1]$ and $[j_L+1,j_U-1]$ contain only main $\emptyset$-positions. Thus, by reasoning as in the base case, the result follows. Now, assume that $\ell<n$. Then,  by hypothesis and construction, it follows that there must be $m_i,m_j>n+3$ and $m_i+m_j$  $h_*$-macro-blocks of $h_*$-type $\xi_{\ell+1}$
 \[
 \Bl_{m_i}^{i},\ldots,\Bl_{1}^{i}, \Bl_{m_j}^{j},\ldots,\Bl_{1}^{j}
 \]
 such that the following holds, where $f_i$ (resp., $f_j$) is the first position of $\Bl_{m_i}^{i}$ (resp., $\Bl_{m_j}^{j}$), and
$l_i$ (resp., $l_j$) is the last position of $\Bl_{1}^{i}$ (resp., $\Bl_{1}^{j}$):
\begin{compactitem}
  \item for all $1\leq r\leq m_i$, $\Bl_{r}^{i}\subseteq [i_L+1,i_U-1]$ and $\Bl_{r}^{i}$ is the $h_*$-successor of
     $\Bl_{r+1}^{i}$  if $r<m_i$.
  \item for all $1\leq r\leq m_j$, $\Bl_{r}^{j}\subseteq [j_L+1,j_U-1]$ and $\Bl_{r}^{j}$ is the $h_*$-successor of
     $\Bl_{r+1}^{j}$  if $r<m_j$.
  \item There is no main $p$-position in $[i_L+1,f_i-1]\cup [j_L+1,f_j-1]\cup  [l_i+1,i_U-1]\cup [l_j+1,j_U-1]$ of $h_*$-type $\xi_r$
  such that $\xi_r \preceq \xi_{\ell+1}$.
\end{compactitem}
\vspace{0.2cm}

Hence, we also deduce that (recall that $\EmphInt\in [0,n]$).
\begin{compactitem}
\item $(f_i,f_j)\in R(h_*,\EmphInt)$ and $(l_i,l_j)\in R(h_*,\EmphInt)$;
  \item for all $1\leq r\leq n+1$, $\Bl_{r}^{i}$ and $\Bl_{r}^{j}$ are $(h_*,\EmphInt)$-compatible $h_*$-macro-blocks;
  \item for all $ n+1<r\leq m_i$ and $ n+1<s\leq m_j$, $\Bl_{r}^{i}$ and $\Bl_{s}^{j}$ are $(h_*,\EmphInt)$-compatible $h_*$-macro-blocks.
\end{compactitem}
\vspace{0.2cm}

Since $\ell+1\geq \max(h,k)+1$, by applying the induction on Property~1, we obtain that:
\begin{compactitem}
\item[(I)] the restriction of $R(h_*,\EmphInt)$ to $[i_L,f_i]\times [j_L,f_j]$ is total and $(i_L+1,j_L+1)\in R(h_*,\EmphInt)$;
\item[(II)] the restriction of $R(h_*,\EmphInt)$ to $[l_i,i_U]\times [l_j,j_U]$ is total;
\item[(III)] for all $1\leq r\leq n+1$, the restriction of  $R(h_*,\EmphInt)$ to $\Bl_{r}^{i}\times \Bl_{r}^{j}$ is total;
  \item [(IV)] for all $ n+1<r\leq m_i$ and $ n+1<s\leq m_j$, the restriction of  $R(h_*,\EmphInt)$ to $\Bl_{r}^{i}\times \Bl_{s}^{j}$ is total.
\end{compactitem}
\vspace{0.2cm}

Hence, since $m_i>n+3$ and $m_j>n+3$, Property~1 follows. Now, we prove Property~2. We distinguish, four cases:
\begin{compactitem}
\item  $\wp\in [i_L,f_i-1]$. Recall that $(f_i,f_j)\in R(h_*,\EmphInt)$, $f_i$ and $f_j$ have $h_*$-type $\xi_{\ell+1}$, and there is no main $p$-position in $[i_L+1,f_i-1]\cup [j_L+1,f_j-1]$ of type $\xi_r$
  such that $\xi_r \preceq \xi_{\ell+1}$. Thus, since $\ell=\max(h,k)$, we can apply the induction hypothesis on Property~2, and the result follows.
\item  $\wp\in [l_i,i_U-1]$. This case is similar to the previous one.
\item there is $1\leq r\leq n+1$ such that $\wp\in \Bl_{r}^{i}$ and $\wp$ is \emph{not} the last position of $\Bl_{r}^{i}$. Let
 $\Bl_{r}^{i}=[f_{r}^{i},l_{r}^{i}]$ and $\Bl_{r}^{j}=[f_{r}^{j},l_{r}^{j}]$.
 Recall that $\Bl_{r}^{i}$ and $\Bl_{r}^{j}$ are $(h_*,\EmphInt)$-compatible
  $h_*$-macro-blocks of  $h_*$-type $\xi_{\ell+1}$. Hence,
  $(f_{r}^{i},f_{r}^{i}),(l_{r}^{i},l_{r}^{i})\in R(h_*,\EmphInt)$, $f_{r}^{i}$, $f_{r}^{j}$, $l_{r}^{i}$, $l_{r}^{j}$ have $h_*$-type $\xi_{\ell+1}$, and there is no main $p$-position in $[f_{r}^{i}+1,l_{r}^{i}-1] \cup [f_{r}^{j}+1,l_{r}^{j}-1]$ of type $\xi_r$
  such that $\xi_r \preceq \xi_{\ell+1}$. Moreover, by Conditions (I), (III) and (IV) above, the restriction of
  $R(h_*,\EmphInt)$ to $[i_L, f_{r}^{i}]\times [j_L, f_{r}^{j}]$ is total. Thus, since
  $\wp\in [f_{r}^{i},l_{r}^{i}-1]$,
    by applying the induction hypothesis on Property~2, the result follows.
 \item there is   $ n+1<r\leq m_i$ such that  $\wp\in \Bl_{r}^{i}$ and $\wp$ is \emph{not} the last position of $\Bl_{r}^{i}$.
 Recall that for all  $ n+1<s\leq m_j$, $\Bl_{r}^{i}$ and $\Bl_{s}^{j}$ are $(h_*,\EmphInt)$-compatible $h_*$-macro-blocks of $h_*$-type $\xi_{\ell+1}$. Choice
 $n+1<s\leq m_j$ such that $s=m_j$ iff $r=m_i$. Note that such a $s$ exists since $m_j>n+3$.
 Let  $\Bl_{r}^{i}=[f_{r}^{i},l_{r}^{i}]$ and $\Bl_{s}^{j}=[f_{s}^{j},l_{s}^{j}]$. By
  Conditions (I), (III) and (IV) above, the restriction of
  $R(h_*,\EmphInt)$ to $[i_L, f_{r}^{i}]\times [j_L, f_{s}^{j}]$ is total.  Since
  $\wp\in [f_{r}^{i},l_{r}^{i}-1]$, by reasoning as in the previous case, Property~2 follows.
\end{compactitem}

\end{proof}

\begin{proposition}\label{prop:Specularity} Let $h_*\in [1,n]$,  $\EmphInt\in [1,n]$, $(\ell,\ell')\in R(h_*,\EmphInt)$,
and $\wp$ be a main position such that $\wp\geq \ell$.  Then,  the following holds:
\begin{enumerate}
  \item either $\ell=\ell'$ or $(\ell+1,\ell'+1)\in R(h_*,\EmphInt-1)$;
  \item there is a main position  $\wp'\geq \ell'$
such that  $(\wp,\wp')\in R(h_*,\EmphInt-1)$ and the restriction of
    $R(h_*,\EmphInt-1)$ to
$[\ell,\wp-1]\times [\ell',\wp'-1]$ is total.
\end{enumerate}
\end{proposition}
\begin{proof}
Let $h_*\in [1,n]$, $\EmphInt\in [1,n]$, $(\ell,\ell')\in R(h_*,\EmphInt)$ and $\wp$ be a main position such that $\wp\geq \ell$.
Recall that $\ell_n$ is the greatest main position. Moreover, by construction, $\ell_n$ is a main $p$-position having $h_*$-type $\xi_1$
 and $h_*$-order $0$.

We prove Properties~1 and~2   by induction on $\ell_n-\ell$. For the base case, $\ell_n-\ell=0$. Since
$(\ell,\ell')\in R(h_*,\EmphInt)$, by definition of $(h_*,\EmphInt)$-compatibility, $\ell_n-\ell'=0$ as well. Hence, Properties~1 and~2 follows.

For the induction step, assume that $\ell_n-\ell>0$. Hence, $\ell_n-\ell'>0$ as well.
First, we consider the case when $\ell$ is a main $p$-position (hence, $\ell'$ is a main $p$-position as well).
If $\ell=\ell'$, Properties~1 and~2 trivially hold. Now, assume that $\ell\neq \ell'$.
 Let $o(\ell)$ and $o(\ell')$
be the $h_*$-orders of $\ell$ and $\ell'$. We distinguish two cases:
\begin{itemize}
  \item $o(\ell)=0$: since $(\ell,\ell')\in R(h_*,\EmphInt)$, it holds that $o(\ell')=0$. Note that $\ell_n$ is the unique
  main $p$-position having $h_*$-type $\xi_1$ and $h_*$-order $0$. Thus, since $\ell\neq \ell_n$, $\ell' \neq \ell_n$,
  and $(\ell,\ell')\in R(h_*,\EmphInt)$,
  the $h_*$-low-ancestors  $a(\ell)$ and $a(\ell')$ of $\ell$ and $\ell'$ are defined and $(a(\ell),a(\ell'))\in R(h_*,\EmphInt)$.
  Moreover, denoting $\xi_h$ (resp., $\xi_k$) the common $h_*$-type of $\ell$ and $\ell'$ (resp., $a(\ell)$ and $a(\ell')$), there is no main $p$-position
in $[\ell,a(\ell)]\cup [\ell',a(\ell')]$ having $h_*$-type $\xi_r$ such that $\xi_r\preceq \xi_h$ and $\xi_r\preceq \xi_k$. Hence, by Proposition~\ref{prop:SpecularityPreliminary}(1), $(\ell+1,\ell'+1)\in R(h_*,\EmphInt)$. Thus, since $R(h_*,\EmphInt)\subseteq R(h_*,\EmphInt-1)$, Property~1 follows.
Now, we prove Property~2.
     We distinguish two cases:
  \begin{compactitem}
    \item $\wp\in [\ell,a(\ell)-1]$. By applying Proposition~\ref{prop:SpecularityPreliminary}(2), there exists
    $\wp'\in [\ell',a(\ell')-1]$ such that $(\wp,\wp')\in R(h_*,\EmphInt)$ and the restriction of
    $R(h_*,\EmphInt)$ to
$[\ell,\wp-1]\times [\ell',\wp'-1]$ is total. Thus, since $R(h_*,\EmphInt-1)\supseteq R(h_*,\EmphInt)$, Property~2 follows.
    \item $\wp\geq a(\ell)$.
  Since $R(h_*,\EmphInt-1)\supseteq R(h_*,\EmphInt)$,  by applying Proposition~\ref{prop:SpecularityPreliminary}(1),  the restriction of
    $R(h_*,\EmphInt-1)$ to
$[\ell,a(\wp)]\times [\ell',a(\wp)]$ is total. Thus, since
$(a(\ell),a(\ell'))\in R(h_*,\EmphInt)$, $\wp\geq a(\ell)$ and $a(\ell)>\ell$, by applying the induction hypothesis, Property~2 follows.
  \end{compactitem}
  \item $o(\ell)>0$: since $(\ell,\ell')\in R(h_*,\EmphInt)$, it holds that $o(\ell')>0$. Hence, there exist two $h_*$-macro-blocks of the form
  $\Bl=[\ell,i_U]$ and  $\Bl=[\ell',i_U']$. Let $o(i_U)$ and $o(i'_U)$ be the $h_*$-orders of $i_U$ and $i'_U$.
  Note that $o(i_U)=o(\ell)-1$ and $o(i'_U)=o(\ell')-1$. Since $(\ell,\ell')\in R(h_*,\EmphInt)$, either $o(\ell)=o(\ell')$, or $o(\ell), o(\ell')>\EmphInt$. Hence, only the following two cases are possible:
   \begin{compactitem}
    \item $o(i_U)=o(i'_U)$, or $o(i_U), o(i'_U)>\EmphInt$.
   Since the $h_*$-low-ancestor of $i_U$ (resp., $i'_U$) coincides with the $h_*$-low-ancestor of $\ell$ (resp., $\ell'$), we have that $(i_U,i'_U)\in R(h_*,\EmphInt)$.
    By reasoning as for the case $o(\ell)=0$ (we just replace $a(\ell)$ and $a(\ell')$ with $i_U$ and $i_U'$, respectively), Property~1 and~2 follow.
    \item there exists $k\geq 1$ such that $\{o(i_U),o(i'_U)\}= \{\EmphInt,\EmphInt+k\}$. It follows that $(i_U,i'_U)\in R(h_*,\EmphInt-1)$.
     Since $(\ell,\ell')\in R(h_*,\EmphInt)$, $(\ell,\ell')\in R(h_*,\EmphInt-1)$. Thus, by applying Proposition~\ref{prop:SpecularityPreliminary}(1), we obtain that $(\ell+1,\ell'+1)\in R(h_*,\EmphInt-1)$, and Property~1  follows. Now, we prove Property~2.
     First, assume that $\wp\in [\ell,i_U-1]$. Applying
    Proposition~\ref{prop:SpecularityPreliminary}(2), there exists $\wp'\in [\ell',i'_U-1]$ such that
    $(\wp,\wp')\in R(h_*,\EmphInt-1)$ and the restriction of
    $R(h_*,\EmphInt-1)$ to
$[\ell,\wp-1]\times [\ell',\wp'-1]$ is total. Hence, Property~2 follows.

Now, assume that $\wp\geq i_U$. By hypothesis, $\{o(i_U),o(i'_U)\}= \{\EmphInt,\EmphInt+k\}$ for some $k\geq 1$. Assume that $o(i_U)=\EmphInt$ and $o(i'_U)=\EmphInt+k$ (the other case being similar). Additionally, for simplicity, we also assume that $k=1$ (the general case can be handled in a similar way). Let $\Bl''=(i'_U,i''_U)$ be the $h_*$-macro-block which is the $h_*$-successor of $\Bl'$ (note that $\Bl''$ exists since $o(i'_U)=\EmphInt+1$).
Then, the $h_*$-order $o(i''_U)$ of $i''_U$ is $\EmphInt$ and $(i_U,i''_U)\in R(h_*,\EmphInt)$, hence, $(i_U,i''_U)\in R(h_*,\EmphInt-1)$ as well. Since
     $(i_U,i'_U)\in R(h_*,\EmphInt-1)$,  by  Proposition~\ref{prop:SpecularityPreliminary}(1),  the restriction of $R(h_*,\EmphInt-1)$ to $\Bl\times \Bl'$ (resp., $\Bl\times \Bl''$) is total. Hence, the restriction of $R(h_*,\EmphInt-1)$ to $[\ell,i_U]\times [\ell',i''_U]$ is total.
     Thus, since
$(i_U,i''_U)\in R(h_*,\EmphInt)$, $\wp\geq i_U$ and $i_U>\ell$, by applying the induction hypothesis, the result follows.
  \end{compactitem}
\end{itemize}

\bigskip

It remains to consider the case when $\ell$ is a $\emptyset$-main position. Since $(\ell,\ell')\in R(h_*,\EmphInt)$, $\ell'$ is a main $\emptyset$-position as well. Moreover, $(p(\ell),p(\ell'))\in R(h_*,\EmphInt)$ and \emph{either}
  $p(\ell)-\ell=p(\ell')-\ell'$, \emph{or} $p(\ell)-\ell>\EmphInt$ and $p(\ell')-\ell'>\EmphInt$. Thus, since we have already proved that Properties~1 and~2 hold when
$\ell$ is a main $p$-position, by applying the induction hypothesis to the main $p$-position $p(\ell)$, the result easily follows.
\end{proof}

\begin{lemma}\label{lemma:specularityForFormulasSatisfaction} Let $h_*\in [1,n]$, $\psi$ be an $\HCTLStar$ formula such that $|\psi|\leq n$, and
$(\ell,\ell')\in R(h_*,|\psi|)$.
   \begin{compactenum}
   \item If $h_*=1$, then for all path assignments $\Pi$ and $x\in \Var$,
    \[
    \Pi,x,\ell\models_{K_n} \psi \Leftrightarrow \Pi,x,\ell'\models_{K_n} \psi   \]
    \item    If $h_*\in [2,n]$, then for all path assignments $\Pi$ such that \emph{$\pi(\xi_{h_*})$ is not bound by $\Pi$}, and $x\in \Var$
    \[
    \Pi,x,\ell\models_{M_n} \psi \Leftrightarrow \Pi,x,\ell'\models_{M_n} \psi   \]
  \end{compactenum}
  \end{lemma}
  \begin{proof} We prove Property~2 (Property~1 being similar).
  Let $h_*\in [2,n]$, $\psi$ be an $\HCTLStar$ formula such that $|\psi|\leq n$, and   $\Pi$ be an assignment path such that $\pi(\xi_{h_*})$ is not bound by $\Pi$. We need to prove that for all $x\in \Var$ and  $(\ell,\ell')\in R(h_*,|\psi|)$,
    \[
    \Pi,x,\ell\models_{M_n} \psi \Leftrightarrow \Pi,x,\ell'\models_{M_n} \psi   \]
  The proof  is by induction on $|\psi|$. The cases for the boolean connectives $\neg$ and $\wedge$ directly follow from the inductive hypothesis and the fact that $R(h_*,\EmphInt)\subseteq R(h_*,\EmphInt')$ for all $\EmphInt,\EmphInt'\in [0,n]$ such that $\EmphInt\geq \EmphInt'$. For the other cases, we proceed as follows.
  \begin{itemize}
    \item Case $\psi=p'[y]$ for some $p'\in\AP$ and $y\in\Var$: we show that the labels of $\Pi(y)(\ell)$ and $\Pi(y)(\ell')$ in $M_n$ coincide.
    Since $(\ell,\ell')\in R(h_*,|\psi|)$, either $\ell$ and $\ell'$ are both main $\emptyset$-positions, or
    $\ell$ and $\ell'$ are both main $p$-positions.

   If $\Pi(y)$ is the initial path  visiting node $\eta$ (i.e., $\Pi(y)=\pi(\eta)$), then by construction, $\Pi(y)(\ell)$ and $\Pi(y)(\ell')$ have the same label in $M_n$, and the result follows. Otherwise, $\Pi(y)=\pi(\xi_k)$ for some $k\in [1,n]$.
       If $\ell$ and $\ell'$ are   main $\emptyset$-positions, then by construction, $\Pi(y)(\ell)$ and $\Pi(y)(\ell')$ have both empty label in $M_n$, and the result follows.
    Now, assume that $\ell$ and $\ell'$ are   main $p$-positions.
    By hypothesis, $k\neq h_*$ and $h_*\in [2,n]$. Since $(\ell,\ell')\in R(h_*,|\psi|)$,
by construction, \emph{either}
    $\ell$ and $\ell'$ have the same type $\xi_h$   and $\ell,\ell'\neq i_\Alert$, \emph{or} $\ell=i_\Alert$ (resp., $\ell'=i_\Alert$) and $\ell'$ has type $\xi_{h_*}$ (resp., $\ell$ has type $\xi_{h_*}$). Thus, since $k\neq h_*$, by construction of $M_n$, the result follows.
    \item Case $\psi=\Next\psi'$. Since $(\ell,\ell')\in R(h_*,|\psi|)$ and $|\psi'|=|\psi|-1$, by Proposition~\ref{prop:Specularity}(1), either $\ell=\ell'$, or $(\ell+1,\ell'+1)\in R(h_*,|\psi'|)$. In the first case, the result trivially follows.
        In the second case, since $(\ell+1,\ell'+1)\in R(h_*,|\psi'|)$, by applying the induction hypothesis, we have
      \[
      \Pi,x,\ell+1\models_{M_n} \psi' \Leftrightarrow \Pi,x,\ell'+1\models_{M_n} \psi'
      \]
    Hence, the result follows.
    \item $\psi=\psi_1\Until\psi_2$: we prove the direction $\Pi,x,\ell\models_{M_n} \psi \Rightarrow \Pi,x,\ell'\models_{M_n} \psi$ (the converse direction being symmetric). Let $\Pi,x,\ell\models_{M_n} \psi$. By hypothesis,
        there exists $\wp\geq \ell$ such that $\Pi,x,\wp\models_{M_n} \psi_2$ and $\Pi,x,i\models_{M_n} \psi_1$ for all $i\in [\ell,\wp-1]$.
       We will prove that $\Pi,x,\ell'\models_{M_n} \psi$. We distinguish two cases:
        \begin{compactitem}
          \item $\wp$ is a main position. Since $(\ell,\ell')\in R(h_*,|\psi|)$, by applying Proposition~\ref{prop:Specularity}(2), there exists
          $\wp'\geq \ell'$ such that $(\wp,\wp')\in R(h_*,|\psi|-1)$ and the restriction of $R(h_*,|\psi|-1)$ to
          $[\ell,\wp-1]\times [\ell',\wp'-1]$ is total.
          Since $R(h_*,|\psi|-1)\subseteq R(h_*,|\psi_2|)$, $(\wp,\wp')\in R(h_*,|\psi|-1)$,  and $\Pi,x,\wp\models_{M_n} \psi_2$, by applying the induction hypothesis, we obtain that $\Pi,x,\wp'\models_{M_n} \psi_2$. Moreover, since
          $R(h_*,|\psi|-1)\subseteq R(h_*,|\psi_1|)$, $\Pi,x,i\models_{M_n} \psi_1$ for all $i\in [\ell,\wp-1]$, and the restriction of $R(h_*,|\psi|-1)$ to
          $[\ell,\wp-1]\times [\ell',\wp'-1]$ is total, by applying the induction hypothesis, we obtain that $\Pi,x,i\models_{M_n} \psi_1$ for all $i\in [\ell',\wp'-1]$. Hence, $\Pi,x,\ell'\models_{M_n} \psi$ and the result follows.
          \item $\wp$ is not a main position. Hence, $\wp>\ell_n$ and $\Pi,x,\ell_n\models_{M_n} \psi$ (recall that $\ell_n$ is the greatest main position).
           By applying Proposition~\ref{prop:Specularity}(2), there exists
          $j\geq \ell'$ such that $(\ell_n,j)\in R(h_*,|\psi|-1)$ and the restriction of $R(h_*,|\psi|-1)$ to
          $[\ell,\ell_n]\times [\ell',j]$ is total. Since $|\psi|-1\geq 1$ and $(\ell_n,j)\in R(h_*,|\psi|-1)$, by  Proposition~\ref{prop:Specularity}(1), either $(\ell_n+1,j+1)\in R(h_*,|\psi|-1)$ or $j=\ell_n$. Since $\ell_n$ is the greatest
          main position, we deduce that $j=\ell_n$.
          Hence, since $R(h_*,|\psi|-1)\subseteq R(h_*,|\psi_1|)$ and $\Pi,x,i\models_{M_n} \psi_1$ for all $i\in [\ell,\ell_n]$, by applying
          the induction hypothesis, we have that $\Pi,x,i\models_{M_n} \psi_1$ for all $i\in [\ell',\ell_n]$. Thus, since $\Pi,x,\ell_n\models_{M_n} \psi$, the result follows.
         \end{compactitem}
     \item  $\psi=\exists y.\,\psi'$. Since $\ell$ and $\ell'$ are main positions, $\ell,\ell'>0$. By construction, for each initial path $\pi$ of $M_n$ and for each position $i >0$, $\pi$ is the unique path having $\pi[0,i]$ as a prefix. Hence, for all $i\in \{\ell,\ell'\}$,
         $\Pi,x,i\models_{M_n} \psi  \Leftrightarrow \Pi,x,i\models_{M_n} \psi'$. Hence, the result directly follows from the induction hypothesis.
  \end{itemize}

 \end{proof}

Now, we  prove the crucial lemma from which Theorem~\ref{theorem;MainResultExpressivenessKCTL} directly follows.
Let $\Bl_\Alert$ be the pure macro-block of type $t_1$ whose last position is $i_\Alert$. The size $|\Pi|$ of a path assignment is the number of initial paths of $K_n$ (or equivalently, $M_n$) which are bound by $\Pi$.

\begin{lemma}\label{lemma:MainLemmaExpressivenessKCTL} Let $\psi$ be an $\HCTLStar$ formula and $\Pi$ be a path assignment  such that
$|\Pi|+|\psi|<n$. Moreover, let $i_F$ be the first position of $\Bl_\Alert$.  Then, for all  $\ell\leq i_F$ and $x\in \Var$,
\[
\Pi,x,\ell\models_{K_n} \psi \Leftrightarrow \Pi,x,\ell\models_{M_n}\psi
\]
\end{lemma}
 \begin{proof}
 Let $\psi$, $\Pi$, $i_F$, $x$ and $\ell$ as in the statement of the lemma.  First, we make some crucial observations.
 The first one (Claim~1) directly follows from construction and the semantics of $\HCTLStar$.\vspace{0.2cm}

\noindent \emph{Claim 1:} for all $\ell>i_\Alert$, $\Pi,x,\ell\models_{K_n} \psi \Leftrightarrow \Pi,x,\ell\models_{M_n}\psi$\vspace{0.2cm}

Moreover, since $|\Pi|<n-1$, there must be $h_*$ such that \vspace{0.2cm}

\noindent \emph{Claim 2:} $h_*\in [2,n]$ and the path $\pi(t_{h_*})$ is \emph{not bound} by $\Pi$.\vspace{0.2cm}

Let $i_N$ be the fourth (in increasing order) main $p$-position of type $t_1$. Since $i_\Alert$ is the third (in increasing order) main $p$-position of type $1$, by construction, the $1$-order and the $h_*$-order of $i_N$ are both $n$. Hence, by definition of $(1,n-1)$-compatibility and
$(h_*,n-1)$-compatibility, the following holds. \vspace{0.2cm}

\noindent \emph{Claim 3:} $(1,i_F)\in R(1,n-1)$, $(i_F,i_\Alert)\in R(1,n-1)$, and $(i_\Alert,i_N)\in R(1,n-1)$. Moreover, $(1,i_F)\in R(h_*,n-1)$ and $(i_F,i_N) \in R(h_*,n-1)$. \vspace{0.2cm}

Now, we prove the lemma  by induction on $|\psi|$. The cases for the boolean connectives
  directly follows from the induction hypothesis. For the other cases, we proceed as follows.
\begin{itemize}
  \item $\psi= p'[y]$ for some $p'\in \AP$ and $y\in \Var$: by hypothesis $\ell\leq i_F$ and $i_F<i_\Alert$.  By construction, it follows that the labels
    of $\Pi(y)(\ell)$ in $K_n$ and $M_n$  coincide. Hence, the result follows.
  \item $\psi=\Next\psi'$. If $\ell<i_F$, we apply the induction hypothesis on $\psi'$ with respect to position $\ell+1\leq i_F$, and the result follows.

  Now, assume that $\ell=i_F$. By Claim~3, $(1,i_F)\in R(1,n-1)$. Moreover, since $|\psi|<n$,   $R(1,n-1)\subseteq R(1,|\psi|)$. Thus, by Lemma~\ref{lemma:specularityForFormulasSatisfaction}(1), we have that
  \[
     \Pi,x,1\models_{K_n} \psi \Leftrightarrow \Pi,x,i_F\models_{K_n}\psi
  \]
 By Claim~2 and~3,
 $(1,i_F)\in R(h_*,n-1)$, $h_*\in [2,n]$, and \emph{$\pi(\xi_{h_*})$ is not bound by the path assignment $\Pi$}.  Since,   $R(h_*,n-1)\subseteq R(h_*,|\psi|)$,  by Lemma~\ref{lemma:specularityForFormulasSatisfaction}(1), we have
  \[
     \Pi,x,1\models_{M_n} \psi \Leftrightarrow \Pi,x,i_F\models_{M_n}\psi
  \]
Thus, since $\psi=\Next\psi'$, by applying the induction hypothesis on $\psi'$ with respect to position $2\leq i_F$, the result follows.
  \item $\psi=\psi_1\Until\psi_2$.  By applying the induction hypothesis on $\psi_1$ and $\psi_2$ and by the semantics of the until modality, it suffices to
  show that
  \[
     \Pi,x,i_F\models_{K_n} \psi \Leftrightarrow \Pi,x,i_F\models_{M_n}\psi
  \]
  By Claim~3, $(i_F,i_\Alert)\in R(1,n-1)$ and $(i_\Alert,i_N)\in R(1,n-1)$. Moreover, since $|\psi|<n-1$,   $R(1,n-1)\subseteq R(1,|\psi|)$. Thus, by applying twice Lemma~\ref{lemma:specularityForFormulasSatisfaction}(1), we obtain
  \[
     \Pi,x,i_F\models_{K_n} \psi \Leftrightarrow \Pi,x,i_N\models_{K_n}\psi
  \]
By Claim~2 and~3,
 $(i_F,i_N)\in R(h_*,n-1)$, $h_*\in [2,n]$, and \emph{$\pi(\xi_{h_*})$ is not bound by the path assignment $\Pi$}.  Since,   $R(h_*,n-1)\subseteq R(h_*,|\psi|)$,  by applying Lemma~\ref{lemma:specularityForFormulasSatisfaction}(2), we obtain
  \[
     \Pi,x,i_F\models_{M_n} \psi \Leftrightarrow \Pi,x,i_N\models_{M_n}\psi
  \]
Since $i_N>i_\Alert$, by Claim~1,
  \[
     \Pi,x,i_N\models_{K_n} \psi \Leftrightarrow \Pi,x,i_N\models_{M_n}\psi
  \]
Hence, the result follows.
  \item $\psi = \exists y.\psi'$. By hypothesis $|\Pi|+|\psi|<n$. Hence, for each initial path $\pi$ of $K_n$ (or equivalently $M_n$) and for each $y\in\Var$,
  $|\Pi[y\leftarrow \pi]| +|\psi'|<n $. By applying the induction hypothesis, we have that
   $\Pi[y\leftarrow \pi],y,\ell\models_{K_n} \psi' \Leftrightarrow \Pi[y\leftarrow \pi],y,\ell\models_{M_n}\psi'$. Hence,
    $\Pi,x,\ell\models_{K_n} \psi \Leftrightarrow \Pi,x,\ell\models_{M_n}\psi $, and the result follows.
\end{itemize}

\end{proof}

\subsection{Proof of Theorem~\ref{theo:FromKCTLStarTOMemoryfulHyper}}\label{APP:FromKCTLStarTOMemoryfulHyper}

\setcounter{aux}{\value{theorem}}
\setcounter{theorem}{\value{theo-FromKCTLStarTOExtHyper}}

\begin{theorem} Given a $\KCTLStarS$ sentence $\psi$ and an observation map $\Obs$, one can construct in \emph{linear time} a  $\MHCTLStar$ sentence $\varphi$ with just two path variables such that for each Kripke structure $K$, $K\models \varphi$ $\Leftrightarrow$ $(K,\Obs)\models \psi$.
\end{theorem}
\setcounter{theorem}{\value{aux}}
\begin{proof}
Let $x_0,x_1\in \Var$ with $x_0\neq x_1$ and $\Obs$ be an observation map. We inductively define a mapping $f_\Obs: \KCTLStarS\times \{0,1\} \rightarrow \MHCTLStar$ as follows, where $h\in \{0,1\}$:
\begin{compactitem}
  \item $f_\Obs(\top,h) = \top$
  \item $f_\Obs(p,h)=p[x_h]$ for all $p\in\AP$;
  \item $f_\Obs(\neg \psi,h)= \neg f_\Obs(\psi,h)$;
   \item $f_\Obs(\psi_1\wedge \psi_2,h)= f_\Obs(\psi_1,h)\wedge f_\Obs(\psi_2,h)$;
   \item $f_\Obs(\Next \psi,h)=\Next f_\Obs(\psi,h)$;
   \item  $f_\Obs(\psi_1\,\Until\, \psi_2,h)= f_\Obs(\psi_1,h)\,\Until\, f_\Obs(\psi_2,h)$;
   \item  $f_\Obs(\exists\psi,h)= \exists x_h. f_\Obs(\psi,h)$;
   \item $f(\Know_\agent\psi,h)= \GForall x_{1-h}.\,  \bigl(\,\, (\PAlways\displaystyle{\bigwedge_{p\in \Obs(\agent)}} (p[x_h] \leftrightarrow p[x_{h-1}])) \longrightarrow  f_\Obs(\psi,x_{1-h})\bigr) $.
\end{compactitem}

By construction, for all $h=1,2$ and $\KCTLStarS$ sentences $\psi$, $f_\Obs(\psi,h)$ is a $\MHCTLStar$ sentence of size linear in $|\psi|$. Thus, Theorem~\ref{theo:FromKCTLStarTOMemoryfulHyper} directly follows from the following claim.\vspace{0.2cm}

\noindent \emph{Claim: } let $K=\tpl{S,s_0,E,V}$ be a Kripke structure and $\pi$ be an initial path of $K$. Then, for all $\KCTLStarS$ formulas $\psi$, $h=0,1$, $i\geq 0$, and path assignment $\Pi$ such that $\Pi(x_h)=\pi$, the following holds:
\[
 \Pi,x_h,i \models_K f_\Obs(\psi,h) \Leftrightarrow \pi,i\models_{(K,\Obs)}\psi
\]
\noindent \emph{Proof of the claim:} the proof is by induction on $|\psi|$.
\begin{itemize}
  \item $\psi =\top$: trivial;
  \item $\psi=p$ with $p\in \AP$: by construction, $f_\Obs(p,h)=p[x_h]$. Thus, since $\Pi(x_h)=\pi$, the result follows.
  \item $\psi=\neg \psi'$ or $\psi =\psi_1\wedge\psi_2$ or $\psi=\Next\psi'$ or $\psi=\psi_1\Until\psi_2$: by construction, the result directly follows from the induction hypothesis.
  \item $\psi=\exists\psi'$: then, $ \Pi,x_h,i \models_K f_\Obs(\psi,h)$ $\Leftrightarrow$ (by construction)
   $\Pi,x_h,i \models_K \exists x_h. f_\Obs(\psi',h)$  $\Leftrightarrow$ (by the semantics of $\MHCTLStar$)
   there exists an initial path $\pi'$ of $K$  such that $\pi'[0,i]=\Pi(x_h)[0,i]$ and
   $\Pi[x_h \leftarrow \pi'],x_h,i \models_K f_\Obs(\psi',h)$ $\Leftrightarrow$ (by the induction hypothesis and since $\Pi(x_h)=\pi$)
   there exists an initial path $\pi'$ of $K$ such that $\pi'[0,i]=\pi[0,i]$ and
   $\pi',i \models_{(K,\Obs)}\psi'$ $\Leftrightarrow$ (by the semantics of $\KCTLStarS$)
   $\pi,i\models_{(K,\Obs)}\psi$. Hence, the result follows.
    \item $\psi=\Know_\agent\psi'$: then, $ \Pi,x_h,i \models_K f_\Obs(\Know_\agent\psi,h)$ $\Leftrightarrow$ (by construction and the semantics of $\MHCTLStar$)
  for all initial paths $\pi'$ of $K$,
  \[
  \Pi[x_{1-h} \leftarrow \pi'],x_{1-h},i\models_K (\PAlways\displaystyle{\bigwedge_{p\in \Obs(\agent)}} (p[x_h] \leftrightarrow p[x_{h-1}])) \longrightarrow  f_\Obs(\psi',x_{1-h})
  \]
  $\Leftrightarrow$ (since $\Pi[x_{1-h} \leftarrow \pi'](x_h)=\Pi(x_h)=\pi$) for all initial paths $\pi'$ of $K$  such that
  $V(\pi'[0,i])$ and $V(\pi[0,i])$ are $\Obs_\agent$-equivalent, $\Pi[x_{1-h} \leftarrow \pi'],x_{1-h},i\models_K f_\Obs(\psi',x_{1-h})$
  $\Leftrightarrow$ (by the induction hypothesis) for all initial paths $\pi'$ of $K$  such that
  $V(\pi'[0,i])$ and $V(\pi[0,i])$ are $\Obs_\agent$-equivalent, $\pi',i\models_{(K,\Obs)} \psi'$ $\Leftrightarrow$ (by the semantics of $\KCTLStarS$)
  $\pi,i\models_{(K,\Obs)} \Know_\agent\psi'$. Hence, the result follows.
\end{itemize}
\end{proof}

\newpage

\section{Proofs from Section~\ref{sec-modelchecking}}\label{APP:sec-modelchecking}

\subsection{Formal definitions of B\"{u}chi $\SNWA$ and two-way $\HAA$}\label{APP:SNWAandHAA}

\noindent \textbf{B\"{u}chi $\SNWA$.} A  B\"{u}chi $\SNWA$ over an input alphabet $\Sigma$
is a tuple $\A=\tpl{Q,Q_0,\rho,F_-,F_+}$, where  $Q$ is a finite set of states, $Q_0\subseteq Q$ is a set of initial states, $\rho:Q\times \{\rightarrow,\leftarrow\}\times \Sigma\rightarrow 2^Q$ is a transition function, and $F_-$ and $F_+$ are sets  of accepting states.
Intuitively, the symbols $\rightarrow$
and $\leftarrow$ are used to denote forward and backward  moves.
 A run of $\A$
over a pointed word $(w,i)$ is a pair $r=(r_\leftarrow,r_\rightarrow)$ such that $r_\rightarrow=q_i,q_{i+1}\ldots$ is an infinite sequence of states, $r_\leftarrow=p_i,p_{i-1}\ldots p_0 p_{-1}$ is a finite sequence of states, and: (i) $q_i=p_i\in Q_0$;
(ii) for each  $h\geq i$, $q_{h+1}\in \rho(q_h,\rightarrow,w(h))$; and (iii)
 for each $h\in [0,i]$, $p_{h-1}\in \rho(p_{h},\leftarrow,w(h))$.

Thus, starting from the initial position $i$ in the input pointed word $(w,i)$, the automaton splits in two copies: the first one moves forwardly along the suffix of $w$ starting from position $i$ and the second one moves backwardly along the prefix $w(0)\ldots w(i)$.
The run $r=(r_\leftarrow,r_\rightarrow)$ is \emph{accepting} if $p_{-1}\in F_-$ and $r_\rightarrow$ visits infinitely often some state in $F_+$. A pointed word $(w,i)$ is accepted by $\A$ if there is an accepting run of $\A$ over $(w,i)$. We denote by $\Lang_p(\A)$ the set of pointed words accepted by $\A$ and by $\Lang(\A)$ the set of infinite words $w$ such that $(w,0)\in\Lang_p(\A)$.\vspace{0.2cm}

\noindent \textbf{Two-way $\HAA$. }
For a set $X$, $\PosBool(X)$ denotes the set of positive Boolean
formulas over $X$ built from elements in $X$ using $\vee$ and $\wedge$   (we also allow the formulas $\true$ and $\false$).
 For a formula
$\theta\in\PosBool(X)$, a \emph{model} $Y$ of $\theta$ is a subset $Y$
of $X$ which satisfies $\theta$. The model $Y$ of $\theta$ is minimal
if no strict subset of $Y$ satisfies $\theta$.

A two-way $\HAA$ $\A$  over an alphabet
$\Sigma$ is a tuple $\A=\tpl{ Q,q_0,\delta, F_-,\FamStratum}$, where $Q$ is a
finite set of states, $q_0\in Q$ is the initial state,  $\delta:
Q\times \Sigma \rightarrow \PosBool(\{\rightarrow,\leftarrow\}\times Q)$ is a transition
function, $F_-\subseteq Q$ is the \emph{backward acceptance condition}, and $\FamStratum$ is a \emph{strata family} encoding a particular kind of parity acceptance condition and imposing some syntactical constraints on the transition function $\delta$.
Before defining  $\FamStratum$,
 we give the notion of run which is independent of $\FamStratum$ and $F_-$. We restrict ourselves to
  \emph{memoryless} runs, in which the behavior of the automaton
  depends only on the current input position and current state. Since
  later we will deal only with parity acceptance conditions,
  memoryless runs are sufficient (see e.g.~\cite{Zielonka98}).\footnote{See references for the Appendix.}
Formally, given a pointed word $(w,i)$ on $\Sigma$
and a state $p\in Q$, a \emph{$(i,p)$-run of $\A$ over $w$} is a
directed graph $\tpl{V,E,v_0}$ with set of vertices $V\subseteq
(\N\cup\{-1\})\times Q$ and initial vertex $v_0=(i,p)$. Intuitively, a vertex $(j,q)$ describes a copy of the automaton which is in state $q$ and reads the $j^{th}$ input position. Additionally, we require that the set of edges $E$ is consistent with the transition function $\delta$. Formally,
 for every vertex $v=(j,q)\in V$ such that $j\geq 0$, there is a \emph{minimal} model
$\{(\dir_1,q_1),\ldots,(\dir_n,q_n)\}$ of $\delta(q,w(j))$ such that the set of successors of $v=(j,q)$ is $\{(j_1,q_1),\ldots,(j_n,q_n)\}$ and for all $k\in [1,n]$, $j_k=j+1$ if $\dir_k= \rightarrow$, and $j_k=j-1$ otherwise.

\noindent An infinite path $\pi$ of a run is \emph{eventually strictly-forward} whenever $\pi$ has
a suffix of the form $(i,q_1),(i+1,q_2),\ldots$ for some $i\geq 0$.

Now, we formally define $\FamStratum$ and give the semantic notion of acceptance.
 $\FamStratum$ is a \emph{strata family} of the
form
$\FamStratum=\{\tpl{\rho_1,Q_1,F_1},\ldots,\tpl{\rho_k,Q_k,F_k}\}$,
where $Q_1,\ldots,Q_k$ is a partition of the set of states $Q$ of $\A$, and for all $i\in [1,k]$,
$\rho_i\in \{-,\Trans,\Bu,\Co\}$ and $F_i\subseteq Q_i$, such that
$F_i=\emptyset$ whenever $\rho_i\in\{\Trans,-\}$. A stratum
$\tpl{\rho_i,Q_i,F_i}$ is called a \emph{negative} stratum if $\rho_i=-$, a
\emph{transient} stratum if $\rho_i=\Trans$, a B\"{u}chi 
stratum
(with B\"{u}chi acceptance condition $F_i$) if $\rho_i=\Bu$, and a
coB\"{u}chi 
stratum (with coB\"{u}chi acceptance condition
$F_i$) if $\rho_i=\Co$.  Additionally, 
there is a
partial order $\leq$ on the sets $Q_1,\ldots,Q_k$ such that the
following holds:

\newcommand{\Req}[1]{\ensuremath{\mathsf{R#1}}\xspace}
\newcommand{\ReqOne}{\Req{1}}
\newcommand{\ReqTwo}{\Req{2}}

\begin{itemize}
\item[\ReqOne.] Moves from states in $Q_i$ lead to states in
  components $Q_j$ such that
  $Q_j\leq Q_i$;  additionally, if $Q_i$ belongs to a transient stratum, there are no
  moves from $Q_i$ leading to $Q_i$.
\item[\ReqTwo.] For all moves $(\dir,q')$ from states  $q\in Q_i$ such that $q'\in Q_i$ as well, the following holds:
 $\dir\in \{\leftarrow\}$ if the
    stratum of $Q_i$   is negative, and $\dir\in
    \{\rightarrow\}$ otherwise.
\end{itemize}
\ReqOne is 
the \emph{stratum order requirement} and it ensures
that every infinite path $\pi$ of a run gets trapped in the component
$Q_i$ of some  stratum.  \ReqTwo is 
the
\emph{eventually syntactical requirement} and it ensures that $Q_i$
belongs to a B\"{u}chi or coB\"{u}chi stratum and that $\pi$ is
eventually strictly-forward.

Now we define when a run is accepting.  Let $\pi$ be an infinite path
 of  a  run,  $\tpl{\rho_i,Q_i,F_i}$ be
the B\"{u}chi or coB\"{u}chi stratum in which $\pi$ gets trapped, and
 $\Inf(\pi)$ be the states from $Q$ that occur infinitely many
times in $\pi$. The path $\pi$ is \emph{accepting} whenever
$\Inf(\pi)\cap F_i\neq \emptyset$ if $\rho_i=\Bu$ and $\Inf(\pi)\cap
F_i= \emptyset$ otherwise (i.e. $\pi$ satisfies the corresponding
B\"{u}chi or coB\"{u}chi requirement).
 A run is
\emph{accepting} if: (i) all its infinite paths are accepting and (ii) for each
vertex $(-1,q)$ reachable from the initial vertex, it holds that $q\in F_-$ (recall that $F_-$ is the backward acceptance condition of $\A$).
 The
\emph{$\omega$-pointed language} $\Lang_{p}(\A)$ of $\A$ is the set
of pointed words $(w,i)$ over $\Sigma$ such that there is an
accepting $(i,q_0)$-run of $\A$ on $w$.

The \emph{dual automaton} $\dual{\A}$ of a two-way $\HAA$ $\A=\tpl{Q,q_0,\delta,F_-,\FamStratum}$ is defined
as $\dual{\A}=\tpl{q,q_0,\dual{\delta},Q\setminus F_-,\dual{\FamStratum}}$, where
$\dual{\delta}(q,\sigma)$ is the dual formula of $\delta(q,\sigma)$ (obtained from $\delta(q,\sigma)$ by
 switching $\vee$ and $\wedge$, and switching $\true$ and $\false$),
and $\dual{\FamStratum}$ is obtained from $\FamStratum$ by converting
a B\"{u}chi stratum $\tpl{\Bu,Q_i,F_i}$ into the coB\"{u}chi stratum
$\tpl{\Co,Q_i,F_i}$ and a coB\"{u}chi stratum $\tpl{\Co,Q_i,F_i}$ into
the B\"{u}chi stratum $\tpl{\Bu,Q_i,F_i}$.  By construction the dual automaton $\dual{\A}$ of $\A$ is
still a two-way $\HAA$. Following standard arguments (see
e.g. \cite{Zielonka98}),   the dual automaton $\dual{\A}$ of a two-way $\HAA$ $\A$ is a two-way $\HAA$ accepting the complement of
 $\Lang_p(\A)$.

\subsection{Proof of Theorem~\ref{theorem:FromHAAtoSNPA}}\label{APP:FromHAAtoSNPA}

In this section we give the details of the translation from two-way $\HAA$ into
B\"{u}chi $\SNWA$ as captured by Theorem~\ref{theorem:FromHAAtoSNPA} (see Appendix~\ref{APP:SNWAandHAA} for a formal definition of B\"{u}chi $\SNWA$ and two-way $\HAA$).
 The proposed  construction is based on a preliminary result.  By using the notion
of odd ranking function for standard coB\"{u}chi alternating
 automata \cite{KupfermanV01} \footnote{See references for the Appendix.}(which intuitively, allows to
convert a coB\"{u}chi acceptance condition into a B\"{u}chi-like
acceptance condition) and a non-trivial generalization of the Miyano-Hayashi
construction \cite{MiyanoH84}, we give a characterization of the pointed words
in $\Lang_p(\A)$ in terms of infinite sequences of finite sets (called
\emph{regions})  satisfying determined requirements which can be easily
checked by  B\"{u}chi $\SNWA$.

 Fix a two-way  $\HAA$
$\A=\tpl{Q,q_{0},\delta,F_-,\FamStratum}$ over an alphabet $\Sigma$.
First, as anticipated above, we give a characterization of the
fulfillment of the acceptance condition for a   coB\"{u}chi
stratum along a run in terms of the existence of an \emph{odd ranking
  function}.

\begin{definition}
  \label{definition:Ranking}
  Let $\Stratum=\tpl{\Co,P,F}$ be a   coB\"{u}chi stratum
  of $\A$ and $n=|P|$ (\emph{the size of the stratum}). For an
  infinite word $w$ on $\Sigma$ and a run $G=\tpl{V,E,v_0}$ of $\A$
  over $w$, a \emph{ranking function of the stratum $\Stratum$ for
    the run $G$} is a function $f_\Stratum: V \rightarrow \{1,\ldots,2n\}$
  satisfying the following:
  \begin{compactenum}
  \item for all $(j,q)\in V$ such that $q\in F$, $f_\Stratum(j,q)$ is even;
  \item for all $(j,q),(j',q')\in V$ such that $(j',q')$ is a successor
    of $(j,q)$ in $G$ and $q,q'\in P$, it holds that
    $f_\Stratum(j',q')\leq f_\Stratum(j,q)$.
  \end{compactenum}
\end{definition}
Thus, since the image of $f_\Stratum$ is bounded, for every infinite path
$\pi=v_0,v_1,\ldots$ of $G$ that get trapped in  the  coB\"{u}chi stratum $\Stratum$, $f_\Stratum$ converges to
a value: there is a number $l$ such that $f_\Stratum(v_{l'})=f_\Stratum(v_l)$ for all
$l'\geq l$. We say that $f_\Stratum$ is \emph{odd} if for all such infinite
paths $\pi$ of $G$, $f_\Stratum$ converges to an odd value (or, equivalently,
any of such paths $\pi$ visits infinitely many times vertices $v$ such
that $f_\Stratum(v)$ is odd). Note that if $f_\Stratum$ is odd, then $\pi$ is
accepting. The following lemma whose proof is a straightforward generalization of the
results in~\cite{KupfermanV01} (regarding coB\"{u}chi alternating finite-state automata),
asserts that the existence of
an odd ranking function is also a necessary condition for a run to be
accepting.

\begin{lemma}
  \label{lemma:Ranking}
  Let $G$ be a run of $\A$ over an infinite word $w$. $G$ is
  accepting \emph{iff}
  \begin{compactenum}
  \item for every  co-B\"{u}chi stratum
    $\Stratum=\tpl{\Co,P,F}$, there is an odd ranking function of $\Stratum$
    for the run $G$;
  \item every infinite path of $G$ which get trapped in the component
    of a B\"{u}chi stratum $\Stratum=\tpl{\Bu,P,F}$  satisfies the
    B\"{u}chi acceptance condition $F$;
  \item for each vertex $(-1,q)$ reachable from the initial vertex, $q\in F_-$ (recall that $F_-$ is the backward acceptance condition of $\A$).
  \end{compactenum}
\end{lemma}

Now, based on Lemma~\ref{lemma:Ranking} and the classical breakpoint
construction \cite{MiyanoH84},\footnote{See references for the Appendix.} we give a characterization of the pointed words
$(w,\ell)\in\Lang_p(\A)$ in terms of infinite sequences of finite sets (called
\emph{regions}) satisfying determined requirements which can be easily
checked by B\"{u}chi $\SNWA$.  A coB\"{u}chi state is a state of $\A$ belonging to some  coB\"{u}chi stratum of $\A$. For a coB\"{u}chi state $q$,
\emph{a rank of $q$} is a natural number in
$\{1,\ldots,2n\}$, where $n$ is the size of the stratum of $q$. A
\emph{region of $\A$} is a triple $(R,O,f)$, where $R\subseteq Q$ is
a set of states, $O\subseteq R$, and $f$ is a mapping assigning to
each coB\"{u}chi state $q\in R$ a rank of $q$ such that $f(q)$ is even
if $q\in F$, where $\tpl{\Co,P,F}$ is the  coB\"{u}chi
stratum of $q$. A state $q$ of $\A$ is \emph{accepting} with respect
to $f$ if (1) either $q$ is an accepting state of a B\"{u}chi stratum,
or (2) $q$ is a coB\"{u}chi state and, additionally, $f(q)$
is odd if $q\in R$. The \emph{stop region} is the region
 $(F_-,\emptyset,f)$
where $f: Q \mapsto \{1\}$ and $F_-$ is the backward acceptance condition of $\A$.

For an atom $(\dir,q)$ of $\A$  and a position $i\geq 0$, the \emph{effect of} (the move) \emph{$(\dir,q)$ w.r.t. $i$} is the pair $(j,q)$, where $j=i+1$ if $\dir =\rightarrow$, and $j=i-1$ otherwise.

\noindent Let $(w,\ell)$ be a pointed word over $\Sigma$ and
$\nu=(R_0,O_0,f_0),(R_1,O_1,f_1),\ldots$ be  an infinite sequence of regions.
We say that $\nu$
 is \emph{good with respect to
  $(w,\ell)$} if for all $i\geq 0$, there is a mapping $g_i$ assigning to
each $q\in R_i$ a minimal model of $\delta(q,w(i))$ such that the
following holds, where $\Acc_i$ denotes the set of accepting states of
$\A$ with respect to $f_i$, and $(R_{-1},\emptyset,f_{-1})$ is the stop region:
\begin{itemize}
\item \emph{Initialization.} $q_0\in R_\ell$.
\item \emph{$\delta$-consistency w.r.t. $g_i$.} For all $q\in R_i$ and $(\dir,p)\in g_i(q)$, let $(h,p)$ be the effect of
$(\dir,p)$ w.r.t.~$i$; then, $p\in R_{h}$.
Additionally, if $q$ and $p$ are coB\"{u}chi states belonging to the same stratum, then $f_{h}(p)\leq
  f_i(q)$ (\emph{ranking requirement w.r.t. $g_i$}).
\item \emph{Miyano-Hayashi requirement w.r.t. $g_i$.}
For all $q\in O_i$ and $(\rightarrow,p)\in g_i(q)$ such that $p\in R_{i+1}\setminus \Acc_{i+1}$,
  it holds that $p\in O_{i+1}$.
\end{itemize}

 The infinite sequence of
regions $\nu$ is \emph{accepting} iff there are infinitely many
 positions $i\geq 0$  such that $O_{i}=\emptyset$ and
$O_{i+1}=R_{i+1}\setminus \Acc_{i+1}$ (\emph{acceptance requirement}).

Intuitively, the infinite sequence of regions $\nu$
represents a graph $G=\tpl{V\subseteq (\N\cup \{-1\})\times Q,E,v_0}$ where for
all input positions $i\geq 0$, $R_i$ is the set of vertices of $G$
associated with position $i$. The initialization and
$\delta$-consistency requirement ensure that $G$ is a $(\ell,q_0)$-run of
$\A$ over $w$ and for each  vertex $(-1,q)$ reachable from the initial vertex, $q\in F_-$.
Additionally, the ranking requirement ensures that for each
non-trivial coB\"{u}chi stratum $\Stratum$, there is a ranking
function $f_{\Stratum}$ of $\Stratum$ for the run $G$. By
Lemma~\ref{lemma:Ranking}, the run is accepting if $f_{\Stratum}$
is odd and Condition~2 in Lemma~\ref{lemma:Ranking} holds. This, in
turn, is equivalent to require that every infinite path of $G$ visits
infinitely many vertices in $\Acc$, where $\Acc$ is the set of
$G$-vertices $(i,q)$ such that $q\in\Acc_i$.
This condition is captured by the Miyano-Hayashi and the acceptance requirements on the sets $O_i$.
Formally, the following holds.

\begin{lemma}[Characterization lemma for $\HAA$]
  \label{lemma:characterizationForHAA}
  $(w,\ell)\in\Lang_p(\A)$ iff there is an accepting infinite sequence of
  regions which is good with respect to $(w,\ell)$.
\end{lemma}
\begin{proof} $\Leftarrow$) First, we prove the if direction. Assume
  that there is an accepting infinite sequence of regions
  $\nu=(R_0,O_0,f_0),(R_1,O_1,f_1),\ldots$ which is good with respect
  to the pointed word $(w,\ell)$. We need to show that $(w,\ell)\in\Lang_p(\A)$. For all $i\geq 0$,
  let $\Acc_i$ be the set of accepting states of $\A$ with respect to
  $f_i$, and $g_i$ be the mapping assigning to each $q\in R_i$ a
  minimal model of $\delta(q,w(i))$ such that $\nu$ satisfies the
  $\delta$-consistency requirement, the ranking requirement, and the
  Miyano-Hayashi requirement w.r.t.~$g_i$. Let $P_s$ be the set of states $p\in Q$ such that for some  $q\in R_0$,  $(\leftarrow,p)\in g_0(q)$. Note that the $\delta$-consistency requirement ensures that $P_s$ contains only states belonging to the backward acceptance condition $F_-$ of $\A$.
    We define a graph
  $G=\tpl{V,E,v_0}$ and show that it is an accepting $(\ell,q_0)$-run of
  $\A$ over $w$. The graph $G$ is defined as follows:

  \begin{itemize}
  \item $v_0=(\ell,q_0)$, $V\subseteq (\N\cup \{-1\})\times Q$ such that:
  (i) $(-1,q)\in V$ iff $q\in P_s$ and (ii)
 for all $i\geq 0$, $(i,q)\in V$ iff $q\in R_i$;
  \item there is an edge from $(i,q)$ to $(j,p)$ iff $i\geq 0$ and for
    some  $(\dir,p)\in g_i(q)$, $(j,p)$ is the effect of $(\dir,p)$ w.r.t.~$i$.
  \end{itemize}

  Since the sequence of regions $\nu$ satisfies the initialization
  requirement and the $\delta$-consistency requirement w.r.t.~$g_i$
  for all $i\geq 0$, $G$ is a $(\ell,q_0)$-run of $\A$ over $w$. It
  remains to show that $G$ is accepting. We assume the contrary and
  derive a contradiction. Then, since $\A$ is a two-way $\HAA$ and the acceptance condition for the vertices $(-1,q)$ is satisfied, there must be a strictly-forward infinite path $\pi=(i,q_{i}),(i+1,q_{i+1}),\ldots$ of $G$ for some $i\geq 0$ such that the
  following holds:

  \begin{itemize}
  \item for some B\"{u}chi stratum $\tpl{\Bu,P,F}$, $q_h\in P\setminus
    F$, for all $h\geq i$. Since $q_h\in R_{h}$, we obtain that
    $q_h\in R_{h}\setminus\Acc_{h}$ for all $h\geq i$.
  \item for some coB\"{u}chi stratum $\tpl{\Co,P,F}$,
    $q_h\in P$ for all $h\geq i$, and for infinitely many $k\geq
    i$, $q_k\in F$.
    Since $\nu$ satisfies the ranking requirement w.r.t. $g_{h}$,
    $f_{h+1}(q_{h+1})\leq f_{h}(q_{h})$ for all $h\geq i$. It
    follows that there is $k\geq i$ such that $q_k\in F$ and for all
    $h\geq k$, $f_{h}(q_{h})= f_{k}(q_{k})$. In particular,
    $f_{h}(q_{h})$ is even. Hence, for all $h\geq k$, $q_h\in
    R_{h}\setminus \Acc_{h}$.
  \end{itemize}

  Thus, we obtain that there is an infinite strictly forward path
  $\pi=(k,q_k),(k+1,q_{k+1}),\ldots$ of $G$ such that $q_h\in R_{h}\setminus \Acc_{h}$ for all $h\geq
  k$. Since the sequence of regions $\nu$
  is accepting, there must be $i\geq k$ such that $O_{i}=R_{i}\setminus
  \Acc_{i}\neq \emptyset$. Moreover, since $\nu$ satisfies the
  Miyano-Hayashi requirement w.r.t. the mappings $g_{j}$, we deduce
  that $O_{j}\neq\emptyset$ for all $j>i$, which contradicts the assumption that the sequence of regions $\nu$ is accepting.

   \bigskip

  $\Rightarrow)$ Now, we prove the only if direction. Let $(w,\ell)\in
  \Lang_p(\A)$. Hence, there is an accepting $(\ell,q_0)$-run
  $G=\tpl{V,E,v_0}$ of $\A$ over $w$. By Lemma~\ref{lemma:Ranking},
  for every  coB\"{u}chi stratum $\Stratum$ of $\A$,
  there is an odd ranking function $f_{\Stratum}$ of $\Stratum$
  for the run $G$.  Let $\Acc$ be the set of vertices $(i,q)$ of the
  run $G$ such that (1) either $q$ is an accepting state of a
  B\"{u}chi stratum, or (2) $q$ belongs to a coB\"{u}chi stratum
  $\Stratum$ and $f_{\Stratum}(i,q)$ is odd. Since $G$ is accepting and every infinite path of
  $G$ gets eventually trapped either in a B\"{u}chi stratum or a
  coB\"{u}chi stratum, it holds that every infinite path of $G$ visits
  infinitely many times vertices in $\Acc$. We define an infinite
  sequence of regions $\nu=(R_0,O_0,f_0),(R_1,O_1,f_1),\ldots$ and
  show that it is accepting and good with respect to the pointed word $(w,\ell)$, hence, the
  result follows.  For all $i\geq 0$, $R_i$ and $f_i$ are defined as
  follows:

  \begin{itemize}
  \item $R_i := \{(i,q)\mid (i,q)\in V\}$;
  \item for all  coB\"{u}chi strata
    $\Stratum=\tpl{\Co,P,F}$ and $q\in R_i\cap P$,
    $f_i(q)=f_{\Stratum}(i,q)$.
  \end{itemize}

Note that since $q_0\in R_\ell$, the sequence $\nu$
  (independently of the form of the sets $O_i$) satisfies the
  initialization requirement (w.r.t. $(w,\ell)$).
Let $\Acc_i$ be the set of the accepting states of $\A$ with
  respect to $f_i$. Note that for all $q\in R_i$,  $q\in \Acc_i$ iff
  $(i,q)\in \Acc$.
  Since $G$ is a run over $w$, for all $i\geq 1$, there must be a
  mapping $g_i$ over $R_i$ such that for all $q\in R_i$, $g_i(q)$ is a
  minimal model of $\delta(q,w(i))$ and the sequence $\nu$
  (independently of the form of the sets $O_i$) satisfies the
  $\delta$-consistency requirement w.r.t. $g_i$.  Moreover, since for
  every  coB\"{u}chi stratum $\Stratum$ of $\A$,
  $f_{\Stratum}$ is an odd ranking function of $\Stratum$ for
  the run $G$, the sequence $\nu$ (independently of the form of the
  sets $O_i$) satisfies the ranking requirement w.r.t. $g_i$.   It remains to define the sets $O_i$ and show that
  the resulting sequence is accepting and satisfies the   Miyano-Hayashi requirement as well.  For this,
  we use the following claim.\vspace{0.2cm}

\noindent  \emph{Claim:}  there is an infinite
  sequence $0=h_1<h_2<\ldots$ of positions of $w$ such that
  for all $j\geq 0$ and finite paths of $G$ of the form $\pi=
  (h_{j},p), \ldots,(h_{j+1}-1,q)$, $\pi$ visits some state in $\Acc$.

   \bigskip

  First, we show that the result follows from the claim above and then
  we prove the claim. So, let $0=h_1<h_2<\ldots$ be an infinite sequence
  of  positions of $w$ satisfying the claim above. For
  every $i\geq 0$, let $j\geq 0$ be the unique natural number such that
  $h_{j}\leq i< h_{j+1}$. Then, $O_i$ is defined as
  follows:
  \begin{itemize}
  \item $O_i$ is the set of states $q$ such that there is a finite
    path of $G$ of the form $\pi=(h_{j},p),\ldots,(i,q)$ which does
    \emph{not} visit vertices in $\Acc$.
  \end{itemize}
  Note that $O_i\cap \Acc_i=\emptyset$ and $O_i\subseteq R_i$.
  By construction and the claim above, we have that
  for all $j> 0$, $O_{h_j-1}=\emptyset$
  and $O_{h_{j}}=R_{h_{j}}\setminus \Acc_{h_{j}}$. Hence, the
  infinite sequence of regions
  $\nu=(R_0,O_0,f_0),(R_1,O_1,f_1),\ldots$ is accepting.  For the
  Miyano-Hayashi requirement w.r.t. $g_i$, let $q\in O_i$ and $(\rightarrow,p)\in g_i(q)$  such that  $p\notin \Acc_{i+1}$ (hence, $(i+1,p)\notin
  \Acc$). We need to show that $p\in O_{i+1}$. Let $j\geq 0$ such that
  $h_j\leq i< h_{j+1}$. Since $O_i\neq \emptyset$ and $O_{h_{j+1}-1}= \emptyset$, we have that $i<h_{j+1}-1$. Hence, $i+1 <
  h_{j+1}$. Thus, since $q\in O_i$ and $(i+1,p)$ is a successor of
  $(i,q)$ in $G$ which is not in $\Acc$, we obtain that $p\in
  O_{i+1}$. Therefore, $\nu=(R_0,O_0,f_0),(R_1,O_1,f_1),\ldots$ is an accepting
  infinite sequence of regions which is good w.r.t. the pointed word $(w,\ell)$. It remains to
  prove the claim.\vspace{0.2cm}

\noindent  \emph{Proof of the claim:}
  fix $k\geq 0$.  For each $i\geq 0$, let $T_i$ be the set of states $q\in Q$ such
  that there is a finite path of $G$ of the form $(k,p),\ldots,(i,q)$
  which does \emph{not} visit $\Acc$-vertices. Since $k$ is arbitrary,
  in order to prove the claim, it suffices to show that there is a
  position $m>k$ such that $T_{m-1}=\emptyset$. Let $H=\{(i,q)\in
  \N\times Q\mid q\in T_i\}$.  Note that $H\cap \Acc=\emptyset$.
  First, we prove that the set $H$ is finite. We assume the contrary and
  derive a contradiction. Let $G_H$ be the subgraph of $G$ given by the
  restriction of $G$ to the set of vertices $H$.  Note that by
  construction, every vertex in $G_H$ is reachable in $G_H$ from a
  vertex of the form $(k,p)$. Moreover, each vertex of $G_H$ has only
  finitely many successors. Since $G_H$ is infinite and the set of
  vertices of the form $(k,p)$ is finite, by K\"{o}nig's Lemma, $G_H$
  contains an infinite path $\pi$. This is a contradiction since $\pi$
  does not visit vertices in $\Acc$ and $\pi$ is also an infinite path
  of $G$. Thus, the set $H=\{(i,q)\in \N\times Q\mid q\in T_i\}$ is
  finite. It follows that there is $j\geq 0$ such that for all $i\geq
  j$, $T_j=\emptyset$. Hence, the result follows, which concludes the proof
  of the claim and the lemma as well.
\end{proof}

Now, we can prove Theorem~\ref{theorem:FromHAAtoSNPA}.

\setcounter{aux}{\value{theorem}}
\setcounter{theorem}{\value{theo-FromHAAtoSNPA}}

\begin{theorem} For a two-way $\HAA$ $\A$ with $n$ states, one can construct ``on the fly'' and in singly
  exponential time a B\"{u}chi $\SNWA$  accepting $\Lang_p(\A)$ with
  $2^{O(n\cdot \log(n))}$ states.
\end{theorem}
\setcounter{theorem}{\value{aux}}
\begin{proof}
  For the fixed two-way $\HAA$ $\A=\tpl{Q,q_{0},\delta,F_-,\FamStratum}$ over $\Sigma$, we construct a
  B\"{u}chi $\SNWA$ $\A_N=\tpl{P,P_0,\rho,F'_-,F_+}$ over $\Sigma$ accepting
  $\Lang_p(\A)$ with $2^{O(|Q|\cdot \log(|Q|))}$ states.    We construct the
  B\"{u}chi $\SNWA$ $\A_N$ in such a way that given a pointed word $(w,i)$ over $\Sigma$,  $\A_N$ accepts $(w,i)$ iff there is an
  \emph{accepting} infinite sequence of regions of $\A$ which is good
  w.r.t. $(w,i)$.
  At each step, the forward (resp., backward) copy of the automaton keeps tracks in its control
  state of the guessed region associated with the current input
  position  and the guessed region associated with the previous (resp., next) input
  position.  Note that in this way, the automaton
  can check locally (i.e., by its transition function) that the
  guessed infinite sequence of regions satisfies the
  $\delta$-consistency requirement, the ranking requirement, and the
  Miyano-Hayashi requirement. Finally, the B\"{u}chi acceptance
  condition of $\A_N$ is used to check that the guessed sequence
  of regions is accepting.

  In order to simplify the formal definition
  of $\A_N$, we introduce additional notation.  For a region
  $\Region=(R,O,f)$ and $\sigma\in\Sigma$, a
  \emph{$(\Region,\sigma)$-model} is a mapping assigning to each $q\in
  R$, a minimal model of $\delta(q,\sigma)$. For a direction
  $\dir\in\{\rightarrow,\leftarrow\}$, two
  regions $\Region=(R,O,f)$ and $\Region_\dir=(R_\dir,O_\dir,f_\dir)$,
  and a $(\Region,\sigma)$-model $g$ for some $\sigma\in \Sigma$, we
  say that \emph{$\Region$ is $\dir$-consistent
    w.r.t. $g$ and $\Region_\dir$} if the following holds:
\begin{itemize}
\item \emph{$\delta$-consistency requirement.} For all $q\in R$ and $(\dir,p)\in g(q)$, $p\in
  R_{\dir}$. If, additionally, $p$ and $q$ are coB\"{u}chi states belonging to the same stratum, then
  $f_{\dir}(p)\leq f(q)$ (\emph{Ranking requirement}).
\item \emph{Miyano-Hayashi requirement.} If $\dir = \rightarrow$, then for all
 $q\in O$ and
  $(\dir,p)\in g(q)$, whenever  $p$ is \emph{not} accepting w.r.t. $f_\dir$, then $p\in O_\dir$.
\end{itemize}

Formally, the B\"{u}chi $\SNWA$ $\A_N=\tpl{P,P_0,\rho,F'_-,F_+}$ is
defined as follows:
\begin{itemize}
\item $P= (\Regions\times
  \Regions)\cup (\textit{in}\times\Regions\times
  \Regions)\cup \{\textit{stop}\}$, where $\Regions$ is the set of regions.
\item $P_0$ is the set of states of the form $(\textit{in},\Region,(R,O,f))$
  such that $q_0\in R$.
\item the transition function $\rho$ is defined as follows, where $\Region_s$ is the stop region:
\begin{compactitem}
\item \emph{Forward transitions:} $p'\in \rho(p,\sigma,\rightarrow)$ iff
 (\emph{either} $p=(\Region_-,\Region)$ \emph{or} $p=(\textit{in},\Region_-,\Region)$), $p'= (\Region,\Region_+)$ and there is a $(\Region,\sigma)$-model $g$ such that
  $\Region$ is $\rightarrow$-consistent  w.r.t. $g$ and
    $\Region_+$ and $\leftarrow$-consistent  w.r.t. $g$ and
    $\Region_-$.
 \item \emph{Backward transitions:} $p'\in \rho(p,\sigma,\leftarrow)$ iff one of the following holds:
 \begin{compactitem}
   \item $p'=\textit{stop}$, and either $p=(\textit{in},\Region_s,\Region)$ or $p=(\Region_s,\Region)$;
   \item $p=(\textit{in},\Region,\Region_+)$ and $p'=(\Region,\Region_+)$;
   \item  $p=(\Region,\Region_+)$, $p'= (\Region_-,\Region)$ and there is a $(\Region,\sigma)$-model $g$ such that
  $\Region$ is $\rightarrow$-consistent  w.r.t. $g$ and
    $\Region_+$ and $\leftarrow$-consistent  w.r.t. $g$ and
    $\Region_-$.
 \end{compactitem}
 \end{compactitem}
\item $F'_-=\{\textit{stop}\}$.
\item $F_+$  consists of the states  of the form
  $((R_-,\emptyset,f_-),(R,O,f))$ such that $O=R\setminus\Acc$,
  where $\Acc$ is the set of accepting states of $\A$ w.r.t. $f$.
\end{itemize}

By construction, it easily follows that $(w,i)\in\Lang_p(\A_N)$ \emph{iff}
there is an accepting infinite sequence of regions which is good w.r.t.~$(w,i)$. By
Lemma~\ref{lemma:characterizationForHAA}, it follows that
$\Lang_p(\A_N)=\Lang_p(\A)$.  Since the number of regions is at most
$2^{2|Q|} \cdot 2^{|Q|\cdot\log(2|Q|)}$,
Theorem~\ref{theorem:FromHAAtoSNPA} follows.
\end{proof}

\subsection{Proof of Theorem~\ref{theorem:FromQPTLtoSNWA}}\label{APP:FromQPTLtoSNWA}

In this Subsection we provide a proof Theorem~\ref{theorem:FromQPTLtoSNWA} (see Appendix~\ref{APP:SNWAandHAA} for a formal definition of B\"{u}chi $\SNWA$ and two-way $\HAA$). We will use the following trivial result.

\begin{proposition}\label{remark:FromSNWAtoHAA} A B\"{u}chi $\SNWA$ $\A$ can be  converted ``on the fly'' in linear time into a two-way $\HAA$ accepting
$\Lang_p(\A)$.
\end{proposition}

\setcounter{aux}{\value{theorem}}
\setcounter{theorem}{\value{theo-FromQPTLtoSNWA}}

\begin{theorem}  Let $\varphi$ be a first-level existential (resp., first-level universal) $\QPTL$ formula  and $h=\sad(\varphi)$. Then, one can construct ``on the fly'' a B\"{u}chi $\SNWA$ $\A_\varphi$ accepting $\Lang_p(\varphi)$   in time $\Tower(h,O(|\varphi|))$ (resp., $\Tower(h+1,O(|\varphi|))$).
\end{theorem}
\setcounter{theorem}{\value{aux}}
\begin{proof}
The proof is by induction on $|\varphi|$. The base case $|\varphi|=1$ is trivial. Now, assume that $|\varphi|>1$. We distinguish four cases depending on the type of root operator of $\varphi$ (either temporal modality, or existential quantifier, or universal quantifier, or  boolean connective). \vspace{0.2cm}

\noindent \textbf{Case 1:} the root operator of $\varphi$ is a temporal modality.
Let $h=\sad(\varphi)$ and
\[
P := \{\exists p_1.\, \theta_1,\ldots,\exists p_n.\, \theta_n,\forall q_1.\,\xi_1,\ldots,\forall q_k.\,\xi_k\}
\]
be the set of quantified subformulas of $\varphi$  which do not occur in the scope of a quantifier. If $P=\emptyset$, then
$\varphi$ is a $\PLTL$ formula. In this case, by a straightforward adaptation of the standard translation of $\LTL$ into B\"{u}chi word automata \cite{Var88},\footnote{See references for the Appendix.} one can construct a B\"{u}chi $\SNWA$  of size  $2^{O(|\varphi|)}$ accepting $\Lang_p(\varphi)$. Hence, the result follows.

Assume now that $P\neq \emptyset$. Then, $\varphi$ can be viewed as a $\PLTL$ formula in positive normal form, written $\PLTL(\varphi)$, over the set of atomic proposition given by $P$.

We first, assume that for each $\psi\in P$, $\sad(\psi)<\sad(\varphi)$. Hence, for all $\psi\in P$,
$\sad(\psi)\leq h-1$ and $h>1$. Moreover, in this case,
$\varphi$ must be a first-level existential formula.
 For all $1\leq j\leq k$, let $\widetilde{\xi}_j$ be the positive normal form of
$\neg\xi_j$. Note that $\sad(\forall q_i.\xi_i)=\sad(\neg\exists q_i.\widetilde{\xi}_i)$ and
$\Lang_p(\forall q_i.\xi_i)=\Lang_p(\neg\exists q_i.\widetilde{\xi}_i)$. Thus,    by applying the induction hypothesis, Proposition \ref{remark:FromSNWAtoHAA} and the complementation lemma for two-way $\HAA$, it follows that for each $\psi\in P$, one can construct ``on the fly'' in time at most $\Tower(h-1,O(|\varphi|))$, a two-way $\HAA$ $\A_\psi$ accepting $\Lang_p(\psi)$. Then, by an easy generalization of the standard linear-time translation of $\LTL$ formulas into B\"{u}chi alternating word automata and by using the two-way $\HAA$ $\A_\psi$ with $\psi\in P$, one can construct ``on the fly'',
 in time $\Tower(h-1,O(|\varphi|))$, a two-way $\HAA$ $\A_\varphi$ accepting $\Lang_p(\varphi)$.
 Intuitively, given an input pointed word, each copy of  $\A_\varphi$ keeps track of the current subformula of $\PLTL(\varphi)$ which needs to be evaluated. The evaluation simulates the semantics of $\PLTL$ (in positive normal form) by using universal and existential branching, but when the current subformula $\psi$ is in $P$, then the current copy of $\A_\varphi$ activates a copy of $\A_\psi$ in the initial state.

 Formally, for each $\psi\in P$, let $\A_\psi=\tpl{Q_\psi,q_\psi,\delta_\psi, F_\psi^{-}, \FamStratum_\psi}$. Without loss of generality, we assume that the state sets of the two-way $\A_\psi$ are pairwise distinct. Then, $\A_\varphi=\tpl{Q,q_0,\delta,F_-,\FamStratum}$, where
 \begin{compactitem}
   \item $Q = \displaystyle{\bigcup_{\psi\in P}Q_\psi}\cup \Sub(\varphi)$, where $\Sub(\varphi)$ is the set of subformulas of $\PLTL(\varphi)$;
   \item $q_0=\varphi$;
   \item The transition function $\delta$ is defined as follows: $\delta(q,\sigma) = \delta_\psi(q,\sigma)$ if $q\in Q_\psi$ for some $\psi\in P$.
   If instead  $ q\in \Sub(\varphi)$, then $\delta(q,\sigma)$ is defined by induction on the structure of  $q$ as follows:
   \begin{itemize}
     \item $\delta(p,\sigma) = \true$ if $p\in\sigma$, and $\delta(p,\sigma) = \false$ otherwise (for all $p\in\AP\cap \Sub(\varphi)$);
     \item $\delta(\neg p,\sigma) = \false$ if $p\in\sigma$, and $\delta(\neg p,\sigma) = \true$ otherwise (for all $p\in\AP\cap \Sub(\varphi)$);
     \item $\delta(\phi_1\wedge \phi_2,\sigma)=\delta(\phi_1,\sigma)\wedge\delta(\phi_2,\sigma)$ and $\delta(\phi_1\vee \phi_2,\sigma)=\delta(\phi_1,\sigma)\vee\delta(\phi_2,\sigma)$;
     \item $\delta(\Next\phi,\sigma)=(\rightarrow,\phi)$ and $\delta(\PNext\phi,\sigma)=(\leftarrow,\phi)$;
      \item $\delta(\phi_1\Until \phi_2,\sigma)=\delta(\phi_2,\sigma)\vee (\delta(\phi_1,\sigma)\wedge (\rightarrow,\phi_1\Until \phi_2))$;
       \item $\delta(\phi_1\PUntil \phi_2,\sigma)=\delta(\phi_2,\sigma)\vee (\delta(\phi_1,\sigma)\wedge (\leftarrow,\phi_1\PUntil \phi_2))$;
       \item $\delta(\phi_1\Release \phi_2,\sigma)=\delta(\phi_2,\sigma)\wedge (\delta(\phi_1,\sigma)\vee (\rightarrow,\phi_1\Release \phi_2))$;
           \item $\delta(\phi_1\PRelease \phi_2,\sigma)=\delta(\phi_2,\sigma)\wedge (\delta(\phi_1,\sigma)\vee (\leftarrow,\phi_1\PRelease \phi_2))$;
      \item for each $\psi\in P$, $\delta(\psi,\sigma)=\delta(q_\psi, \sigma)$.
   \end{itemize}
   \item $F_-=\displaystyle{\bigcup_{\psi\in P}}F_\psi^{-}$
   \item $\FamStratum = \displaystyle{\bigcup_{\psi\in P}\FamStratum_\psi}\vee \bigcup_{\phi\in \Sub(\varphi)}\mathcal{S}_\phi$, where for each
   $\phi\in \Sub(\varphi)$, $\mathcal{S}_\phi$ is defined as follows:
   \begin{compactitem}
     \item if $\phi$ has as root a past temporal modality, then $\mathcal{S}_\phi$ is the negative stratum $(\{\phi\},-,\emptyset)$;
     \item if $\phi$ has as root the (future) until modality, then $\mathcal{S}_\phi$ is the B\"{u}chi stratum $(\{\phi\},\Bu,\emptyset)$;
          \item if $\phi$ has as root the (future) release modality, then $\mathcal{S}_\phi$ is the coB\"{u}chi stratum              $(\{\phi\},\Co,\emptyset)$;
     \item otherwise, $\mathcal{S}_\phi$ is the transient stratum given by $(\{\phi\},\Trans,\emptyset)$.
   \end{compactitem}
 \end{compactitem}
 Finally, since $h>1$ and the size of  the two-way $\HAA$ $\A_\varphi$ is $\Tower(h-1,O(|\varphi|))$,
by applying Theorem~\ref{theorem:FromHAAtoSNPA}, one can construct ``on the fly'' a B\"{u}chi $\SNWA$ accepting $\Lang_p(\varphi)$ of size
$\Tower(h,O(|\varphi|))$. Hence, the result follows.\vspace{0.2cm}

Now, assume that for some $\psi\in P$, $\sad(\psi)= \sad(\varphi)$. Let $h=\sad(\varphi)$. There are two cases:
\begin{compactitem}
  \item $\psi=\exists p.\, \psi'$. Since the root of $\varphi$ is a temporal modality, by definition of strong alternation depth, either $\varphi=\PEventually \varphi'$ or $\varphi=\Eventually \varphi'$ (and $\psi$ is a subformula of $\varphi'$). Moreover, $\varphi$ and $\varphi'$ must be  first-level existential formulas. Hence, by applying the induction hypothesis, the result directly follows from the following claim.\vspace{0.2cm}

      \noindent \emph{Claim.} Given a B\"{u}chi  $\SNWA$ $\A$, one can construct ``on the fly'' and in linear time two B\"{u}chi $\SNWA$ $\A_+$
      and $\A_-$ such that
      \begin{compactitem}
        \item $\Lang_p(\A_+) = \{(w,i)\mid \text{ for some }j\geq i\,\,(w,j)\in \Lang_p(\A)\}$;
        \item $\Lang_p(\A_-) = \{(w,i)\mid \text{ for some }j\leq i\,\,(w,j)\in \Lang_p(\A)\}$.
      \end{compactitem}
      \vspace{0.2cm}

      \noindent \emph{Proof of the Claim.} We illustrate the construction of $\A_+$ (the construction of $\A_-$ being similar).
      Intuitively, given an input pointed word $(w,i)$, $\A_+$ guesses a position $j\geq i$ and checks that $(w,j)\in\Lang_p(\A)$ as follows.
      Initially, $\A_+$ keeps track of a guessed state $q$ of $\A$ which represents the state where the backward copy of $\A$ would be on reading the
      $i^{th}$ position of $w$ in some guessed accepting run of $\A$ over $(w,j)$. If $j=i$, then $q$ needs to be some initial state of $\A$, and $\A_+$ simply simulates the behavior of $\A$ on $(w,i)$. Otherwise, $\A_+$ splits in two copies: the backward copy simulates the backward copy of $\A$, while the forward copy of $\A_+$ behaves as follows. In the first step, the forward copy of $\A$ moves to the same state $q$, and after this step, such a copy starts to simulate  in forward-mode the backward copy of $\A$ until, possibly, a `switch' occurs at the guessed position $j$, where the forward copy of $\A_+$ simulates in a unique step from the current state some initial split of $\A$ in the backward and forward copy.
      After such a switch (if any), the forward copy of $\A_+$ simply simulates the forward copy of $\A$. We use two flags to distinguish the different phases of the simulation (in particular, the initial phase and the switch phase).

      Formally, let
      $\A=\tpl{Q,Q_0,\rho,F_-,F_+}$. Then, $\A_+=\tpl{Q',Q'_0,\rho',F'_-,F'_+}$, where  $Q'=Q\times \{\bot,\top\}\times \{\IN,\NoIN\}$, $Q'_0=Q\times \{\bot\}\times \{\IN\}$,
      $F'_-=F_-\times \{\top\}\times \{\NoIN\}$, $F'_+=F_+\times\{\top\}\times \{\NoIN\}$, and $\rho'$ is defined as follows:
      \begin{compactitem}
        \item Backward moves: $(q',f'_1,f'_2)\in \rho'((q,f_1,f_2),\leftarrow,\sigma)$ iff $f'_1=\top$, $f'_2=\NoIN$,  and $q'\in \rho(q,\leftarrow,\sigma)$;
        \item Forward moves: $(q',f'_1,f'_2)\in \rho'((q,f_1,f_2),\rightarrow,\sigma)$ iff one of the following holds:
        \begin{compactitem}
          \item $f_2=\IN$, $f'_2=\NoIN$, and \emph{either} $q'=q$ and $f'_1=\bot$, \emph{or} $q\in Q_0$, $q' \in \rho(q,\rightarrow,\sigma)$, and $f'_1=\top$ (\emph{initialization});
           \item $f'_2=f_2=\NoIN$, $f'_1=\bot$, and $q \in \rho(q',\leftarrow,\sigma)$ (\emph{simulation of backward moves});
           \item $f'_2=f_2=\NoIN$, $f_1=\bot$, $f'_1=\top$, and there is $q_0\in Q$ such that $q \in \rho(q_0,\leftarrow,\sigma)$ and
           $q' \in \rho(q_0,\rightarrow,\sigma)$ (\emph{switch});
           \item $f'_2=f_2=\NoIN$, $f'_1=f_1=\top$, and
           $q' \in \rho(q,\rightarrow,\sigma)$ (\emph{simulation of the forward moves of $\A$ after the switch}).
      \end{compactitem}
      \end{compactitem}

  \item $\psi=\forall p.\, \psi'$. Since the root of $\varphi$ is a temporal modality and $\sad(\psi)=\sad(\varphi)=h$, by definition of strong alternation depth, either $\varphi=\PAlways \varphi'$ or $\varphi=\Always \varphi'$ (and $\psi$ is a subformula of $\varphi'$). Moreover, $\varphi$ and $\varphi'$ must be  first-level universal formulas and $\sad(\varphi')=h$. Assume that $\varphi=\Always \varphi'$ (the other case being similar).
      Let $\widetilde{\varphi}'$ be the positive normal form of $\neg\varphi'$.
Note that $\sad(\neg \Eventually \widetilde{\varphi}')=h$ and $\Lang_p(\Eventually \widetilde{\varphi}')=\Lang_p(\neg\varphi)$. Hence, by the previous case, one can construct ``on the fly'' a B\"{u}chi  $\SNWA$ $\A_{\neg\varphi}$
  of size $\Tower(h,O(|\varphi|))$ accepting $\Lang_p(\neg\varphi)$. By Proposition~\ref{remark:FromSNWAtoHAA}, the complementation lemma for two-way $\HAA$ and Theorem~\ref{theorem:FromHAAtoSNPA}, it follows that one
  can construct ``on the fly'' a B\"{u}chi  $\SNWA$ $\A_{\varphi}$
  of size $\Tower(h+1,O(|\varphi|))$ accepting $\Lang_p(\varphi)$. Hence, the result follows.
\end{compactitem}

\bigskip

\noindent \textbf{Case 2:} $\varphi$ is an existential quantified formula of the form $\varphi=\exists p.\,\varphi'$.
Hence, in particular, $\varphi$ is a first-level existential formula.
Let $h=\sad(\varphi)$ and $h'=\sad(\varphi')$.
We observe that like B\"{u}chi nondeterministic automata, $\SNWA$ are  efficiently  closed under projection. In particular, given
a   B\"{u}chi  $\SNWA$ $\A$ over $2^{\AP}$ and $p\in \AP$, one can construct ``on the fly'' and in linear time a B\"{u}chi $\SNWA$
accepting the language $\{w\in (2^{\AP})^{\omega}\mid \text{ there is } w'\in\Lang_p(\A) \text{ such that }w'=_{\AP\setminus \{p\}}w\}$.
Thus, by applying the induction hypothesis, it follows that one can construct ``on the fly'' a B\"{u}chi $\SNWA$ accepting $\Lang_p(\varphi)$ of size
$\Tower(h',O(|\varphi'|))$ if $\varphi'$ is a first-level existential formula, and of size $\Tower(h'+1,O(|\varphi'|))$ otherwise.
Since $h'\leq h$, and $h'=h-1$ if $\varphi'$ is a  first-level universal formula, the result follows.

\bigskip

\noindent \textbf{Case 3:} $\varphi$ is an universal quantified formula of the form $\varphi=\forall p.\,\varphi'$.
Hence, in particular, $\varphi$ is a first-level universal formula.
Let $h=\sad(\varphi)$ and $\widetilde{\varphi}'$ be the positive normal form of $\neg\varphi'$.
Note that $\sad(\neg \exists p.\, \widetilde{\varphi}')=h$ and $\Lang_p(\exists p.\, \widetilde{\varphi}')=\Lang_p(\neg\varphi)$. Hence, by Case~2, one can construct ``on the fly'' a B\"{u}chi  $\SNWA$ $\A_{\neg\varphi}$
  of size $\Tower(h,O(|\varphi|))$ accepting $\Lang_p(\neg\varphi)$. By Proposition~\ref{remark:FromSNWAtoHAA}, the complementation lemma for two-way $\HAA$ and Theorem~\ref{theorem:FromHAAtoSNPA}, it follows that one
  can construct ``on the fly'' a B\"{u}chi  $\SNWA$ $\A_{\varphi}$
  of size $\Tower(h+1,O(|\varphi|))$ accepting $\Lang_p(\varphi)$. Hence, the result follows.

 \bigskip

 \noindent \textbf{Case 4:} $\varphi$ is of the form $\varphi=\varphi_1\wedge\varphi_2$ or $\varphi=\varphi_1\vee \varphi_2$. Assume that
  $\varphi=\varphi_1\wedge\varphi_2$ (the other case being similar). Let $h_1=\sad(\varphi_1)$, $h_2=\sad(\varphi_2)$, and $h=\sad(\varphi)$.
  Note that $h=\max(h_1,h_2)$.
  We use the fact that like B\"{u}chi nondeterministic automata, $\SNWA$ are trivially and efficiently  closed under intersection. In particular, given
two   B\"{u}chi  $\SNWA$ $\A_1$ and $\A_2$, one can construct ``on the fly'' and in time $O(|\A_1||\A_2|)$ a B\"{u}chi $\SNWA$
accepting the language $\Lang_p(\A_1)\cap \Lang_p(\A_2)$. We distinguish two cases:
\begin{compactitem}
  \item $\varphi$ is a first-level existential formula: assume that $h=h_1=h_2$ (the other cases, i.e., when either $h=h_1$ and $h_2<h$, or $h=h_2$ and $h_1<h$, are similar). Hence, both $\varphi_1$ and $\varphi_2$ are existential. Since $h=\max(h_1,h_2)$, by applying the induction hypothesis and the closure of $\SNWA$ under intersection, it follows that one can construct ``on the fly''  a  B\"{u}chi $\SNWA$
accepting the language $\Lang_p(\varphi)$ whose size is at most $\Tower(h_1, O(|\varphi_1|))\cdot \Tower(h_2, O(|\varphi_2|))=\Tower(h, O(|\varphi|))$. Hence, in this case, the result follows.
  \item $\varphi$ is a first-level universal formula: hence, there is $j=1,2$ such that $\varphi_2$ is a first-level universal formula and $h_j=h$.
  Since $h=\max(h_1,h_2)$, by applying the induction hypothesis and the closure of $\SNWA$ under intersection, it follows that one can construct ``on the fly''  a  B\"{u}chi $\SNWA$
accepting the language $\Lang_p(\varphi)$ whose  size is at most $\Tower(h_1+1, O(|\varphi_1|))\cdot \Tower(h_2+1, O(|\varphi_2|))=\Tower(h+1, O(|\varphi|))$. Hence,  the result follows.
\end{compactitem}

This concludes the proof of Theorem~\ref{theorem:FromQPTLtoSNWA}.

\end{proof}

\subsection{Lower bounds in Theorem~\ref{theorem:QPTLsatisfiability}}\label{APP:QPTLsatisfiability}

\newcommand{\Mark}{\textit{mark}}
\newcommand{\code}{\textit{cod}}
\newcommand{\Chec}{\textit{check}}
\newcommand{\Inc}{\textit{inc}}
\newcommand{\Succ}{\textit{suc}}
\newcommand{\First}{\textit{first}}
\newcommand{\Last}{\textit{last}}

\newcommand{\Beg}{\textit{beg}}
\newcommand{\End}{\textit{end}}
\newcommand{\Init}{\textit{init}}
\newcommand{\Fair}{\textit{fair}}
\newcommand{\Conf}{\textit{conf}}
\newcommand{\Tag}{\textit{Tag}}
\newcommand{\NNext}{\textit{next}}
\newcommand{\Mac}{{\cal M}}
\newcommand{\Sym}{\ensuremath{\$}}

For each $h\geq 1$, let $\QPTL^{h}$ be the fragment of $\QPTL$ consisting of formulas whose strong alternation depth is at mose
$h$.
In this section, for all $h\geq 1$, we provide the lower bounds for $\QPTL^{h}$ and the existential fragment of $\QPTL^{h}$ as captured by
Theorem~\ref{theorem:QPTLsatisfiability}. We focus on the existential fragment of $\QPTL^{h}$. The proof of $h$-\EXPSPACE-hardness
of unrestricted $\QPTL^{h}$ is simpler.\footnote{Note that by the well-known $h$-\EXPSPACE-hardness of satisfiability of $\QPTL$ formulas in prenex form whose alternation depth of existential and universal quantifiers is at most $h$, we immediately deduce $(h-1)$-\EXPSPACE-hardness for satisfiability of unrestricted $\QPTL^{h}$. One can enforce this result by showing that satisfiability of unrestricted $\QPTL^{h}$ is in fact $h$-\EXPSPACE-hard.}
Therefore, in the rest of this section, we show that
satisfiability for the \emph{existential fragment} of
$\QPTL^{h}$ is $(h-1)$-\EXPSPACE-hard even for formulas using temporal modalities in $\{\Next,\PNext,\Eventually,\PEventually,\Always,\PAlways\}$. This is proved by a reduction from the non-halting problem for $\exp[h-1]$-space bounded deterministic Turing Machines, where $\exp[h-1]$ denotes the class of  functions  $f:\N\rightarrow \N$ such that for some  constant  $c\geq 1$,  $f(n)=\Tower(h-1,n^c)$ for each $n\in\N$.

Let $\AP$ be the infinite set of atomic propositions given by
\[
\AP:= \{0,1\}\cup \{\Sym_1,\Sym_2,\ldots\}
\]
Moreover, for each $h\geq 1$, let $\AP_h$ be the finite subset of $\AP$ given by
\[
\AP_h:= \{0,1\}\cup \{\Sym_1,\ldots,\Sym_h\}
\]
First, for all $n\geq 1$ and $h\geq 1$, we define an encoding of the natural numbers in
$[0,\Tower(h,n)-1]$ by finite words over $\AP_h$, called $(h,n)$\emph{-codes}. In particular, for $h>1$, an $(h,n)$-code encoding a natural number
$m\in [0,\Tower(h,n)-1]$ is a sequence of $\Tower(h-1,n)$ $(h-1,n)$-codes, where the $i^{th}$ $(h-1,n)$-code encodes both the value and (recursively) the position of the $i^{th}$-bit in the binary representation of $m$. Formally, the set of $(h,n)$-codes is
defined by induction on $h$ as follows.\vspace{0.2cm}

\noindent \textbf{Base Step: $h= 1$.} A $(1,n)$-code
is a finite word
$w$ over  $\AP_1$ of the form
$w=\Sym_{1}b b_1\ldots b_n\Sym_1$,  where  $b,b_1,\ldots,b_n\in \{0,1\}$.
The \emph{content} of $w$ is the bit $b$ and the \emph{index} of $w$ is the natural number in
$[0,\Tower(1,n)-1]$ (recall that $\Tower(1,n)=2^n$) whose binary code is
$b_1\ldots b_n$ (we assume that $b_1$ is the least significant bit).\vspace{0.2cm}

\noindent \textbf{Induction Step: let $h\geq 1$.} An $(h+1,n)$-code is a word
$w$ over $\AP_{h+1}$ of the form
\[
\Sym_{h+1}b\Sym_h w_1\Sym_h w_2\Sym_h\ldots \Sym_h w_{\Tower(h,n)}\Sym_h\Sym_{h+1}
\]
 where
$b\in \{0,1\}$  and for all $i\in [1,\Tower(h,n)]$,
$\Sym_h w_i \Sym_h$ is an $(h,n)$-code whose index is $i-1$.
Let $b_i$ be the content of the $(h,n)$-code $\Sym_h w_i \Sym_h$.
Then, the \emph{content}  of $w$ is the bit $b$, and the  \emph{index}
of $w$ is the natural number in $[0,\Tower((h+1)-1]$ whose binary code is given by
$b_1\ldots b_{\Tower(h,n)}$.

\bigskip

Given a finite alphabet $\Sigma$ such that $\AP\cap \Sigma=\emptyset$, we also introduce the notion of \emph{$(h,n)$-block over $\Sigma$} which is defined as an
$(h,n)$-code but we require that the content is a symbol in $\Sigma$. The index of an $(h,n)$-block over $\Sigma$ is defined as the index of an $(h,n)$-code.
Intuitively, $(h,n)$-blocks are used to encode the cells of the configurations reachable by $\exp[h]$-space bounded  deterministic Turing machines on inputs of size $n$.

\begin{example} Let $n=2$ and $h=2$. In this case $\Tower(h,n)=16$ and $\Tower(h-1,n)=4$. Thus, we can encode by $(2,2)$-codes all the integers in $[0,15]$. For example, let us consider the number 14 whose binary code
$($using $\Tower(h-1,n)=4$ bits$)$ is given by $0111$ $($assuming that the first bit is the least significant one$)$. The
$(2,2)$-code with content $0$ encoding number 14 is given by
\[
  \Sym_{2}0\Sym_{1}0\,00\Sym_{1}1\,10\Sym_{1}1\,01\Sym_{1}1\,11\Sym_{1}\Sym_{2}
\]
\noindent Note that we encode also the position of each bit in the binary code of  14.
 \end{example}

Let $\Tag$ be an extra infinite set of atomic propositions disjoint from $\AP$ given by
\[
\Tag:= \{\Bl,\First,\Last\}\cup \{\Beg_1,\End_1,\Beg_2,\End_2,\ldots\}
\]
and for each $h\geq 1$, let $\Tag_h$ be the finite subset of $\Tag$ given by
\[
\Tag_h:= \{\Bl,\First,\Last\}\cup \{\Beg_1,\End_1,\ldots,\Beg_h,\End_h\}
\]

Intuitively, we use the propositions in $\Tag_h$ to mark $(h,n)$-blocks.

For all $h\geq  1$, the lower bound for satisfiability of existential $\QPTL^{h}$ is crucially based on the following Propositions~\ref{pro:encodingForLowerBoundQPTL} and~\ref{pro:encodingForLowerBoundQPTLTwo}.
For a set $P$ and a word $w$ over $2^{P'}$ with $P'\supseteq P$, we say that $w$ is \emph{$P$-simple} if for each position $i$ of $w$, $w(i)\cap P$ is a singleton.

\begin{proposition} \label{pro:encodingForLowerBoundQPTL}
For all $n\geq 1$ and $h\geq 1$, one can construct in time polynomial in $n$ and $h$ three \emph{existential}
$\QPTL^{h}$ formulas  $\psi_\Bl(h,n)$,   $\psi_=(h,n)$, and $\psi_\Inc(h,n)$ over $\AP_h\cup \Tag_h$ using only temporal modalities in $\{\Next,\PNext,\Eventually,\PEventually,\Always,\PAlways\}$ such that
for all $\AP_h$-simple pointed words $(w,i)$, the following holds:
\begin{compactitem}
  \item $(w,i)\models \psi_\Bl(h,n)$ $\Leftrightarrow$ there is $j>i$ such that $w[i,j]$ encodes an $(h,n)$-code.
  \item Let $j>i$ such that
  \begin{compactitem}
 \item the projection of $w[i,j]$ over $\AP_h$ is of the form $\Sym_h w_1\Sym_h w' \Sym_h w_2 \Sym_h$, where  $\Sym_h w_1\Sym_h$ and $\Sym_h w_2\Sym_h$ are $(h,n)$-codes, and
  \item the beginning and the end of $\Sym_h w_1\Sym_h$ and $\Sym_h w_2\Sym_h$ are marked by $\Beg_h$ and $\End_h$, respectively, and
  no other position of $w$ is marked by $\Beg_h$ and $\End_h$.
\end{compactitem}
Then, $(w,i)\models \psi_=(h,n)$ $\Leftrightarrow$ $\Sym_h w_1\Sym_h$ and $\Sym_h w_2\Sym_h$ have the same index.
\item Let $j>i$ such that the projection of $w[i,j]$ over $\AP_h$ has the form $\Sym_h w_1\Sym_h  w_2 \Sym_h$ so that $\Sym_h w_1\Sym_h$ and $\Sym_h w_2\Sym_h$ are $(h,n)$-codes. Then,   $(w,i)\models \psi_\Inc(h,n)$ $\Leftrightarrow$ there is $i\in [0,\Tower(h,n)-2]$ such that the index of $\Sym_h w_1\Sym_h$ is $i$ and the index of $\Sym_h w_2\Sym_h$ is $i+1$
\end{compactitem}
Moreover, each existential quantifier in $\psi_=(h,n)$ is in the scope of some temporal modality.
\end{proposition}
\begin{proof}
Fix $n\geq 1$. For each $h\geq 1$, the construction of formulas $\psi_\Bl(n,h)$, $\psi_=(n,h)$, and $\psi_{\Inc}(h,n)$
is given by induction on $h$. Since $n$ is fixed, for clarity of presentation, we write $\psi_\Bl^{h}$, $\psi_=^{h}$, and $\psi_\Inc^{h}$ instead of
$\psi_\Bl(h,n)$, $\psi_=(h,n)$, and $\psi_\Inc(h,n)$, respectively.\vspace{0.2cm}

\noindent \textbf{Base Step: $h=1$}
\[
\psi^1_{\Bl}:=\Sym_1\,\wedge\,\Next^{n+2}\,\Sym_1\,\,\wedge\,\bigwedge_{i=1}^{n+1}\Next^{i}(0\vee 1)
\]
\[
\psi_=^1:=\bigwedge_{i=1}^n\bigvee_{b\in \{0,1\}}\Next^{i+1}(b\,\wedge\,\Eventually (\End_1\,\wedge\, \Eventually (\Beg_1 \wedge \Next^{i+1}\,b)))
\]
\[
\psi^1_{\Inc}:=\Next\bigvee_{i=1}^n\Bigl([\,\bigwedge_{j=1}^{i-1}\Next^{j}\,(1\,\wedge\,\Next^{n+2}\,0)]\,\,\wedge\,\,\Bigl[\Next^{i}\,(0\,\wedge\,
\Next^{n+2}\,1)]\,\wedge\,
\Bigl[\,\bigwedge_{j=i+1}^n\bigvee_{b\in \{0,1\}}\Next^{j}\,(b\,\wedge\,\Next^{n+2}\,b)\Bigr]\Bigr)
\]
\noindent \textbf{Induction Step: let $h\geq 1$.} In order to construct the formulas $\psi_\Bl^{h+1}$, $\psi_=^{h+1}$, and $\psi_\Inc^{h+1}$, for a  proposition $p$, we use the following $\PLTL$ formulas $\theta(1,p)$ and $\theta(2,p)$, which are satisfied by a pointed word $(w,i)$ iff there are at most one position  and two positions, respectively, along $w$ where $p$ holds.
\begin{eqnarray*}
  \theta(1,p) &:=& \PEventually((\neg\PNext \top) \wedge \Always(\,p\, \rightarrow \Next\Always  \neg\,p ))\\
  \theta(2,p) &:=& \PEventually((\neg\PNext \top) \wedge \Always(\,p\, \rightarrow \Next\Always(\,p\, \rightarrow \Next\Always  \neg\,p )))
\end{eqnarray*}
\noindent \textbf{Definition of formula $\psi^{h+1}_{\Bl}$.}
\[
\psi_{\Bl}^{h+1} := \exists \Bl.\,\exists\First.\,\exists\Last.\,\Bigl(\psi_h^{h+1}\wedge \psi_\First^{h+1}\wedge \psi_\Last^{h+1}\wedge \psi_\Succ^{h+1}\Bigr)
\]
 where
\begin{itemize}
  \item $\psi^{h+1}_h$ is an existential $\QPTL^{h+1}$ formula which uses $\psi^{h}_{\Bl}$ and requires that for the given $\AP_{h+1}$-simple pointed word $(w,i)$, there is $j\geq i$ such that the projection of $w[i,j]$ over $\AP_{h+1}$ is of the form $\Sym_{h+1} b \Sym_h  w_1
   \Sym_h\ldots \Sym_h  w_p \Sym_h  \Sym_{h+1}$, where $b\in \{0,1\}$ and for each $i\in [1,p]$, $\Sym_h w_i \Sym_h$ is an $(h,n)$-code;
   we use existential quantification over $\Bl$ to mark exactly the first and the last position of $w[i,j]$ by proposition $\Bl$.
 \item $\psi^{h+1}_{\First}$ and $\psi^{h+1}_{\Last}$ are $\PLTL$ formulas: the first one  requires that the index of the first $(h,n)$-code $\Sym_h w_1 \Sym_h$ of $w[i,j]$ is $0$, and the second one  requires that the index of the last $(h,n)$-code   $\Sym_h w_p \Sym_h$ of $w[i,j]$ is $\Tower(h,n)-1$;
     we use existential quantification over $\First$ and $\Last$ to mark the first and the last position of $\Sym_h w_1 \Sym_h$ by $\First$, and the first and last position of $\Sym_h w_p \Sym_h$ by $\Last$.
 \item $\psi^{h+1}_{\Succ}$ is an existential  $\QPTL^{h+1}$ formula using $\psi^h_{\Inc}$ and requiring that for consecutive $(h,n)$-codes along $w[i,j]$, the index is incremented.
      \end{itemize}
\begin{eqnarray*}
  \psi_{h}^{h+1} &:=& \theta(2,\Bl) \wedge \Bl \wedge \Sym_{h+1} \wedge \Next\Eventually(\Bl\wedge \Sym_{h+1}) \wedge \Next\Always((\Next\Eventually \Bl)\rightarrow\neg \Sym_{h+1})\wedge \Next(0\vee 1) \wedge \\
  & & \Next^{2}\Sym_h \wedge \Next^{3}\neg\Sym_{h+1}\wedge
  \underbrace{\Next\Always((\Sym_h\wedge \Next( \neg \Sym_{h+1}\wedge \Eventually \Bl)) \rightarrow \psi^h_{\Bl})\Bigr)}_{\text{check that the $\Bl$-marked prefix is a sequence of $(h,n)$-codes}}
  \end{eqnarray*}
  \[
  \psi_{\First}^{h+1}:= \left\{
    \begin{array}{ll}
     \Next^{3} \displaystyle{\bigwedge_{i=1}^{n}}\Next^{i}0
      &    \text{ if  } h=1
      \\
   \theta(2,\First)\wedge \Next^{2}\First\wedge \Next^{3}\Eventually(\First\wedge \Sym_h)\,\,\wedge
     & \\
   \Next^{3}\Always ((\Next\Eventually \First) \rightarrow  \neg\Sym_h) \wedge \Next^{3}\Always ((\Sym_{h-1}\wedge \Next^{2}\Eventually\First) \rightarrow  \Next \,0)
      &    \text{ otherwise }
    \end{array}
  \right.
  \]
\begin{eqnarray*}
  \psi_{\Last}^{h+1} &:=& \theta(2,\Last) \wedge \Eventually(\Last \wedge \Sym_h \wedge \Next\Eventually(\Last \wedge \Sym_h \wedge \Next \Bl)) \wedge \\
 & & \left\{
    \begin{array}{ll}
     \Always\Bigl(( (\PNext\PEventually\,\Last)\wedge \Next^{2}\Eventually\,\Last)\longrightarrow\Next \,1\Bigr)
      &    \text{ if  } h=1
      \\
     \Always\Bigl(( (\PNext\PEventually\,\Last)\wedge \Next\Eventually\,\Last)\longrightarrow (\neg\Sym_h\wedge ((\Sym_{h-1}\wedge \neg\Next \Sym_{h})\rightarrow \Next \,1))\Bigr)
      &    \text{ otherwise }
    \end{array}
  \right.
  \end{eqnarray*}
\[
\psi_\Succ^{h+1} := \underbrace{\Always \Bigl((\Sym_h\wedge \Next\Eventually(\,\Last \wedge\Next\Eventually\,\Last)) \, \longrightarrow \, \psi_{\Inc}^{h}\Bigr)}_{\text{check that for consecutive $(h,n)$-codes the index is incremented}}
\]
\noindent \textbf{Definition of formula $\psi^{h+1}_{=}$.}
\begin{eqnarray*}
  \psi_{=}^{h+1} &:=& \Next\Always \Bigl[ \Bigl(\Sym_h\wedge  \Next^{2}\Eventually (\End_{h+1}\wedge \Eventually\,\Beg_{h+1})\Bigr) \,\,\longrightarrow \,\, \exists \Beg_h.\,\exists\End_h.\, \Bigl(\theta(2,\Beg_h)\,\wedge\, \theta(2,\End_h)\, \wedge \,\\ \\
 & &  \underbrace{\{\Beg_h \wedge \Next\Eventually(\End_h \wedge \Sym_h \wedge \PNext\PAlways((\PNext\PEventually \Beg_h) \rightarrow \neg \Sym_h))\}}_{\text{mark the current $(h,n)$-code of the first $(h+1,n)$-code }} \,\,\,\wedge \,\,\\ \\
 & &  \underbrace{\{\Eventually(\Beg_{h+1}\wedge \Eventually(\Beg_h \wedge \Sym_h \wedge \Next\Eventually (\End_h\wedge \Sym_h\wedge \Eventually \, \End_{h+1}) \wedge \Next\Always ((\Next\Eventually\, \End_h) \rightarrow \neg \Sym_h) )
)\}}_{\text{select an $(h,n)$-code of the second $(h+1,n)$-code }} \,\,\,\wedge \,\, \\ \\
& &  \underbrace{\{\psi_{=}^{h}\,\wedge \,\bigvee_{b\in \{0,1\}} (\Next\, b \wedge \Next\Eventually (\Beg_h\wedge \Next \,b))\}}_{\text{check that the two selected $(h,n)$-codes have the same content and index; note that we use $\psi_{=}^{h}$ }}
\end{eqnarray*}
Note that $\psi_{=}^{h+1}$ ensures by using the always modality  that each $(h,n)$-code of the first $(h+1,n)$-code is selected.\vspace{0.2cm}

\noindent \textbf{Definition of formula $\psi_{\Inc}^{h+1}$. }
Let $(w,i)$ be an $\AP_{h+1}$-simple pointed word and $j\geq i$ such that the projection of $w[i,j]$ over $\AP_{h+1}$ is of the
form $\Sym_{h+1}w_1\Sym_{h+1}w_2\Sym_{h+1}$, where
$\Sym_{h+1}w_1\Sym_{h+1}$ and $\Sym_{h+1} w_2\Sym_{h+1}$ are $(h+1,n)$-codes. Then, the requirement that there is $\ell\in [0,\Tower(h+1,n)-2]$ such that
the index of $\Sym_{h+1}w_1\Sym_{h+1}$  is $\ell$ and the index of $\Sym_{h+1} w_2\Sym_{h+1}$ is $\ell+1$
 is equivalent to the following requirement
\begin{compactitem}
  \item there is a $(h,n)$-code $\Bl$  of $\Sym_{h+1}w_{1}\Sym_{h+1}$ such that denoting with $\Bl'$ the $(h,n)$-code of $\Sym_{h+1}w_{2}\Sym_{h+1}$ having the same index as $\Bl$, it holds that: (1) the content of $\Bl$ is $0$ and the content of each $(h,n)$-code  of $\Sym_{h+1}w_{1}\Sym_{h+1}$ that precedes $\Bl$ is 1, (2) the content of $\Bl'$ is $1$ and the content of each $(h,n)$-code of $\Sym_{h+1}w_{2}\Sym_{h+1}$ that precedes $\Bl'$ is 0, and (3) each $(h,n)$-code $\Bl_s$ of $\Sym_{h+1}w_{1}\Sym_{h+1}$ that follows $\Bl$ has the same  content as the $(h,n)$-code  of $\Sym_{h+1}w_{2}\Sym_{h+1}$ having the same index as $\Bl_s$.
\end{compactitem}\vspace{0.2cm}

\noindent  Thus, formula $\psi^{h+1}_{inc}$ uses $\psi^{h}_{=}$ and is defined as follows. Note that we use existential quantification over $\Bl$ to mark the first position of the (guessed) first $(h,n)$-code of the first $(h+1,n)$-code whose content is $0$. Moreover, we use existential quantification over $\First$ and $\Last$ to mark by $\First$, the first and the last position of the first $(h+1,n)$-code, and by $\Last$, the first and the last position of the second $(h+1,n)$-code.
\[
\psi_{\Inc}^{h+1} := \exists \First.\,\exists \Last.\,\exists\Bl.\,\Bigl(\psi_{\Mark}^{h+1}\wedge \psi_{\Chec}^{h+1}\Bigr)
\]
\begin{eqnarray*}
  \psi_{\Mark}^{h+1} &:=& \theta(2,\First) \wedge \theta(2,\Last) \wedge \theta(1,\Bl) \wedge  \\
 & & \underbrace{\First \wedge \Next\Eventually(\First \wedge \Sym_{h+1}) \wedge \Next\Always ((\Next\Eventually\, \First)\rightarrow \neg \Sym_{h+1})}_{\text{mark with $\First$ the beginning and the end of the first $(h+1,n)$-code}} \wedge\\ \\
 & & \underbrace{\Next\Eventually(\First \wedge  \Last  \wedge \Next\Eventually(\,\Last \wedge \Sym_{h+1})
  \wedge \Next\Always ((\Next\Eventually \,\Last) \rightarrow  \neg \Sym_{h+1}))}_{\text{mark with $\Last$ the beginning and the end of the second $(h+1,n)$-code}} \wedge\\ \\
& & \underbrace{\Next\Eventually(\Bl \wedge \Sym_h \wedge \Next^{2}\Eventually\First)}_{\text{mark with $\Bl$ the beginning of some $(h,n)$-code of the first $(h+1,n)$-code}}
\end{eqnarray*}
\begin{eqnarray*}
  \psi_{\Chec}^{h+1} &:=& \Always\Bigl(\,\, (\Sym_h  \wedge \Next^{2}\Eventually\,\First) \,\, \rightarrow \,\,
                                        \exists \Beg_h.\,\exists \End_h.\, \Bigl\{ \theta(2,\Beg_h)\wedge \theta(2,\End_h)\wedge \\
 & & \underbrace{\Beg_h \wedge  \Next\Eventually(\End_h\wedge \Sym_h \wedge \Next\Eventually\First \wedge  \PNext\PAlways ((\PNext\PEventually\, \Beg_h)\rightarrow \neg \Sym_{h}))}_{\text{mark with $\Beg_h$ and $\End_h$ the current $(h,n)$-code $\code_1$  of the first $(h+1,n)$-code}} \\
 & &  \text{\hspace{5cm}} \wedge \\
  & &  \underbrace{\Next\Eventually(\,\First\wedge \Next\Eventually(\Beg_h\wedge \Sym_h \wedge \Next\Eventually \,\Last \wedge \Next\Eventually(\End_h \wedge \Sym_h) \wedge
   \Next\Always ((\Next\Eventually\, \End_h)\rightarrow \neg \Sym_{h})   ))}_{\text{mark with $\Beg_h$ and $\End_h$ some $(h,n)$-code $\code_2$ of the second $(h+1,n)$-code}}  \\
   & &  \text{\hspace{5cm}} \wedge \\
  & & \text{\hspace{2.4cm}}\underbrace{\psi_{=}^{h}}_{\text{check that $\code_1$ and $\code_2$ have the same index}}  \\
  & &  \text{\hspace{5cm}} \wedge \\
    & & \underbrace{(\,\Bl \rightarrow (\Next 0\wedge \Next\Eventually(\Beg_h\wedge \Next 1) ))}_{\text{if $\code_1$ is the $(h,n)$-code marked by $\Bl$, the contents of $\code_1$ and $\code_2$ are $0$ and $1$}}  \\
    & &  \text{\hspace{5cm}} \wedge \\
       & & \text{\hspace{-0.5cm}}\underbrace{(\Next\Eventually\,\Bl \rightarrow (\Next 1\wedge \Next\Eventually(\Beg_h\wedge \Next 0) ))}_{\text{if $\code_1$ precedes the $(h,n)$-code marked by $\Bl$, the contents of $\code_1$ and $\code_2$ are $1$ and $0$}}  \\
    & &  \text{\hspace{5cm}} \wedge \\
   & & \underbrace{(\PNext\PEventually\,\Bl \rightarrow \displaystyle{\bigvee_{b\in\{0,1\}}}(\Next \,b\wedge \Next\Eventually(\Beg_h\wedge \Next \,b) ))\,\,\,\Bigr\}}_{\text{if $\code_1$ follows the $(h,n)$-code marked by $\Bl$, the contents of $\code_1$ and $\code_2$ coincide}}
\end{eqnarray*}
By construction, it easily follows  that the sizes of $\psi^h_{\Bl}$, $\psi^h_{=}$,
 $\psi^h_{\Inc}$ are polynomial in $n$ and $h$, $\psi^h_{\Bl}$, $\psi^h_{=}$, and
 $\psi^h_{\Inc}$ are $\QPTL^{h}$ formulas, and each existential quantifier in $\psi_=^{h}$ is in the scope of some temporal modality. This concludes the proof of Proposition~\ref{pro:encodingForLowerBoundQPTL}.
\end{proof}

By a straightforward adaptation of the proof of Proposition~\ref{pro:encodingForLowerBoundQPTL}, we obtain the following result.

 \begin{proposition} \label{pro:encodingForLowerBoundQPTLTwo}
For all $n\geq 1$, $h\geq 1$, and finite alphabets $\Sigma$, one can construct in time polynomial in $n$, $h$, and $\Sigma$ three \emph{existential}
$\QPTL^{h}$ formulas  $\psi_\Bl(h,n,\Sigma)$,   $\psi_=(h,n,\Sigma)$, and $\psi_\Inc(h,n,\Sigma)$ over $\AP_h\cup \Tag_h\cup \Sigma$ using only temporal modalities in $\{\Next,\PNext,\Eventually,\PEventually,\Always,\PAlways\}$ such that
for all $(\AP_h\cup \Sigma)$-simple pointed words $(w,i)$, the following holds:
\begin{compactitem}
  \item $(w,i)\models \psi_\Bl(h,n,\Sigma)$ $\Leftrightarrow$ there is $j>i$ such that $w[i,j]$ encodes an $(h,n)$-block over $\Sigma$.
  \item Let $j>i$ such that
  \begin{compactitem}
 \item the projection of $w[i,j]$ over $\AP_h\cup \Sigma$ is of the form $\Sym_h w_1\Sym_h w' \Sym_h w_2 \Sym_h$, where  $\Sym_h w_1\Sym_h$ and $\Sym_h w_2\Sym_h$ are $(h,n)$-blocks over $\Sigma$, and
  \item the beginning and the end of $\Sym_h w_1\Sym_h$ and $\Sym_h w_2\Sym_h$ are marked by $\Beg_h$ and $\End_h$, respectively, and
  no other position of $w$ is marked by $\Beg_h$ and $\End_h$.
\end{compactitem}
Then, $(w,i)\models \psi_=(h,n,\Sigma)$ $\Leftrightarrow$ $\Sym_h w_1\Sym_h$ and $\Sym_h w_2\Sym_h$ have the same index.
\item Let $j>i$ such that the projection of $w[i,j]$ over $\AP_h\cup \Sigma$ has the form $\Sym_h w_1\Sym_h  w_2 \Sym_h$ so that $\Sym_h w_1\Sym_h$ and $\Sym_h w_2\Sym_h$ are $(h,n)$-blocks over $\Sigma$. Then,   $(w,i)\models \psi_\Inc(h,n,\Sigma)$ $\Leftrightarrow$ there is $i\in [0,\Tower(h,n)-2]$ such that the index of $\Sym_h w_1\Sym_h$ is $i$ and the index of $\Sym_h w_2\Sym_h$ is $i+1$.
\end{compactitem}
Moreover, each existential quantifier in $\psi_=(h,n,\Sigma)$ is in the scope of some temporal modality.
\end{proposition}

Now, we can establish for each $h\geq 1$, the lower bound for the existential fragment of $\QPTL^{h}$.

\begin{theorem}
For each $h\geq 1$, satisfiability for the \emph{existential fragment} of
$\QPTL^{h}$ is $(h-1)$-\EXPSPACE-hard even for formulas whose temporal modalities are in $\{\Next,\PNext,\Eventually,\PEventually,\Always,\PAlways\}$.
\end{theorem}
\begin{proof}

It is well-known that satisfiability of $\PLTL$ is \PSPACE-complete  even if the unique allowed temporal modalities are in $\{\Next,\PNext,\Eventually,\PEventually,\Always,\PAlways\}$ \cite{Var88}.\footnote{See references for the Appendix.} Since $\QPTL^{1}$ subsumes $\PLTL$, the result
for $h=1$ follows.

Now, we prove the result for  $h+1$ with $h\geq 1$ by a polynomial time reduction from the non-halting problem of
$\exp[h]$-space bounded  deterministic Turing Machines (TM, for short). Fix such a TM
 $\Mac=\tpl{A,Q,q_0,\delta}$ over the input alphabet $A$, and let $c\geq 1$
be a constant such that for each  $\alpha\in A^*$, the
 space needed by $\mathcal{M}$ on input $\alpha$ is bounded by
 $\Tower(h,|\alpha|^c)$. Fix an input $\alpha\in A^*$ and let $n=|\alpha|^c$.
 Note that any reachable configuration of $\Mac$
over $\alpha$ can be seen as a word $\alpha_1\cdot (q,a)\cdot \alpha_2$ in $A^*\cdot(Q\times A)\cdot
A^*$ of length  $\Tower(h,n)$, where $\alpha_1\cdot a\cdot \alpha_2$ denotes  the tape content, $q$ the current state, and the reading head is at position $|\alpha_1|+1$.
 If $\alpha=a_1\ldots a_{r}$ (where $r=|\alpha|$),
then the initial configuration is given by $(q_0,a_1)a_2\ldots
a_r\underbrace{\#\#\ldots\#}_{\Tower(h,n)-r}$, where $\#$ is the blank symbol.   Let $C=u_1\ldots u_{\Tower(h,n)}$ be a TM configuration. For  $1\leq i\leq \Tower(h,n)$, the value $u'_i$ of the $i^{th}$ cell of the  $\Mac$-successor of $C$ is completely determined by the values
         $u_{i-1}$, $u_i$ and $u_{i+1}$ (taking $u_{i+1}$ for $i=\Tower(h,n)$ and
         $u_{i-1}$ for $i=1$ to be some special symbol, say $\bot$). Let
         $\NNext(u_{i-1},u_{i},u_{i+1})$ be our
         expectation for $u'_i$ (this function can be trivially obtained from the transition function  $\delta$ of
         $\Mac$).

Let $\Sigma=A\cup (Q\times A)$. We build in time polynomial in $\mathcal{M}$ and $n$ an existential $\QPTL^{h+1}$ formula
 $\varphi_{\Mac,\alpha}$ over
 $\Sigma\cup \AP_{h+1}\cup \Tag_{h}$ which is satisfiable iff $\Mac$ does not halt on the input $\alpha$.
 Moreover, $\varphi_{\Mac,\alpha}$ uses only temporal modalities in $\{\Next,\PNext,\Eventually,\PEventually,\Always,\PAlways\}$.
  Hence, the theorem follows.

\noindent  A TM configuration $C=u_1\ldots u_{\Tower(h,n)}$ is encoded
 by the word over $\Sigma \cup \AP_{h+1}$ given by
 \[\Sym_{h+1}\Sym_h w_1\Sym_h\ldots
\Sym_h w_{\Tower(h,n)}\Sym_h\Sym_{h+1}\]
where for each $i\in [1,\Tower(n,h)]$,
$\Sym_h w_i\Sym_h$ is an $(h,n)$-block whose content is $u_i$ (the $i^{th}$ symbol of $C$) and whose index is $i-1$.

Then, the formula $\varphi_{\Mac,\alpha}$ uses the existential $\QPTL^{h}$ formulas  $\psi_\Bl(h,n,\Sigma)$, $\psi_=(h,n,\Sigma)$, and
 $\psi_\Inc(h,n,\Sigma)$ of Proposition~\ref{pro:encodingForLowerBoundQPTLTwo}, and is given by
 \[
 \varphi_{\Mac,\alpha}= \Always(\bigvee_{p\in \Sigma\cup \AP_{h+1}}(\,p\,\wedge \,\bigwedge_{p'\in (\Sigma\cup \AP_{h+1})\setminus \{p\}}\,\neg\,p'))\,\wedge\, \varphi_{\Conf}\,\wedge\,\varphi_{\Init}\,\wedge\,\varphi_{\Fair}
\]
 where: (i) the first conjunct checks that the given word is $\Sigma\cup \AP_{h+1}$-simple, (ii) the second conjunct $\varphi_{\Conf}$  checks that the projection of the given word  over $\Sigma\cup \AP_{h+1}$ is an infinite sequence of TM configuration codes, (iii) the third conjunct ensures that the first TM configuration is the initial one,
 and (iv) the last conjunct guarantees that the sequence of  TM configuration codes is faithful to the evolution of $\Mac$. The construction of $\varphi_\Init$ is straightforward.
    Thus, we focus on $\varphi_{\Conf}$ and $\varphi_{\Fair}$, which are existential $\QPTL^{h+1}$ formulas. In the construction, we also use
 the $\PLTL$ formulas $\theta(1,p)$ and $\theta(2,p)$ (for an atomic proposition $p$) in the proof of
 Proposition~\ref{pro:encodingForLowerBoundQPTL}, which are satisfied by a pointed word $(w,i)$ iff there are at most one position  and two positions, respectively, along $w$ where $p$ holds.
    The existential $\QPTL^{h+1}$ formula  $\varphi_{\Conf}$  uses the existential $\QPTL^{h}$ formulas  $\psi_\Bl(h,n,\Sigma)$ and $\psi_\Inc(h,n,\Sigma)$, and is defined as follows. We assume that $h>1$ (the case for $h=1$ is simpler).
\begin{eqnarray*}
  \varphi_{\Conf} &:=&  \Sym_{h+1} \wedge \Always\Eventually \Sym_{h+1} \wedge \Always(\Sym_{h+1}\rightarrow (\Next\Sym_h \wedge \Next^{2}\neg\Sym_{h+1}))
      \wedge \Next\Always(\Sym_{h+1}\rightarrow \PNext\Sym_h) \wedge\\
  & &      \underbrace{\Always ((\Sym_h\wedge \neg \Next \Sym_{h+1}) \rightarrow \psi_\Bl(h,n,\Sigma))}_{\text{for every subword of the form $\Sym_{h+1}w\Sym_{h+1}$, $w$ is a sequence of $(h,n)$-blocks}}\\
   & &  \text{\hspace{6cm}} \wedge \\
  & & \Always \Bigl((\Sym_h\wedge \neg \Next \Sym_{h+1}) \rightarrow \\
  & & \underbrace{(\psi_\Inc(h,n,\Sigma)\vee \exists\,\Last.\,[ \theta(1,\Last)
   \wedge \Next\Eventually(\Last\wedge \Sym_h\wedge \Next \Sym_{h+1}) \wedge \Next\Always((\Next\Eventually\,\Last)\rightarrow \neg\Sym_h)])\Bigr)}_{\text{for consecutive $(h,n)$-blocks the index is incremented}}\\
   & &  \text{\hspace{6cm}} \wedge \\
   & & \Always \Bigl((\Sym_h\wedge \PNext \Sym_{h+1}) \rightarrow \exists\,\First.\,[ \theta(1,\First)\wedge\\
  & & \underbrace{
   \Next\Eventually(\First\wedge \Sym_h) \wedge \Next\Always((\Next\Eventually\,\First)\rightarrow \neg\Sym_h)\wedge
   \Always((\Sym_{h-1}\wedge\Next^{2}\Eventually\,\First)\rightarrow \Next\,0)]\Bigr)}_{\text{the first $(h,n)$-block of a TM configuration code has index $0$}}\\
   & &  \text{\hspace{6cm}} \wedge \\
   & & \Always \Bigl((\Sym_h\wedge \Next \Sym_{h+1}) \rightarrow \exists\,\Last.\,[ \theta(1,\Last)
   \wedge \\
  & & \underbrace{ \PNext\PEventually(\Last\wedge \Sym_h) \wedge \PNext\PAlways((\PNext\PEventually\,\Last)\rightarrow \neg\Sym_h)\wedge
   \PNext\PNext\PAlways((\Sym_{h-1}\wedge  \PEventually\,\Last)\rightarrow  \Next\,1)]\Bigr)}_{\text{the last $(h,n)$-block of a TM configuration code has index $\Tower(h,n)$-1}}
\end{eqnarray*}
 \noindent Finally, we define the formula $\varphi_\Fair$, which uses the existential  $\QPTL^{h}$ formula $\psi_=(h,n,\Sigma)$ of Proposition~\ref{pro:encodingForLowerBoundQPTLTwo}.
  For a  word $w$ encoding a sequence of TM configurations,
 we have to require that for each subword $\Sym_{h+1}w_1\Sym_{h+1}w_2\Sym_{h+1}$, where $\Sym_{h+1}w_1\Sym_{h+1}$ and $\Sym_{h+1}w_2\Sym_{h+1}$ encode two TM configurations $C_1$ and $C_2$, $C_2$ is the TM successor of $C_1$, i.e.,
 for each $(h,n)$-block $\Bl'$ of $\Sym_{h+1}w_2\Sym_{h+1}$, the content $u'$ of $\Bl'$
 satisfies $u'=\NNext_{\Mac}(u_p,u,u_s)$, where $u$ is the content of the $(h,n)$-block $\Bl$ of $\Sym_{h+1}w_1\Sym_{h+1}$
         having the same index as $\Bl'$, and
         $u_p$ (resp., $u_s$) is the content of the $(h,n)$-block  of $\Sym_{h+1}w_1\Sym_{h+1}$, if any,
         that precedes (resp., follows) $\Bl$. Note that $u_p=\bot$ (resp., $u_s=\bot$) iff $\Bl$ is the first (resp., the last) $(h,n)$-block of
  $\Sym_{h+1}w_1\Sym_{h+1}$.\footnote{Since the first configuration is the initial one (this is ensured by $\varphi_{\Init}$), $\varphi_{\Fair}$ also ensures that for each TM configuration code $C$, there is exactly one $(h,n)$-block of $C$ whose content is in $Q\times A$.}
\[
\varphi_{\Fair} := \bigwedge_{u\in\Sigma}\Always\Bigl(\,\, (\Sym_h  \wedge \Next\,u) \,\, \rightarrow \,\,\bigvee_{u_p,u_s\in \Sigma\cup\{\bot\}}\,
                                       \phi_{u_p,u,u_s}\Bigr)
\]
where  $\phi_{u_p,u,u_s}$ uses  $\psi_=(h,n,\Sigma)$ and is defined as follows. Here, we only consider the case where $u_p\neq \bot$ and $u_s\neq \bot$
(the other cases being similar).
\begin{eqnarray*}
  \phi_{u_p,u,u_s} &:=&      \exists \Beg_h.\,\exists \End_h.\, \exists \First.\,\exists \Last\, \Bigl\{   \theta(2,\Beg_h)\wedge \theta(2,\End_h)\wedge \theta(1,\First)\wedge \theta(1,\Last)\wedge \\
 & &  \text{\hspace{2cm}}\underbrace{\Eventually(\Sym_{h+1}\wedge \First) \wedge \Always((\Next\Eventually\First) \rightarrow \neg\Sym_{h+1})}_{\text{mark with $\First$ the end of the current TM configuration}} \\
 & &  \text{\hspace{5.5cm}} \wedge \\
  & &  \text{\hspace{1.5cm}}\underbrace{\Eventually(\First \wedge \Next\Eventually(\Last \wedge \Sym_{h+1} \wedge \PNext\PAlways((\PNext\PEventually\First) \rightarrow \neg\Sym_{h+1})))}_{\text{mark with $\Last$ the end of the next TM configuration}} \\
 & &  \text{\hspace{5.5cm}} \wedge \\
 & & \text{\hspace{0.5cm}}\underbrace{\Beg_h \wedge  \Next\Eventually(\End_h\wedge \Sym_h  \wedge  \PNext\PAlways ((\PNext\PEventually\, \Beg_h)\rightarrow \neg \Sym_{h}))}_{\text{mark with $\Beg_h$ and $\End_h$  the current $(h,n)$-block $\Bl$ of the current TM configuration}} \\
 & &  \text{\hspace{5.5cm}} \wedge \\
  & &  \underbrace{\Next\Eventually(\,\First\wedge \Next\Eventually(\Beg_h\wedge \Sym_h \wedge \Next\Eventually \,\Last \wedge \Next\Eventually(\End_h \wedge \Sym_h) \wedge
   \Next\Always ((\Next\Eventually\, \End_h)\rightarrow \neg \Sym_{h})   ))}_{\text{mark with $\Beg_h$ and $\End_h$ some $(h,n)$-block $\Bl'$ of the next TM configuration}}  \\
   & &  \text{\hspace{5.5cm}} \wedge \\
  & & \text{\hspace{0.5cm}} \underbrace{\psi_{=}(h,n,\Sigma)\wedge \Next\Eventually(\Beg_h\wedge \Next\,\NNext(u_p,u,u_s))}_{\text{check that $\Bl$ and $\Bl'$ have the same index and the content of $\Bl'$ is $\NNext(u_p,u,u_s)$}}  \\
  & &  \text{\hspace{5.5cm}} \wedge \\
    & & \text{\hspace{1.7cm}} \underbrace{\PNext\PEventually(\Sym_h \wedge (\Next \,u_p)  \wedge \Next\Always(\Next\Eventually(\Beg_h\wedge \Next\Eventually \Beg_h) \rightarrow \neg\Sym_h))}_{\text{check that the $(h,n)$-block preceding $\Bl$ has content $u_p$}}  \\
    & &  \text{\hspace{5.5cm}} \wedge \\
       & & \text{\hspace{2.0cm}} \underbrace{ \Eventually(\End_h \wedge (\Next \,u_s) \wedge \Next\Eventually \Beg_h )}_{\text{check that the $(h,n)$-block following $\Bl$ has content $u_s$}}\quad \Bigr\}
\end{eqnarray*}
By Proposition~\ref{pro:encodingForLowerBoundQPTLTwo}, each existential quantifier in $\psi_=(h,n,\Sigma)$ is in the scope of some temporal modality. Hence, by construction, $\varphi_{\Fair}$ is an existential $\QPTL^{h+1}$ formula. This concludes the proof of the theorem.
\end{proof}

\subsection{Proof of Theorem~\ref{theorem:MHCTLModelChecking}}\label{APP:MHCTLModelChecking}

For a $\QPTL$ formula $\varphi$ and $\AP'\subseteq \AP$ with $\AP' =\{p_1,\ldots, p_n\}$, we write $\exists \AP'.\varphi$ to mean
$\exists p_1.\ldots \exists p_n.\, \varphi$. A $\QPTL$ sentence is a $\QPTL$ formula such that each proposition $p$ occurs in the scope of a quantifier binding $p$.\vspace{0.2cm}

\setcounter{aux}{\value{theorem}}
\setcounter{theorem}{\value{theo-MHCTLModelChecking}}

\begin{theorem}  For all $h\geq 1$ and $\MHCTLStar$ sentences $\varphi$ with strong alternation depth at most $h$, model-checking against $\varphi$ is  $h$-\EXPSPACE-complete, and  $(h-1)$-\EXPSPACE-complete in case $\varphi$ is \emph{existential} (even if the allowed temporal modalities are in $\{\Next,\PNext,\Eventually,\PEventually,\Always,\PAlways\}$).
\end{theorem}
\setcounter{theorem}{\value{aux}}
Both the lower bounds and the upper bounds of Theorem~\ref{theorem:MHCTLModelChecking} are based on Theorem~\ref{theorem:QPTLsatisfiability}.
Without loss of generality, we only consider \emph{well-named} $\QPTL$ (resp., $\MHCTLStar$ formulas), i.e.,
$\QPTL$ (resp., $\MHCTLStar$ formulas) where each quantifier introduces a different proposition (resp., path variable). Moreover, note that
Theorem~\ref{theorem:QPTLsatisfiability} holds even if we restrict ourselves to consider $\QPTL$ sentences.

\bigskip

\noindent \textbf{Upper bounds of Theorem~\ref{theorem:MHCTLModelChecking}.}  We show that given a finite Kripke structure $K$ and a  well-named $\MHCTLStar$ sentence $\varphi$,  one can construct in linear time a $\QPTL$ sentence $\varphi'$ such that $\varphi'$ is satisfiable iff $K$ satisfies $\varphi$.
Moreover, $\varphi'$ has the same strong alternation depth as $\varphi$, $\varphi'$ is existential if  $\varphi$ is existential, and
$\varphi'$ uses only temporal modalities in $\{\Next,\PNext,\Eventually,\PEventually,\Always,\PAlways\}$ if the same holds for $\varphi$. Hence, by
Theorem~\ref{theorem:QPTLsatisfiability}, the upper bounds follow. Now, we give the details of the reduction.

Fix a finite Kripke structure $K=\tpl{S,s_0,E,V}$ over $\AP$.  We consider a new finite set $\AP'$ of atomic propositions defined as follows:
\[
\AP':= \bigcup_{x\in \Var}\AP_x \cup S_x \text{ where } \AP_x := \{p_x\mid p\in\AP\} \text{ and }  S_x :=\{s_x\mid s\in S\}
\]
Thus, we associate to each variable $x\in \Var$ and atomic proposition $p\in \AP$, a fresh atomic proposition $p_x$, and to each variable $x\in \Var$ and
state $s$ of $K$, a fresh atomic proposition $s_x$.
For each $x\in \Var$ and initial path $\pi=s_0,s_1,\ldots$ of $K$, we denote by $w(x,\pi)$   the infinite word over
$2^{\AP_x\cup S_x}$, encoding $\pi$, defined as follows: for all $i\geq 0$,
\[
w(x,\pi)(i):=\{(s_i)_x\}\cup \{p_x\mid p\in V(s_i)\}
\]
We encode path assignments $\Pi$  of $K$ (over $\Var$) by infinite words $w(\Pi)$ over $2^{\AP'}$ as follows:
for all $x\in \Var$,  the projection of $w(\Pi)$ over $S_x\cup \AP_x$ is $w(x,\Pi(x))$.

Next, for all $x\in\Var$, we construct in linear time a  $\PLTL$ formula $\theta(x,K)$ over $2^{\AP_x\cup S_x}$
 encoding the initial paths of $K$ as follows:
\begin{eqnarray*}
\theta(x,K)&:=&  \PEventually \Bigl\{ (\neg \PNext \top) \,\wedge\,\\
 & &  (s_0)_x \,\wedge\, \Always \bigwedge_{s\in S} \Bigl(s_x \rightarrow \Bigl[\bigwedge_{p\in V(s)}p_x \wedge
   \bigwedge_{p\in \AP\setminus V(s)}\neg p_x \wedge \bigwedge_{t\in S\setminus\{s\}}\neg t_x \wedge \bigvee_{t\in E(s)}\Next\, t_x \Bigr] \Bigr)\Bigr\}
\end{eqnarray*}
where $E(s)$ denotes the set of successors of $s$ in $K$. By construction, the following holds.\vspace{0.2cm}

\noindent \emph{Claim 1: }
 for all $x\in\Var$ and infinite pointed words $(w,i)$ over $2^{\AP'}$, $(w,i)\models\theta(x,K)$ \emph{iff} there is an initial path $\pi$ of $K$
such that the projection of $w$ over $S_x\cup \AP_x$ is $w(x,\pi)$.\vspace{0.2cm}

Finally, we inductively define a mapping $f$ associating to each pair $(x,\psi)$ consisting of a variable $x\in \Var$ and a \emph{well-named} $\HCTLStar$ formula $\psi$ over $\AP$ and $\Var$ such that there is no quantifier binding $x$ which occurs in $\psi$, a $\QPTL$ formula $f(x,\psi)$ over $2^{\AP'}$ as follows:\footnote{Intuitively, $x$ represents the current quantified path variable.}
\begin{compactitem}
 \item $f(x,\top)=\top$;
  \item $f(x,p[y]) = p_y$ for all $p\in \AP$ and $x,y\in\Var$;
  \item $f(x,\neg \psi)= \neg f(x,\psi)$;
   \item $f(x,\psi_1\wedge \psi_2)= f(x,\psi_1)\wedge f(x,\psi_2)$;
   \item $f(x,\Next \psi)=\Next f(x,\psi)$;
  \item $f(x,\PNext \psi)=\PNext f(x,\psi)$;
   \item  $f(x,\psi_1\,\Until\, \psi_2)= f(x,\psi_1)\,\Until\, f(x,\psi_2)$;
  \item  $f(x,\psi_1\,\PUntil\, \psi_2)= f(x,\psi_1)\,\PUntil\, f(x,\psi_2)$;
  \item  $f(x,\exists^{G} y.\psi)= \exists (\AP_y\cup S_y).\, \Bigl(\theta(y,K)\wedge f(y,\psi)\Bigr)$;
   \item  $f(x,\exists y.\psi)= \exists(\AP_y\cup S_y).\,  \Bigl(\theta(y,K)\wedge f(y,\psi) \wedge \PAlways \displaystyle{\bigwedge_{s\in S}}(s_x \leftrightarrow s_y) \Bigr)$.
\end{compactitem}

By construction, $f(x,\psi)$ has size linear in $\psi$ and has the same strong alternation depth as $\psi$.
Moreover, $f(x,\psi)$ is a $\QPTL$ sentence if $\psi$ is a $\MHCTLStar$ sentence, $f(x,\psi)$ is  existential if  $\psi$ is existential, and
$f(x,\psi)$ uses only temporal modalities in $\{\Next,\PNext,\Eventually,\PEventually,\Always,\PAlways\}$ if the same holds for $\psi$. Hence, by
Theorem~\ref{theorem:QPTLsatisfiability}, the upper bounds of Theorem~\ref{theorem:MHCTLModelChecking} directly follow from the following claim.\vspace{0.2cm}

\noindent \emph{Claim 2:} let $x\in\Var$, $\Pi$ be a path assignment of $K$ and $\psi$ be a well-named $\MHCTLStar$ formula $\psi$ over $\AP$ and $\Var$ such that there is no quantifier binding $x$ which occurs in $\psi$. Then, for all $i\geq 0$:
\[
\Pi,x,i \models_K \psi \Leftrightarrow w(\Pi),i \models f(x,\psi)
\]
\noindent \emph{Proof of Claim~2: }
 Let $x\in\Var$, $\Pi$, and $\psi$ as in the statement of the claim.  The proof is by induction on $|\psi|$. The cases for the boolean connectives and the temporal modalities $\Next$, $\PNext$, $\Until$, and $\PUntil$ easily follow from the induction hypothesis. For the other cases, we proceed as follows:
\begin{itemize}
  \item $\psi = p[y]$ for some $p\in \AP$ and $y\in\Var$: we have that
   $\Pi,x,i \models_K p[y]$ $\Leftrightarrow$ $p\in \Pi(y)(i)$ $\Leftrightarrow$  $p_y\in w(y,\Pi(y))(i)$ $\Leftrightarrow$
   $p_y \in w(\Pi)(i)$ $\Leftrightarrow$ $w(\Pi),i\models f(x,p[y])$. Hence, the result follows.
  \item $\psi= \exists y.\, \psi'$: by hypothesis, $x\neq y$.
For the implication, $\Pi,x,i \models_K \psi \Rightarrow w(\Pi),i \models f(x,\psi)$, assume that
$\Pi,x,i \models_K \psi$. Hence, there exists an initial path $\pi$ of $K$  such that
$\pi[0,i]=\Pi(x)[0,i]$ and $\Pi[y \leftarrow \pi],y,i \models \psi'$. Since $\psi$ is well-named,
there is no quantifier binding $y$ which occurs in $\psi'$. Hence,
by the induction hypothesis,
$w(\Pi[y \leftarrow \pi]),i \models f(y,\psi')$. Moreover, since $x\neq y$, by construction, the projections   of
$w(\Pi[y \leftarrow \pi])$ over $\AP_x\cup S_x$ and $\AP_y\cup S_y$, respectively, are  $w(x,\Pi(x))$ and $w(y,\pi)$. Thus, since
 $\pi[0,i]=\Pi(x)[0,i]$, by Claim~1 in the proof of Theorem~\ref{theorem:MHCTLModelChecking},
 it follows that
 \[
w(\Pi[y \leftarrow \pi]),i \models \theta(y,K)\wedge f(y,\psi') \wedge \PAlways \displaystyle{\bigwedge_{s\in S}}(s_x \leftrightarrow s_y)
 \]
 Since the projections of $w(\Pi[y \leftarrow \pi])$ and $w(\Pi)$ over $\AP'\setminus (S_y\cup \AP_y)$ coincide, we obtain that
 \[
w(\Pi),i \models \exists (\AP_y\cup S_y).\,\Bigl( \theta(y,K)\wedge f(y,\psi') \wedge \PAlways \displaystyle{\bigwedge_{s\in S}}(s_x \leftrightarrow s_y)\Bigr) = f(x,\psi)
 \]
and the result follows.

The converse implication  $w(\Pi),i \models f(x,\psi) \Rightarrow  \Pi,x,i \models_K \psi $ is similar, and we omit the details here. 
 %
 %
\item  $\psi= \exists^{G} y.\, \psi'$: this case is similar to the previous one.
\end{itemize}
This concludes the proof of Claim~2.\qed

\bigskip

\noindent \textbf{Lower bounds of Theorem~\ref{theorem:MHCTLModelChecking}.} We show that given a well-named $\QPTL$ sentence $\varphi$ over $\AP$, one can construct in linear time a finite Kripke structure $K_\AP$ (depending only on $\AP$) and a $\MHCTLStar$ sentence $\varphi'$ such that $\varphi$ is satisfiable iff $K_\AP$ satisfies $\varphi'$.
Moreover, $\varphi'$ has the same strong alternation depth as $\varphi$, $\varphi'$ is existential if  $\varphi$ is existential, and
$\varphi'$ uses only temporal modalities in $\{\Next,\PNext,\Eventually,\PEventually,\Always,\PAlways\}$ if the same holds for $\varphi$. Hence, by
Theorem~\ref{theorem:QPTLsatisfiability}, the result follows. Now, we proceed with the details of the reduction.

Let $\AP'=\AP\cup \{\pTag\}$, where $\pTag$ is a fresh proposition, and fix an ordering $\{p_1,\ldots, p_n\}$ of the propositions in $\AP$.
First, we encode an infinite word $w$ over $2^{\AP}$ by an infinite word $en(w)$ over $2^{\AP'}$ defined as follows: $w=w_0 \cdot w_1 \cdot \ldots$, where
for each $i\geq 0$, $w_i$ (the encoding of the $i^{th}$ symbol of $w$) is the finite word over $2^{\AP'}$ of length $n+1$ given by $\{\pTag\}P_1,\ldots P_n$, where $P_k= \{p_k\}$ if $p_k\in w(i)$, and $P_k=\emptyset$ otherwise (for all $k\in [1,n]$). Then, the finite Kripke structure $K_\AP=\tpl{S,s_0,E,V}$ has size linear in $|\AP|$ and it is constructed in such a way that the set of traces of the initial paths of $K_\AP$   coincides with the set of the encodings $en(w)$ of the infinite words $w$ over $2^{\AP}$. Formally, $K_\AP$ is defined as follows:
\begin{compactitem}
  \item $S=\{p_h,\overline{p}_h\mid h\in\{1,\ldots,n\}\}\cup \{\pTag\}$ and $s_0 = \pTag$;
  \item $E$ consists of the edges $(p_k,p_{k+1})$, $(p_k,\overline{p}_{k+1})$,
  $(\overline{p}_k,p_{k+1})$ and $(\overline{p}_k,\overline{p}_{k+1})$ for all $k\in [1,n-1]$, and the edges
  $(\pTag,p_1)$, $(\pTag,\overline{p}_1)$, $(p_n,\pTag)$, and $(\overline{p}_n,\pTag)$.
  \item $V(\pTag)=\{\pTag\}$ and  $V(p_k)=\{p_k\}$ and $V(\overline{p}_k)=\emptyset$ for all $k\in [1,n]$.
\end{compactitem}

\bigskip

Finally, we inductively define a mapping $g$ associating to each pair $(h,\psi)$ consisting of an index $h\in [1,n]$\footnote{intuitively, $p_h$ represents the current quantified proposition.}
 and a well-named $\QPTL$ formula $\psi$ over $\AP$ such that there is no quantifier in $\psi$ binding proposition $p_h$, a $\MHCTLStar$ formula $g(h,\psi)$ over $\AP'$ and $\Var=\{x_1,\ldots,x_n\}$:
\begin{compactitem}
 \item $g(h,\top)=\top$;
  \item $g(h,p_i) = \Next^{i}\, p_i[x_h]$ for all $p_i\in \AP$;
  \item $g(h,\neg \psi)= \neg g(h,\psi)$;
   \item $g(h,\psi_1\wedge \psi_2)= g(h,\psi_1)\wedge g(h,\psi_2)$;
   \item $g(h,\Next \psi)= \Next^{n+1} g(h,\psi)$;
    \item $g(h,\PNext \psi)= \Next^{-n-1} g(h,\psi)$;
   \item  $g(h,\psi_1\,\Until\, \psi_2)= (\pTag[x_h] \rightarrow g(h,\psi_1))\,\Until\, (g(h,\psi_2)\wedge \pTag[x_h])$;
    \item  $g(h,\psi_1\,\PUntil\, \psi_2)= (\pTag[x_h] \rightarrow g(h,\psi_1))\,\PUntil\, (g(h,\psi_2)\wedge \pTag[x_h])$;
     \item  $g(h,\exists p_k.\psi)= \exists^{G}\,x_k. \,  \Bigl(g(k,\psi)  \wedge \PEventually ((\neg\PNext\top)\wedge \Always \displaystyle{\bigwedge_{j\in [1,n]\setminus \{k\}}}(p_j[x_h] \leftrightarrow p_j[x_k])\Bigr)$.
\end{compactitem}

By construction, $g(h,\psi)$ has size linear in $\psi$ and has the same strong alternation depth as $\psi$.
Moreover, $g(h,\psi)$ is a $\MHCTLStar$ sentence if $\psi$ is a $\QPTL$ sentence, $g(h,\psi)$ is  existential if  $\psi$ is existential, and
$g(h,\psi)$ uses only temporal modalities in $\{\Next,\PNext,\Eventually,\PEventually,\Always,\PAlways\}$ if the same holds for $\psi$. Hence, by
Theorem~\ref{theorem:QPTLsatisfiability}, the lower bounds of Theorem~\ref{theorem:MHCTLModelChecking} directly follow from the following claim,  where
for each $i\geq 0$,   $s(i):= i\cdot(n+1)$. Intuitively, $s(i)$ is the $\pTag$-position associated with the $2^{\AP'}$-encoding of the position $i$ of an infinite word over $2^{\AP}$.\vspace{0.2cm}

\noindent \emph{Claim 3:} Let $\psi$ be a well-named $\QPTL$ formula $\psi$ over $\AP$ and  $h\in [1,n]$
such that there is no quantifier in $\psi$ binding proposition $p_h$. Then, for all pointed words $(w,i)$ over $2^{\AP}$ and assignment maps
 $\Pi$ of $K_\AP$ such that $V(\Pi(x_h)) = en(w)$,
\[
(w,i) \models \psi \Leftrightarrow \Pi,x_h,s(i) \models_{K_\AP} g(h,\psi)
\]
\noindent \emph{Proof of Claim~3:}
 let $\psi$, $h$, $(w,i)$, and $\Pi$ as in the statement of the claim. The proof is by induction on $|\psi|$. The cases for the boolean connectives  easily follow from the induction hypothesis. For the other cases, we proceed as follows:
\begin{itemize}
  \item $\psi = p_j$ for some $p_j\in \AP$: we have that
  $(w,i) \models p_j$  $\Leftrightarrow$ $p_j\in w(i)$ $\Leftrightarrow$  $p_j\in en(w)(s(i)+j)$ $\Leftrightarrow$
  $p_j\in V(\Pi(x_h))(s(i)+j)$ $\Leftrightarrow$ $\Pi,x_h,s(i) \models_{K_\AP} \Next^{j}\, p_j[x_h] $
  $\Leftrightarrow$ $\Pi,x_h,s(i) \models_{K_\AP} g(h,p_j)$.
 Hence, the result follows.
 \item $\psi =\Next \psi'$: we have that
  $(w,i) \models \Next \psi'$  $\Leftrightarrow$ $(w,i+1)\models \psi'$ $\Leftrightarrow$ (by the induction hypothesis)
   $\Pi,x_h,s(i+1) \models_{K_\AP} g(h,\psi')$ $\Leftrightarrow$ (since $s(i+1)= s(i)+n+1$)
   $\Pi,x_h,s(i) \models_{K_\AP} \Next^{n+1} g(h,\psi')$ $\Leftrightarrow$ $\Pi,x_h,s(i) \models_{K_\AP} g(h,\Next\psi')$.
 Hence, the result follows.
 \item $\psi = \PNext \psi'$: similar to the previous case.
 \item $\psi =\psi_1\Until \psi_2$: we have that
  $(w,i) \models \psi_1 \Until \psi_2$  $\Leftrightarrow$
  there is $t\geq i$ such that
 $(w,t)\models \psi_2$ and $(w,r) \models\psi_1$ for all $i\leq r<t$ $\Leftrightarrow$
 (by the induction hypothesis)
   there is $t\geq i$ such that
 $\Pi,x_h,s(t) \models_{K_\AP} g(h,\psi_2)$ and $\Pi,x_h,s(r) \models_{K_\AP} g(h,\psi_1)$ for all $i\leq r<t$ $\Leftrightarrow$
 there is $t'\geq s(i)$ such that $\Pi,x_h,t' \models_{K_\AP} g(h,\psi_2)$ and $\pTag\in V(\Pi(x_h))(t')$, and for all
 $s(i)\leq r'<t'$ such that $\pTag\in V(\Pi(x_h))(r')$, $\Pi,x_h,r' \models_{K_\AP} g(h,\psi_1)$
 $\Leftrightarrow$ $\Pi,x_h,s(i) \models_{K_\AP} g(h,\psi_1\Until\psi_2)$.
  Hence, the result follows.
 \item $\psi =\psi_1\PUntil \psi_2$: similar to the previous case.
  \item $\psi= \exists p_k.\, \psi'$: by hypothesis, $k\neq h$.
For the implication, $(w,i) \models \psi \Rightarrow \Pi,x_h,s(i) \models_{K_{\AP}} g(h,\psi)$, assume that
$(w,i) \models \psi$. Hence, there exists a pointed  word $(w',i)$ such that $w'=_{\AP\setminus\{p_k\}}w$
and $(w',i) \models \psi'$. Let $\pi'$ be the initial path of $K_\AP$  such that $V(\pi')=en(w')$.
Since $\psi$ is well-named, there is no quantifier of $\psi'$ binding proposition $p_k$. Thus, by the induction hypothesis,
\[
\Pi[x_k \leftarrow \pi'], x_k, s(i)\models_{K_\AP} g(k,\psi')
\]
Moreover, since $V(\Pi[x_h])=en(w)$, it holds that for all positions $\ell\geq 0$ and propositions $p_j\in \AP\setminus \{p_k\}$,
$p_j\in V(\Pi(x_h)(\ell))$ iff $p_j\in V(\Pi[x_k \leftarrow \pi'](x_k)(\ell))$. Thus, since $h\neq k$, it holds that
\[
\Pi[x_k \leftarrow \pi'], x_k, s(i)\models_{K_\AP}  \PEventually ((\neg\PNext\top)\wedge \Always \displaystyle{\bigwedge_{j\in [1,n]\setminus \{k\}}}(p_j[x_h] \leftrightarrow p_j[x_k])
\]
By construction, it follows that $\Pi, x_h, s(i)\models_{K_\AP} g(h,\exists p_k.\, \psi')$, and the result follows.

The converse implication  $\Pi, x_h, s(i)\models_{K_\AP} g(h,\psi) \Rightarrow  (w,i) \models \psi $ is similar, and we omit the details here. 
%
%
\end{itemize}\vspace{0.2cm}

This concludes the proof of Claim~3.\qed

\putbib
\end{bibunit}


\begin{thebibliography}{10}

\bibitem{alur2007model}
R.~Alur, P.~{\v{C}}ern{\`y}, and S.~Chaudhuri.
\newblock Model checking on trees with path equivalences.
\newblock In {\em Proc. 13th TACAS}, LNCS 4424, pages 664--678. Springer, 2007.

\bibitem{AlurCZ06}
R.~Alur, P.~Cern{\'y}, and S.~Zdancewic.
\newblock Preserving secrecy under refinement.
\newblock In {\em Proc. 33rd ICALP}, LNCS 4052, pages 107--118. Springer, 2006.

\bibitem{BalliuDG11}
M.~Balliu, M.~Dam, and G.~Le Guernic.
\newblock Epistemic temporal logic for information flow security.
\newblock In {\em Proc. PLAS}, page~6. ACM, 2011.

\bibitem{DBLP:journals/ijisec/BryansKMR08}
J.~Bryans, M.~Koutny, L.~Mazar{\'e}, and P.Y.A. Ryan.
\newblock Opacity generalised to transition systems.
\newblock {\em Int. J. Inf. Sec.}, 7(6):421--435, 2008.

\bibitem{ClarksonFKMRS14}
M.R. Clarkson, B.~Finkbeiner, M.~Koleini, K.K. Micinski, M.N. Rabe, and
  C.~S{\'a}nchez.
\newblock Temporal logics for hyperproperties.
\newblock In {\em Proc. 3rd POST}, LNCS 8414, pages 265--284. Springer, 2014.

\bibitem{ClarksonS10}
M.R. Clarkson and F.B. Schneider.
\newblock Hyperproperties.
\newblock {\em Journal of Computer Security}, 18(6):1157--1210, 2010.

\bibitem{DaxK08}
C.~Dax and F.~Klaedtke.
\newblock Alternation elimination by complementation (extended abstract).
\newblock In {\em Proc. 15th LPAR}, LNCS 5330, pages 214--229. Springer, 2008.

\bibitem{Dima08}
C.~Dima.
\newblock Revisiting satisfiability and model-checking for {CTLK} with
  synchrony and perfect recall.
\newblock In {\em Proc. 9th CLIMA}, LNCS 5405, pages 117--131. Springer, 2008.

\bibitem{DimitrovaFKRS12}
R.~Dimitrova, B.~Finkbeiner, M.~Kov{\'a}cs, M.N. Rabe, and H.~Seidl.
\newblock Model checking information flow in reactive systems.
\newblock In {\em Proc. 13th VMCAI}, LNCS 7148, pages 169--185. Springer, 2012.

\bibitem{EmersonH86}
E.A. Emerson and J.Y. Halpern.
\newblock ``{S}ometimes'' and ``not never'' revisited: on branching versus
  linear time temporal logic.
\newblock {\em Journal of the ACM}, 33(1):151--178, 1986.

\bibitem{fagin1995reasoning}
R.~Fagin, J.Y. Halpern, and M.Y. Vardi.
\newblock {\em Reasoning about knowledge}, volume~4.
\newblock MIT press Cambridge, 1995.

\bibitem{goguen1982security}
J.A. Goguen and J.~Meseguer.
\newblock Security policies and security models.
\newblock In {\em IEEE Symposium on Security and privacy}, volume~12, 1982.

\bibitem{HalpernO08}
J.Y. Halpern and K.R. O'Neill.
\newblock Secrecy in multiagent systems.
\newblock {\em ACM Trans. Inf. Syst. Secur.}, 12(1), 2008.

\bibitem{halpern2004complete}
J.Y. Halpern, R.~van~der Meyden, and M.Y. Vardi.
\newblock {Complete Axiomatizations for Reasoning about Knowledge and Time}.
\newblock {\em SIAM J. Comput.}, 33(3):674--703, 2004.

\bibitem{KupfermanPV12}
O.~Kupferman, A.~Pnueli, and M.Y. Vardi.
\newblock Once and for all.
\newblock {\em J. Comput. Syst. Sci.}, 78(3):981--996, 2012.

\bibitem{KupfermanVW00}
O.~Kupferman, M.Y. Vardi, and P.~Wolper.
\newblock An {A}utomata-{T}heoretic {A}pproach to {B}ranching-{T}ime {M}odel
  {C}hecking.
\newblock {\em Journal of ACM}, 47(2):312--360, 2000.

\bibitem{MilushevC12}
D.~Milushev and D.~Clarke.
\newblock Towards incrementalization of holistic hyperproperties.
\newblock In {\em Proc. 1st POST}, LNCS 7215, pages 329--348. Springer, 2012.

\bibitem{Pnueli77}
A.~Pnueli.
\newblock The temporal logic of programs.
\newblock In {\em Proc. 18th FOCS}, pages 46--57. IEEE Computer Society, 1977.

\bibitem{ShilovG02}
N.V. Shilov and N.O. Garanina.
\newblock Model checking knowledge and fixpoints.
\newblock In {\em Proc. FICS}, BRICS Notes Series, pages 25--39, 2002.

\bibitem{SistlaVW87}
A.P. Sistla, M.Y. Vardi, and P.~Wolper.
\newblock The complementation problem for {B}{\"u}chi automata with
  appplications to temporal logic.
\newblock {\em Theoretical Computer Science}, 49:217--237, 1987.

\bibitem{MeydenS99}
R.~van~der Meyden and N.V. Shilov.
\newblock Model checking knowledge and time in systems with perfect recall
  (extended abstract).
\newblock In {\em Proc. 19th FSTTCS}, LNCS 1738, pages 432--445. Springer,
  1999.

\bibitem{Zielonka98}
W.~Zielonka.
\newblock Infinite games on finitely coloured graphs with applications to
  automata on infinite trees.
\newblock {\em Theoretical Computer Science}, 200(1-2):135--183, 1998.

\end{thebibliography}


\begin{thebibliography}{1}

\bibitem{KupfermanV01}
O.~Kupferman and M.Y. Vardi.
\newblock Weak alternating automata are not that weak.
\newblock {\em ACM Transactions on Computational Logic}, 2(3):408--429, 2001.

\bibitem{MiyanoH84}
S.~Miyano and T.~Hayashi.
\newblock Alternating finite automata on $\omega$-words.
\newblock {\em Theoretical Computer Science}, 32:321--330, 1984.

\bibitem{Var88}
M.Y. Vardi.
\newblock A temporal fixpoint calculus.
\newblock In {\em Proc. 15th POPL}, pages 250--259. ACM, 1988.

\bibitem{Zielonka98}
W.~Zielonka.
\newblock Infinite games on finitely coloured graphs with applications to
  automata on infinite trees.
\newblock {\em Theoretical Computer Science}, 200(1-2):135--183, 1998.

\end{thebibliography}
\end{document}